\documentclass{article}

\usepackage{iclr2027_conference,times}
\usepackage{amsmath,amssymb,amsthm,mathtools}
\usepackage{booktabs}
\usepackage{xcolor}
\usepackage{colortbl}
\usepackage{hhline}
\usepackage{tabularx}
\usepackage{multirow}
\usepackage{float}
\usepackage{algorithm}
\usepackage[noend]{algpseudocode}
\usepackage{enumitem}
\usepackage{needspace}
\usepackage{microtype}
\usepackage{graphicx}
\usepackage{hyperref}
\usepackage{url}
\newcommand{\E}{\mathbb{E}}
\newcommand{\R}{\mathbb{R}}

\newcommand{\PaperAuthors}{%
  Yuxing Peng\textsuperscript{1}\quad
  Zhiqing Tang\textsuperscript{1,*}\quad
  Weijia Jia\textsuperscript{1,2}}
\newcommand{\PaperAffiliations}{%
  \textsuperscript{1}Institute of Artificial Intelligence and Future Networks,
  Faculty of Arts and Sciences,\\
  Beijing Normal University, Zhuhai 519087, China\\
  \textsuperscript{2}Guangdong Key Lab of AI and Multi-Modal Data Processing,\\
  BNU-HKBU United International College, Zhuhai 519087, China\\
  \textsuperscript{*}Corresponding author}
\newcommand{\PaperPdfAuthors}{Yuxing Peng; Zhiqing Tang; Weijia Jia}

\newcommand{\PdfAuthorMetadata}{\PaperPdfAuthors}

\definecolor{DeepLinkBlue}{RGB}{0,72,144}
\hypersetup{
  colorlinks=true,
  linkcolor=DeepLinkBlue,
  citecolor=DeepLinkBlue,
  urlcolor=DeepLinkBlue,
  pdftitle={Dense Weak Hiding: Closing Complexity Gaps in Nonconvex and PL Finite-Sum Optimization under Individual Smoothness},
  pdfauthor={\PdfAuthorMetadata},
  pdfsubject={Tight IFO complexity for nonconvex stationarity and PL finite-sum optimization},
  pdfkeywords={finite-sum optimization, oracle complexity, lower bounds, individual smoothness, Polyak-Lojasiewicz, restarted PAGE}
}

\newtheorem{theorem}{Theorem}[section]
\newtheorem{lemma}[theorem]{Lemma}
\newtheorem{proposition}[theorem]{Proposition}
\newtheorem{corollary}[theorem]{Corollary}
\theoremstyle{definition}

\newcommand{\Prob}{\mathbb{P}}

\newcommand{\cZ}{\mathcal{Z}}
\newcommand{\cEfuture}{\mathcal{E}_{\mathrm{future}}}
\newcommand{\ind}{\mathbf{1}}
\newcommand{\op}{\mathrm{op}}

\newcommand{\proofguidespace}{\Needspace{10\baselineskip}}

\title{Dense Weak Hiding: Closing Complexity Gaps\\
in Nonconvex and PL Finite-Sum Optimization\\
under Individual Smoothness}

\author{%
  \PaperAuthors\\[5pt]
  \normalfont\small \PaperAffiliations
}
\iclrfinalcopy

\begin{document}
\raggedbottom
\maketitle
\fancyhead{}
\renewcommand{\headrulewidth}{0pt}

\begin{abstract}
Under individual smoothness, the optimal incremental first-order oracle (IFO)
complexity of nonconvex finite-sum optimization is still open.  Known
algorithms use $O(n+\sqrt n\,\Delta L_{\max}/\varepsilon^2)$ calls, but the
best lower bounds miss a factor of $\sqrt n$ on the
$\Delta L_{\max}/\varepsilon^2$ term.  We prove the matching lower bound
$\Omega(n+\sqrt n\,\Delta L_{\max}/\varepsilon^2)$.
Thus PAGE and SPIDER are minimax optimal up to universal constants, under both
individual and mean-squared smoothness.
Unlike earlier bounds for linear-span algorithms, our result applies to the
broader class of randomized IFO algorithms, allowing arbitrary choices of
component indices and query points based on the complete preceding history.

Under the global Polyak--\L{}ojasiewicz (PL) condition, we use a similar idea
to obtain an $\Omega(n+\kappa_{\max}\sqrt n\log(\Delta/\varepsilon))$ lower bound for
large $\kappa_{\max}$.  We also note that existing PL lower bounds have
focused only on the regime where $\kappa$ is relatively large.  To fill this
gap, we study the regime
$\kappa_{\max}<\sqrt n$ and obtain a new lower-bound rate,
$\Omega\!\left(n+n\log(\Delta/\varepsilon)/(1+\log(\sqrt n/\kappa_{\max}))\right)$.
The new lower bound motivates us to propose Restarted PAGE.  Its upper bound
matches the new rate for small $\kappa_{\max}$ and recovers the standard PAGE
rate for large $\kappa_{\max}$, showing that our new lower bound is nearly tight.

Our new lower bounds are based on our proposed \emph{dense weak hiding}
construction.  By spreading each hidden direction across all components, the
construction makes every single IFO query weakly informative while preserving
the direction in the full row average.  A sufficiently large dimension makes
premature future-stage activation unlikely within the prescribed query budget.
Combining this control with the stagewise information bound gives the revelation
cost.  We choose the bias
to balance this revelation cost against the number of stages permitted by
individual smoothness, and this balance produces the missing $\sqrt n$ factor.
\end{abstract}

\section{Introduction}

Finite-sum methods optimize
$F(x)=n^{-1}\sum_{i=1}^n f_i(x)$ by querying a fixed set of component
functions.  Unlike an online stochastic oracle, the same components remain
available for repeated queries.  Unlike full-gradient optimization, one call
accesses only one component.  The algorithm may choose its next component and
query point from all preceding observations.  Modern
variance-reduced methods exploit this access pattern
\citep{defazio2014saga,fang2018spider,li2021page}.  To determine whether these
methods are optimal, a lower bound must use the same finite-sum oracle model
and smoothness assumption.

For nonconvex stationarity, PAGE and SPIDER use
$O(n+\sqrt n\,\Delta L_{\rm ms}/\varepsilon^2)$ IFO calls under mean-squared
smoothness \citep{fang2018spider,li2021page}.  Individual
$L_{\max}$-smoothness implies the same condition with
$L_{\rm ms}\le L_{\max}$, so this upper bound is
$O(n+\sqrt n\,\Delta L_{\max}/\varepsilon^2)$ under individual smoothness.
The known lower bounds do not match it.  \citet{zhou2019lower} and
\citet{han2024lower} cover linear-span algorithm classes and give only
$\Omega(n+\Delta L_{\max}/\varepsilon^2)$
under individual smoothness.  \citet{emmenegger2022oracle} allow unrestricted
queries but give only
$\Omega(\Delta L_{\max}/\varepsilon^2)$.  This leaves the first question.
\emph{Is the factor $\sqrt n$ unavoidable for randomized IFO algorithms that
may choose component indices and query points using their full preceding
history?}

Closing this gap requires more than a rescaled version of an existing hard
instance.  The disjoint-block construction of \citet{zhou2019lower} preserves
sequential progress, but its mean-squared-smooth calibration has component
smoothness of order $\sqrt n\,L_{\rm ms}$.
Like that of \citet{han2024lower}, its linear-span argument does not cover queries
outside the generated span.  \citet{emmenegger2022oracle} assign a separate
block to each component.  This allows arbitrary query points but weakens the
average gradient by $n^{-1/2}$.  We therefore seek a construction that
distributes hidden information across components, preserves the average-gradient
signal without this $n^{-1/2}$ loss, and controls premature future-stage
activation under unrestricted queries.

We introduce \emph{dense weak hiding} to meet these requirements.  At each
stage, a hidden direction is encoded in a fixed $n\times D$ sign table.  Each
row is only weakly biased toward the direction, while the exact row average is
$\rho\Theta^{(j)}$.  A bounded radial map turns any query point into an
alignment vector in the unit ball.  A smooth gate satisfies $G=G'=0$ below
its threshold, so unopened links drop out of both the function value and the
gradient.  Revealing one stage requires
$\Omega(\min\{n,\rho^{-2}\})$ distinct rows.  The construction satisfies the
gap and individual smoothness constraints with
$\Theta(\Delta L_{\max}\rho/\varepsilon^2)$ stages.  Their product is
\[
 \Omega\!\left(
 \frac{\Delta L_{\max}}{\varepsilon^2}
 \min\{n\rho,\rho^{-1}\}
 \right),
\]
which is maximized at $\rho\asymp n^{-1/2}$.  This gives the missing
$\sqrt n$ factor and shows that PAGE and SPIDER are minimax optimal for the
randomized IFO class considered here.

The global Polyak--\L{}ojasiewicz (PL) condition poses a second minimax
question.  Standard PAGE gives
$O((n+\kappa_{\rm ms}\sqrt n)\log(\Delta/\varepsilon))$
\citep[Corollary~6]{li2021page}, but two gaps remained.  Under individual
smoothness, prior PL IFO lower bounds did not capture the correct dependence on
$n$, and the regime of small $\kappa$ remained open for randomized
IFO algorithms.  \citet{yue2023lower} gave the first PL lower
bound.  It treats deterministic full-gradient methods and assumes a large
objective condition number.  \citet[Theorem~3.4]{bai2024complexity} obtain the
$\kappa_{\rm ms}\sqrt n$ term for linear-span IFO algorithms at large condition
numbers.  Their Theorem~3.5 gives only $\Omega(n)$ for all condition numbers.  This raises the second question.  \emph{What
is the minimax complexity under individual smoothness as $\kappa_{\max}$ moves
across $\sqrt n$, and is the standard PAGE guarantee under mean-squared
smoothness tight below this transition?}

We determine the complexity in both $\kappa_{\max}$ regimes.  For
$3\le\kappa_{\max}\le\sqrt n$, a two-scale construction gives
$\Omega(n+n\log(\Delta/\varepsilon)/
[1+\log(\sqrt n/\kappa_{\max})])$.  Its average objective is globally strongly
convex.  For $\kappa_{\max}\ge\sqrt n$, a geometrically weighted construction
gives
$\Omega(n+\kappa_{\max}\sqrt n\log(\Delta/\varepsilon))$.
These lower bounds hold under individual smoothness for randomized IFO
algorithms and hence also give the mean-squared lower bounds.  When
$\kappa_{\rm ms}\le\sqrt n$, PAGE epochs cost $O(n)$ calls in expectation and
bound the conditional expected gap by
$q_{\rm ep}=\Theta(\kappa_{\rm ms}/\sqrt n)$ times the starting gap.
Restarted PAGE uses these epochs as restart periods.  It gives the same
logarithmic dependence and matches the lower bounds in both $\kappa_{\rm ms}$
regimes under mean-squared smoothness.

\paragraph{Contributions.}
Our main contributions are summarized as follows.
\begin{enumerate}[leftmargin=*,label=\textbf{\arabic*.},itemsep=2pt,topsep=3pt]
\item \textbf{Dense weak hiding.}
We give a construction for fixed finite sums.  Its information analysis covers
randomized IFO algorithms with unrestricted query selection.

\item \textbf{Tight nonconvex complexity under individual smoothness.}
We prove the matching
$\Omega(n+\sqrt n\,\Delta L_{\max}/\varepsilon^2)$ lower bound.

\item \textbf{PL lower bounds in both $\kappa_{\max}$ regimes.}
We prove the lower bounds under individual smoothness for every
$\kappa_{\max}\ge3$.  The average objective of the small-$\kappa_{\max}$ hard
instance is globally strongly convex.

\item \textbf{Restarted PAGE matching the lower bound.}
Under mean-squared smoothness, Restarted PAGE improves the standard guarantee
for $\kappa_{\rm ms}<\sqrt n$ and matches the lower bound in both
$\kappa_{\rm ms}$ regimes.
\end{enumerate}

\begingroup
\hfuzz=2pt
\begin{table}[t]
\caption{Finite-sum bounds; lower bounds hold with constant success
probability.}
\label{tab:complexity}
\centering
\footnotesize
\setlength{\tabcolsep}{2.5pt}
\renewcommand{\arraystretch}{1.10}

\begin{tabular}{|>{\centering\arraybackslash}m{0.105\linewidth}|>{\centering\arraybackslash}m{0.385\linewidth}|>{\centering\arraybackslash}m{0.205\linewidth}|>{\centering\arraybackslash}m{0.19\linewidth}|}
\hline
\multicolumn{4}{|l|}{\cellcolor{black!5}\textbf{Nonconvex stationarity}} \\
\hline
\rowcolor{black!2}
\textbf{Result} & \textbf{Oracle complexity}
& \textbf{Algorithm class} & \textbf{Reference} \\
\hline
Upper
& $\displaystyle O\!\left(n+\sqrt n\,\Delta L_{\rm ms}/\varepsilon^2\right)$
& PAGE / SPIDER
& \citealp{fang2018spider,li2021page} \\
\hline
Lower
& $\displaystyle \Omega\!\left(\Delta L_{\max}/\varepsilon^2\right)$
& \shortstack{randomized methods;\\arbitrary queries}
& \citealp{emmenegger2022oracle} \\
\hline
Lower
& $\displaystyle \Omega\!\left(n+\Delta L_{\max}/\varepsilon^2\right)$
& \shortstack{linear-span\\IFO / PIFO}
& \citealp{zhou2019lower,han2024lower} \\
\hline
\rowcolor{black!5}
\shortstack{\textbf{New}\\[-2pt]{\scriptsize\bfseries lower}}
& $\displaystyle \boldsymbol{\Omega\!\left(n+\sqrt n\,\Delta L_{\max}/\varepsilon^2\right)}$
& \shortstack{\textbf{randomized}\\\textbf{IFO algorithms}}
& \textbf{Theorem~\ref{thm:nonconvex-main}} \\
\hline
\end{tabular}

\vspace{3pt}

\begin{tabular}{|>{\centering\arraybackslash}m{0.075\linewidth}|>{\centering\arraybackslash}m{0.518\linewidth}|>{\centering\arraybackslash}m{0.16\linewidth}|>{\centering\arraybackslash}m{0.117\linewidth}|}
\hline
\multicolumn{4}{|l|}{\cellcolor{black!5}\textbf{Global PL condition}} \\
\hline
\rowcolor{black!2}
\textbf{Result} & \textbf{Oracle complexity and range}
& \textbf{Algorithm class} & \textbf{Reference} \\
\hline
Upper
& $\displaystyle O\!\left((n+\kappa_{\rm ms}\sqrt n)
  \log\frac{\Delta}{\varepsilon}\right),\quad \kappa_{\rm ms}\ge1$
& standard PAGE
& \citealp[Cor.~6]{li2021page} \\
\hline
Lower
& $\displaystyle \Omega\!\left(\kappa_F\log\frac{\Delta}{\varepsilon}\right),
  \quad \kappa_F\gtrsim10^6$
& \shortstack{deterministic\\full-gradient}
& \citealp{yue2023lower} \\
\hline
Lower
& $\displaystyle \Omega\!\left(n+\kappa_{\rm ms}\sqrt n
  \log\frac{\Delta}{\varepsilon}\right),\quad
  \kappa_{\rm ms}\gtrsim10^5\sqrt n$
& \shortstack{linear-span\\IFO algorithms}
& \citealp{bai2024complexity} \\
\hline
\shortstack{\textbf{New}\\[-2pt]{\scriptsize\bfseries upper}}
& $\displaystyle
  \mkern5mu\left\{\begin{array}{@{}l@{\;}l@{}}
  \boldsymbol{O\!\left(
    n+\dfrac{n\log(\Delta/\varepsilon)}{1+\log(\sqrt n/\kappa_{\rm ms})}
  \right)}, & 1\le\kappa_{\rm ms}\le\sqrt n,\\[5pt]
  \boldsymbol{O\!\left(
    n+\kappa_{\rm ms}\sqrt n\log\dfrac{\Delta}{\varepsilon}
  \right)}, & \kappa_{\rm ms}\ge\sqrt n.
  \end{array}\right.\mkern-5mu$
& \textbf{Restarted PAGE}
& \shortstack{\textbf{Alg.~\ref{alg:restarted-page}};\\
  \textbf{Prop.~\ref{prop:restarted-page}}} \\
\hline
\shortstack{\textbf{New}\\[-2pt]{\scriptsize\bfseries lower}}
& $\displaystyle
  \mkern5mu\left\{\begin{array}{@{}l@{\;}l@{}}
  \boldsymbol{\Omega\!\left(
    n+\dfrac{n\log(\Delta/\varepsilon)}{1+\log(\sqrt n/\kappa_{\max})}
  \right)}, & 3\le\kappa_{\max}\le\sqrt n,\\[5pt]
  \boldsymbol{\Omega\!\left(
    n+\kappa_{\max}\sqrt n\log\dfrac{\Delta}{\varepsilon}
  \right)}, & \kappa_{\max}\ge\sqrt n.
  \end{array}\right.\mkern-5mu$
& \shortstack{\textbf{randomized}\\\textbf{IFO algorithms}}
& \textbf{Theorem~\ref{thm:pl-main}} \\
\hline
\end{tabular}
\end{table}
\endgroup

\section{Related Work}
\label{sec:related}

\paragraph{Nonconvex stationarity.}
Early work established oracle lower bounds for finding stationary points
of smooth nonconvex full objectives
\citep{carmon2020lower1,carmon2021lower2}.
\citet{arjevani2023stochastic} extend zero-chain techniques to arbitrary
randomized stochastic first-order methods.  They prove sharp lower bounds under
bounded variance and mean-squared smoothness, with fresh oracle randomness at
each call.  Concurrent work by \citet{jin2026boundednoise} also uses weakly
biased fresh signs under bounded noise.  For finite-sum upper bounds, SPIDER
and PAGE achieve
$O(n+\sqrt n\,\Delta L_{\rm ms}/\varepsilon^2)$ under mean-squared smoothness
\citep{fang2018spider,li2021page}.  Under individual smoothness, prior
finite-sum lower bounds do not match this rate.  \citet{zhou2019lower},
\citet{xie2019framework}, and \citet{han2024lower} track linear-span progress.
Of these, \citet{zhou2019lower} and \citet{han2024lower}
obtain only $\Omega(n+\Delta L_{\max}/\varepsilon^2)$.  The construction of
\citet{emmenegger2022oracle} allows arbitrary queries but gives only
$\Omega(\Delta L_{\max}/\varepsilon^2)$.  All three bounds therefore miss the
factor $\sqrt n$ multiplying $\Delta L_{\max}/\varepsilon^2$.
Table~\ref{tab:complexity}
also records which algorithms each bound covers.  Our lower bound recovers this
factor without a linear-span restriction.

\paragraph{PL lower bounds.}
At large condition numbers, \citet{yue2023lower} prove a lower bound that
depends on the condition number for deterministic full gradient methods,
while \citet[Theorem~3.4]{bai2024complexity} obtain the
$\kappa_{\rm ms}\sqrt n$ term for linear-span IFO algorithms.  Their Theorem~3.5 gives only $\Omega(n)$ for all
condition numbers.  Two gaps therefore remained.  Under individual smoothness,
previous IFO lower bounds for PL finite sums did not capture the correct dependence
on $n$, and the regime of small $\kappa$ remained open for randomized
IFO algorithms.  Our lower bound covers every
$\kappa_{\max}\ge3$ for randomized IFO algorithms, and Restarted PAGE gives
matching upper bounds under mean-squared smoothness.  The exact
parameter ranges in prior work are recorded in
Appendix~\ref{app:smoothness-calibration}.

\paragraph{Relation to strongly convex finite-sum lower bounds.}
The average objective in our small-$\kappa$ construction is
globally strongly convex, although its components need not be convex.
\citet[Theorem~66]{han2024lower} obtain the same small-$\kappa$ denominator
$1+(\log(\sqrt n/\kappa))_+$ for a strongly convex objective under average
smoothness in their linear-span PIFO class.  Our result establishes this
dependence under individual smoothness for randomized IFO algorithms with
unrestricted query selection.  With convex components, classical strongly
convex finite-sum lower bounds exhibit a $\sqrt{n\kappa_{\max}}$ term at large
condition numbers \citep{agarwal2015lower,woodworth2016tight,lan2018optimal}.
Katyusha attains the corresponding accelerated upper bound
\citep{allenzhu2017katyusha}.

\paragraph{Proof techniques.}
Standard tools for lower bounds include zero-chains, resisting oracles, random
rotations, and arguments that track how component queries reveal new coordinates
\citep{nemirovski1983problem,nesterov2018lectures,agarwal2015lower,woodworth2016tight,carmon2020lower1,carmon2021lower2}.
The closest related construction is that of \citet{zhou2026sharp}, which studies
deterministic first-order algorithms for full objectives with higher-order
smoothness.  That work lifts a scalar hard function from
\citet{carmon2020lower1,carmon2021lower2} to a block chain and transfers
zero-respecting hardness by an orthogonal rotation.  Here the oracle accesses a
fixed finite sum, and the algorithm may randomize both its component choices and
query points.  An information argument therefore replaces the span condition.
Fixed encoding tables hide the directions, and a mutual information bound for
rows drawn without replacement controls what the replies reveal.

\section{Setting and Main Results}
\label{sec:setting}

For $d\ge1$, write $F\in\mathcal F_{\rm ind}^{n,d}(\Delta,L_{\max})$ when a
specified finite-sum decomposition satisfies
\begin{equation}
 F:\R^d\to\R,
 \qquad F=\frac1n\sum_{i=1}^nf_i,
 \qquad F(0)-\inf_xF(x)\le\Delta,
 \label{eq:function-class-gap}
\end{equation}
and every component satisfies
\begin{equation}
 \|\nabla f_i(x)-\nabla f_i(y)\|
 \le L_{\max}\|x-y\|
 \qquad\forall x,y.
 \label{eq:individual-smoothness}
\end{equation}
Every hard component constructed in this paper is $C^\infty$.
The ambient dimension is arbitrary and
finite, and the dimension of the hard instance may depend on the problem parameters
and oracle budget.  When $d$ is omitted from a class notation, all finite
dimensions are allowed.  Appendix~\ref{app:formal-ifo-model} gives the formal
definition.

We also write $F\in\mathcal F_{\rm ms}^{n,d}(\Delta,L_{\rm ms})$ when
the same finite-sum and gap conditions hold and
\begin{equation}
 \left(\frac1n\sum_{i=1}^n
 \|\nabla f_i(x)-\nabla f_i(y)\|^2\right)^{1/2}
 \le L_{\rm ms}\|x-y\|
 \qquad\forall x,y.
 \label{eq:mean-squared-smoothness}
\end{equation}

At the same numerical bound $L$,
$\mathcal F_{\rm ind}^n(\Delta,L)\subseteq
\mathcal F_{\rm ms}^n(\Delta,L)$.
Following the literature on finite-sum lower bounds
\citep{zhou2019lower,emmenegger2022oracle,han2024lower}, we call
\eqref{eq:individual-smoothness} \emph{individual smoothness} (also called componentwise smoothness).
Condition~\eqref{eq:mean-squared-smoothness} is
\emph{mean-squared smoothness} (often called average smoothness), as used in
the SPIDER/PAGE analyses \citep{fang2018spider,li2021page}.

\paragraph{Incremental first-order oracle.}
We use the standard IFO model introduced by \citet{agarwal2015lower} and later
studied for randomized incremental gradient methods by
\citet{lan2018optimal}.  Fix a dimension $d$.  An exact IFO call selects
$(i,x)\in[n]\times\R^d$ and returns
\begin{equation}
    \mathcal O_i(x)=\bigl(f_i(x),\nabla f_i(x)\bigr).
    \label{eq:ifo-response}
\end{equation}
Throughout the paper, \emph{IFO} refers to this exact oracle returning both
value and gradient.

\paragraph{Algorithmic scope.}
We consider randomized IFO algorithms with unrestricted query selection.  At
every round the algorithm may choose both the component index and the query
point as arbitrary functions of its private randomness and the complete
history of queries and responses.  It may repeat indices, stop at a random
time, and return a point that it never queried.  We impose no linear-span or
zero-respecting restriction, and the rule for choosing component indices may
depend on the full history.
This formulation follows standard conventions in oracle complexity
\citep{carmon2020lower1,arjevani2023stochastic}.
Appendix~\ref{app:formal-ifo-model} gives the measurable model and the
minimax definition over all dimensions used below.

For a class $\mathcal F$ and error criterion $\mathcal E$, write
$\operatorname{Comp}_{\varepsilon,\delta}^{\rm IFO}
(\mathcal F;\mathcal E)$ for the smallest fixed budget that achieves
error at most $\varepsilon$ with probability at least $1-\delta$, uniformly
over dimensions and instances in the worst case.  We use
$\mathcal E_F^{\rm nc}(x)=\|\nabla F(x)\|$ and
$\mathcal E_F^{\rm PL}(x)=F(x)-F^*$.

\begin{theorem}[Nonconvex lower bound for randomized IFO algorithms]
\label{thm:nonconvex-main}
There are universal constants $n_0,c_0,p_0>0$ such that, for every
$n\ge n_0$, all $L_{\max},\Delta>0$, every $\varepsilon$ satisfying
$0<\varepsilon^2\le c_0\Delta L_{\max}$, and every fixed
$\delta\in(0,p_0)$,
\begin{equation}
 \operatorname{Comp}_{\varepsilon,\delta}^{\rm IFO}
 \bigl(\mathcal F_{\rm ind}^n(\Delta,L_{\max});\mathcal E^{\rm nc}\bigr)
 =\Omega\!\left(n+\sqrt n\,
 \frac{\Delta L_{\max}}{\varepsilon^2}\right).
 \label{eq:nonconvex-main-rate}
\end{equation}
The implicit constant is universal.
\end{theorem}

Appendix~\ref{app:stopping} extends the same rate to expected call budgets in
the worst case and to expected gradient norm.

To state both smoothness models compactly, let
$\mathsf s\in\{\mathrm{ind},\mathrm{ms}\}$ indicate the model and let
$L_{\mathsf s}$ denote the corresponding smoothness bound.

\begin{corollary}[Tight IFO complexity for nonconvex stationarity]
\label{cor:matching-complexity}
For every $n\ge n_0$, $L_{\mathsf s},\Delta>0$, any fixed
$\delta\in(0,p_0)$, and
$\mathsf s\in\{\mathrm{ind},\mathrm{ms}\}$, the minimax IFO complexity is
\begin{equation}
 \operatorname{Comp}_{\varepsilon,\delta}^{\rm IFO}
 \bigl(\mathcal F_{\mathsf s}^n(\Delta,L_{\mathsf s});\mathcal E^{\rm nc}\bigr)
 =\Theta\!\left(
 n+\sqrt n\,\frac{\Delta L_{\mathsf s}}{\varepsilon^2}
 \right)
 \label{eq:matching-complexity}
\end{equation}
whenever $0<\varepsilon^2\le c_0\Delta L_{\mathsf s}$.  The upper bound follows
from PAGE and SPIDER \citep{fang2018spider,li2021page}, since individual
$L_{\max}$-smoothness implies mean-squared smoothness with
$L_{\rm ms}\le L_{\max}$.  Theorem~\ref{thm:nonconvex-main} gives the lower
bound for $\mathsf s=\mathrm{ind}$, and the same hard instances satisfy the
mean-squared condition at the same numerical bound.
\end{corollary}

For the PL result, use the corresponding smoothness constant and condition
number:
\begin{equation}
 (L_{\mathsf s},\kappa_{\mathsf s})=
 \begin{cases}
  (L_{\max},\kappa_{\max}),&\mathsf s=\mathrm{ind},\\
  (L_{\rm ms},\kappa_{\rm ms}),&\mathsf s=\mathrm{ms},
 \end{cases}
 \qquad
 \kappa_{\max}:=\frac{L_{\max}}\mu,
 \quad
 \kappa_{\rm ms}:=\frac{L_{\rm ms}}\mu.
 \label{eq:class-condition-numbers}
\end{equation}

Following the standard PL convention
\citep{polyak1963gradient,karimi2016linear}, define
$\mathcal F_{\mathsf s,\rm PL}^{n,d}(\Delta,L_{\mathsf s},\mu)$ as the subclass
of $\mathcal F_{\mathsf s}^{n,d}(\Delta,L_{\mathsf s})$ for which $F^*>-\infty$
and
\begin{equation}
 \|\nabla F(x)\|^2\ge2\mu(F(x)-F^*)
 \qquad\forall x.
 \label{eq:pl-definition}
\end{equation}

Here $F^*=\inf_xF(x)$.  When $d$ is omitted, the same convention over finite
dimensions applies.

\begin{theorem}[Tight PL complexity under both smoothness assumptions]
\label{thm:pl-main}
There is a universal constant $c_\varepsilon>0$ such that,
for every integer $n\ge n_0$, every $\mathsf s\in\{\mathrm{ind},\mathrm{ms}\}$,
all $L_{\mathsf s},\mu,\Delta>0$ with $\kappa_{\mathsf s}\ge3$, every
$\varepsilon$ satisfying $0<\varepsilon\le c_\varepsilon\Delta$, and every
fixed $\delta\in(0,p_0)$,
\begin{equation}
 \operatorname{Comp}_{\varepsilon,\delta}^{\rm IFO}
 \bigl(\mathcal F_{\mathsf s,\rm PL}^n(\Delta,L_{\mathsf s},\mu);
 \mathcal E^{\rm PL}\bigr)
 =
 \begin{cases}
 \displaystyle
 \Theta\!\left(n+
 \frac{n\log(\Delta/\varepsilon)}
 {1+\log(\sqrt n/\kappa_{\mathsf s})}\right),
 &3\le\kappa_{\mathsf s}\le\sqrt n,\\[2mm]
 \displaystyle
 \Theta\!\left(n+
 \kappa_{\mathsf s}\sqrt n\log\frac{\Delta}{\varepsilon}\right),
 &\kappa_{\mathsf s}\ge\sqrt n
 \end{cases}
 .
 \label{eq:tight-pl-complexity}
\end{equation}
\end{theorem}

The hard instances are individually $L_{\max}$-smooth and also satisfy
mean-squared smoothness with bound $L_{\max}$.  In the
small-$\kappa_{\max}$ branch, their average
objective is globally strongly convex.  By contrast, the average-smooth hard
instances of \citet{zhou2019lower} and \citet{han2024lower} have component
smoothness of order $\sqrt n\,L_{\rm ms}$.  Rewriting their bounds in terms of
$L_{\max}$ therefore loses the factor $\sqrt n$.  Section~\ref{sec:restarted-page}
gives the matching Restarted PAGE upper bound under mean-squared smoothness in
terms of $\kappa_{\rm ms}$.

\section{Dense Weak Hiding}
\label{sec:construction}

Dense weak hiding combines a fixed encoding table with a bounded radial map
and a smooth gate.  The table preserves the hidden signal in the average, and
the radial map turns any query point into an alignment vector in the unit ball.
Because the gate satisfies $G=G'=0$ below its threshold, unopened links
contribute neither value nor gradient.  Figure~\ref{fig:pipeline} shows these
ingredients and the resulting bias tradeoff.

\paragraph{Fixed encoding table.}
Fix a bias $\rho\in(0,1/8]$ on the parity-compatible grid specified in
Appendix~\ref{app:gate-code}.  At each stage $j$, independently draw
$\Theta^{(j)}\sim\operatorname{Unif}\{\pm1\}^D$ and sample one $n\times D$
sign table $S^{(j)}=(S^{(j)}_{i\ell})$ whose columns satisfy
\begin{equation}
 \sum_{i=1}^n S^{(j)}_{i\ell}=n\rho\Theta^{(j)}_\ell.
 \label{eq:near-balanced}
\end{equation}
Component $i$ stores row $S^{(j)}_{i:}\in\{\pm1\}^D$.  The table is fixed before
the interaction, so querying the same component returns the same row.  Any
one row is only weakly biased toward $\Theta^{(j)}$, but the column constraint
gives the exact identity
$n^{-1}\sum_iS^{(j)}_{i:}=\rho\Theta^{(j)}$.  The average therefore preserves
the full direction after rescaling by $1/\rho$.  Appendix~\ref{app:information} proves
that choosing indices from previous replies still requires many distinct rows.

\paragraph{Bounded radial map.}
Following the bounded radial map of \citet{carmon2020lower1} and its adaptation
to finite sums with unrestricted queries by \citet{emmenegger2022oracle}, define
\begin{equation}
 \Psi(x)=\frac{x}{\sqrt{1+\|x\|^2}},
 \qquad u_\theta=\frac\theta{\sqrt D}.
\end{equation}
For this local calculation, suppress the stage index.  Given a direction
$\theta$ and a row $S_i$, define
\begin{equation}
 q_\theta(x)=\langle u_\theta,\Psi(x)\rangle,
 \qquad
 \widetilde q_i(x)=\frac1\rho
 \left\langle\frac{S_i}{\sqrt D},\Psi(x)\right\rangle.
 \label{eq:radial-hidden-coordinate}
\end{equation}
Then $n^{-1}\sum_i\widetilde q_i=q_\theta$ exactly.  Since $\|\Psi(x)\|<1$,
even an arbitrarily large query point produces a bounded alignment vector.
Appendix~\ref{app:gate-code} shows that the derivative bounds for $\Psi$ do not
depend on $D$, while those for $\widetilde q_i$ scale as $O(1/\rho)$.

\paragraph{Smooth threshold gate.}
Fix $0<a<b<1$ and a nondecreasing $G\in C^\infty(\R)$ satisfying
\begin{equation}
 G(s)=G'(s)=0 \quad(s\le a),
 \qquad
 G(s)=1,\quad G'(s)=0 \quad(s\ge b).
 \label{eq:flat-gate-summary}
\end{equation}
Below $a$, a link and all of its value and gradient contributions vanish
exactly.  Thus unopened links are invisible in an IFO reply.  Bounded
derivatives of $G$ and the bounds for $\Psi$ control each component Hessian.
Appendix~\ref{app:gate-code} gives the explicit choices.
Appendix~\ref{app:information} makes the coupling across stages precise.

\begin{figure}[H]
\centering
\includegraphics[width=\linewidth]{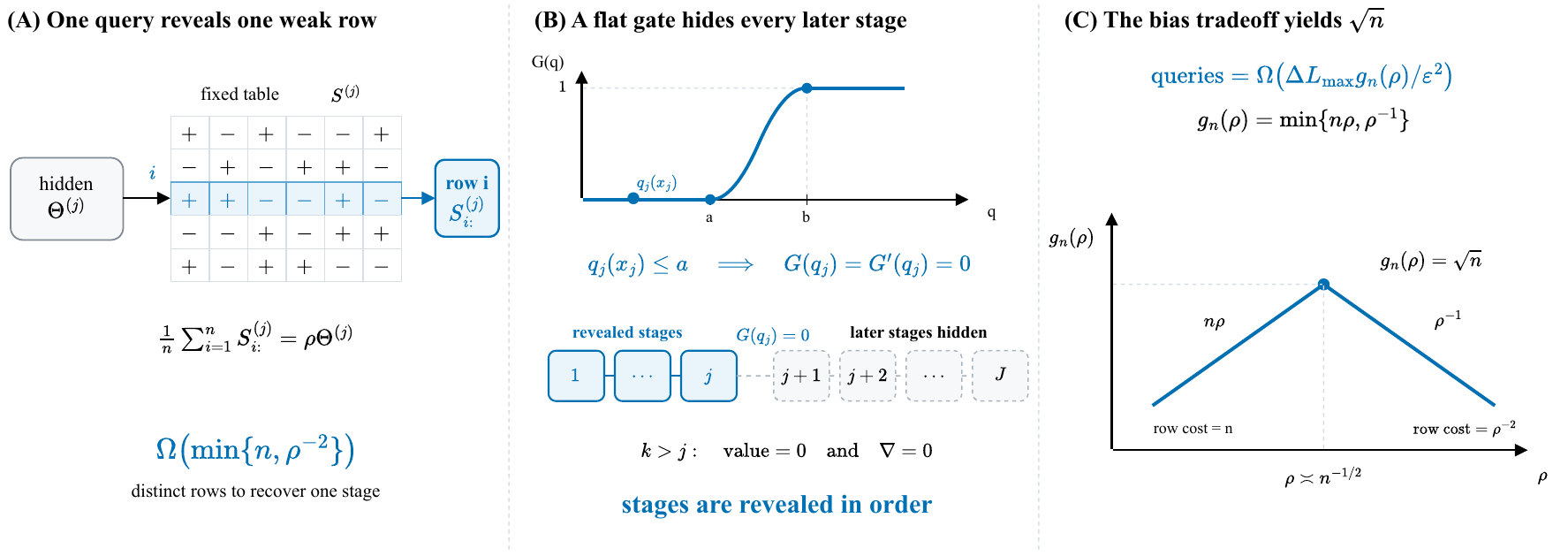}
\caption{Dense weak hiding.  (A) A fixed sign table encodes each stage
direction; the blue outline marks the row accessed when component $i$ is
queried, and
averaging all rows gives $\rho\Theta^{(j)}$.  (B) Below $a$, the link
immediately following stage $j$ contributes neither value nor gradient.  On
$\cEfuture^c$, later links remain inactive while stage $j$ is current.  (C) The product of the row cost per stage and the
number of stages is proportional to $g_n(\rho)$, which is maximized at
$\rho\asymp n^{-1/2}$ and yields $\sqrt n$.}
\label{fig:pipeline}
\end{figure}

\paragraph{Origin of the $\sqrt n$ factor.}
The one-stage information bound in Section~\ref{sec:information} shows that
completing one stage requires
\begin{equation}
 r_\rho=\Omega\!\left(\min\{n,\rho^{-2}\}\right)
\end{equation}
distinct rows, even when their indices are chosen from all previous responses.
The nonconvex scaling in Appendix~\ref{app:nonconvex} can take
$T_\rho=\Theta(\Delta L_{\max}\rho/\varepsilon^2)$ stages and still preserve
the gap and individual smoothness bounds.  Therefore
\begin{equation}
 r_\rho T_\rho
 =\Omega\!\left(
 \frac{\Delta L_{\max}}{\varepsilon^2}
 \min\{n\rho,\rho^{-1}\}
 \right).
 \label{eq:bias-smoothness-tradeoff}
\end{equation}
The two branches balance at $n\rho=\rho^{-1}$, hence
$\rho\asymp n^{-1/2}$ and the product is
$\Omega(\sqrt n\,\Delta L_{\max}/\varepsilon^2)$.
Proposition~\ref{prop:bias-smoothness-tradeoff} in
Appendix~\ref{app:nonconvex} formalizes this tradeoff for every admissible bias,
using the one-stage row cost, sequential revelation, and the gradient bound for
the final stage.  Appendix~\ref{app:gate-code} handles the discrete choice of
$\rho$.

\section{Information Cost of Revealing Hidden Stages}
\label{sec:information}

The one-stage information bound allows adaptive row selection.  In the
truncated interaction, $(\Theta^{(j)},S^{(j)})$ is independent of the history
before its initialization row is returned.  Coupling transfers the revelation
bound to the actual interaction on $\cEfuture^c$, where stage costs add
sequentially.

\subsection{One-stage information bound}

\paragraph{Cost of one stage.}
Fix an admissible bias $\rho$.  In a sufficiently large block dimension,
Lemma~\ref{lem:tunable-bias-kl} gives the row threshold
\begin{equation}
 r_\rho=\Omega\!\left(\min\{n,\rho^{-2}\}\right).
\end{equation}
Before this many distinct rows have been revealed, the probability that any
query chosen from the preceding history crosses the gate threshold is at most a
fixed constant below $1/2$.  Repeated indices and side information independent
of the hidden stage do not improve this bound.

Order the distinct indices by their first occurrence.  Conditional on the
past, exchangeability makes each newly observed row a uniformly chosen
remaining row of the fixed table.  The mutual information between $\Theta$ and
the first $r$ such rows is only $O(Dr\rho^2)$, and repeats add nothing.  The
data processing inequality then bounds the probability of alignment.  At
$\rho\asymp n^{-1/2}$, this requires $\Omega(n)$ distinct rows.

\subsection{Sequential composition across stages}

Stage $j$ becomes current after stages $1,\ldots,j-1$ have been completed, and
it is completed when its alignment first exceeds the gate threshold $a$.
Progress can be faster than this sequential process in two ways: the current
stage may finish in unusually few calls, or a query may align with a stage that
is not yet current.  The one-stage bound controls the first event.  The event
$\cEfuture$ covers the second.  On $\cEfuture^c$, unopened
future links satisfy $G=G'=0$, so stages can only be completed in order.

Let $\nu_N$ be the number of stages completed within $N$ calls, and let
$\cEfuture$ be the event that a query or the final output crosses the alignment
threshold for a stage beyond the current one.

\paragraph{Stage costs add.}
Assign each call to the stage that is current when the query is made.  Suppose
every reached stage has conditional probability at most $p$ of finishing in
fewer than $r$ assigned calls.  If $J$ stages finish within $N\le rJ/2$ calls,
at least $J/2$ of them must finish in fewer than $r$ calls.  Their expected
number is at most $pJ$, so Markov's inequality gives
$\Prob(\nu_N\ge J)\le2p$.  Together with
$\Prob(\cEfuture)\le\delta$, Theorem~\ref{thm:sequential} gives
\begin{equation}
 \Prob(\nu_N<J,\ \cEfuture^c)\ge1-2p-\delta.
\end{equation}

On $\cEfuture^c$, every omitted value and gradient term is zero because its gate
satisfies $G=G'=0$.  Appendix~\ref{app:information} gives the filtration and coupling
details.

\section{Derivation of the Main Rates}
\label{sec:hard-chains}

At the optimized bias, one hidden stage costs $\Omega(n)$ distinct rows, and
these costs add across the stages, which are revealed in order.  The three constructions
below specify the number of stages and the error that remains when one of them
is unfinished.

\subsection{Nonconvex stationarity}

Write the scaled objective as $F(y)=\alpha_\rho R(\beta_\rho y)$.  Up to
universal constants, individual smoothness, the initial gap, and the gradient
left by an unfinished final stage are calibrated by
\begin{equation}
 \frac{\alpha_\rho\beta_\rho^2}{\rho}\asymp L_{\max},
 \qquad
 \alpha_\rho T_\rho\asymp\Delta,
 \qquad
 \alpha_\rho\beta_\rho\asymp\varepsilon.
 \label{eq:main-nonconvex-calibration}
\end{equation}
Solving these relations gives
$\alpha_\rho\asymp\varepsilon^2/(L_{\max}\rho)$,
$\beta_\rho\asymp L_{\max}\rho/\varepsilon$, and
$T_\rho\asymp\Delta L_{\max}\rho/\varepsilon^2$.
Appendix~\ref{app:nonconvex} fixes the universal constants so that individual
smoothness is at most $L_{\max}$, the initial gap is at most $\Delta$, and the
gradient barrier is at least $8\varepsilon$.  At the optimized bias
$\rho\asymp n^{-1/2}$, the one-stage
row threshold is $r_\rho=\Theta(n)$, and hence
$r_\rho T_\rho=\Omega(\sqrt n\,\Delta L_{\max}/\varepsilon^2)$.
When $T_\rho<2$, the auxiliary quadratic in Appendix~\ref{app:nonconvex} gives
the complementary $\Omega(n)$ term.  This proves
Theorem~\ref{thm:nonconvex-main}.

\subsection{Global PL: two geometrically scaled constructions}

Both PL constructions use $\rho\asymp n^{-1/2}$ and hence
$r_\rho=\Theta(n)$ rows per stage.  Up to universal constants, their scales set
$\Delta$, $L_{\max}$, $\mu$, and the remaining objective gap $\varepsilon$ as
follows:
\begin{equation}
\begin{array}{c|cccc}
 & \Delta & L_{\max} & \mu & \varepsilon\\ \hline
 \text{small }\kappa_{\max}
 & \zeta_1^2/\mu
 & \zeta_j/(\rho W_j)
 & \zeta_j/R_j
 & \zeta_{J+1}^2/\mu\\[2pt]
 \text{large }\kappa_{\max}
 & \alpha/(1-r_{\rm w})
 & \alpha\beta^2/\rho
 & \alpha\beta^2(1-r_{\rm w})
 & \Delta r_{\rm w}^J
\end{array}
\label{eq:main-pl-calibration}
\end{equation}
For small $\kappa_{\max}$, the $\mu$ entry is a strong convexity constant.  For
large $\kappa_{\max}$, it is the PL constant.  The last column is the objective
gap left by the unrevealed terminal stage or suffix.

\paragraph{Geometric stage scales.}
Under PL, the gradient at the first unfinished stage must control the objective
contribution of every later stage, so the stage scales decrease geometrically.
The large $\kappa_{\max}$ construction uses one scale $w_j$ for both the gate
and the stage contribution.  The first two entries of its row in
\eqref{eq:main-pl-calibration} give
$\alpha\asymp\Delta(1-r_{\rm w})$ and
$\beta^2\asymp L_{\max}\rho/[\Delta(1-r_{\rm w})]$.
Dividing the $\mu$ and $L_{\max}$ entries then gives
$1-r_{\rm w}\asymp1/(\kappa_{\max}\rho)$.
A geometric chain requires $0<r_{\rm w}<1$, so this scaling works only
when $\kappa_{\max}\gtrsim\sqrt n$.  Appendix~\ref{app:pl} makes this concrete by
taking $r_{\rm w}\in[1/2,1)$.  For
$3\le\kappa_{\max}<\sqrt n$, the required geometric ratio is no longer
admissible.  In the small $\kappa_{\max}$ row, the smoothness and curvature
entries instead give $R_j/W_j\asymp\kappa_{\max}\rho$.  The construction also
takes $\vartheta\asymp R_j/W_j$ to control the remaining link curvature.  Thus
$R_j\ll W_j$ when $\kappa_{\max}\ll\sqrt n$.  Completion occurs above alignment
$a$.  A gate is \emph{fully open} only above $b$, where $G=1$ and $G'=0$.

\paragraph{Small $\kappa_{\max}$: separate gate and contribution scales.}
The gate scale $R_j$ sets the displacement needed to open the next stage, while
the larger contribution scale $W_j$ keeps each component Hessian within the
$L_{\max}$ budget.
Since $\zeta_{j+1}=\vartheta\zeta_j$, the $\Delta$ and $\varepsilon$ entries in
\eqref{eq:main-pl-calibration} require
$\vartheta^{2J}\asymp\varepsilon/\Delta$.  For
$3\le\kappa_{\max}\le\sqrt n$, one has
$\vartheta\asymp\kappa_{\max}/\sqrt n$, and therefore
\begin{equation}
 J=\Theta\!\left(
 \frac{\log(\Delta/\varepsilon)}
 {1+\log(\sqrt n/\kappa_{\max})}\right),
 \qquad
 r_\rho J=\Theta\!\left(
 \frac{n\log(\Delta/\varepsilon)}
 {1+\log(\sqrt n/\kappa_{\max})}\right).
 \label{eq:main-small-kappa-lower}
\end{equation}
Appendix~\ref{app:pl} verifies that every component is $L_{\max}$-smooth, the
average is $\mu$-strongly convex, and an unfinished last stage leaves objective
gap at least $8\varepsilon$.  The auxiliary quadratic covers the remaining
case, where the target expression is $O(n)$.

\paragraph{Large $\kappa_{\max}$: geometric weights.}
For a universal constant $c_{\rm w}>1$ and
$\kappa_{\max}\ge c_{\rm w}\sqrt n$, the weighted construction uses the same
scale $w_j$ for the gate and stage contribution.  The $\Delta$ and
$\varepsilon$ entries in~\eqref{eq:main-pl-calibration} give
$r_{\rm w}^J\asymp\varepsilon/\Delta$.  Together with
$1-r_{\rm w}\asymp\sqrt n/\kappa_{\max}$, this gives
\begin{equation}
 J=\Theta\!\left(\frac{\kappa_{\max}}{\sqrt n}
 \log\frac{\Delta}{\varepsilon}\right),
 \qquad
 r_\rho J=\Theta\!\left(\kappa_{\max}\sqrt n
 \log\frac{\Delta}{\varepsilon}\right).
 \label{eq:main-large-kappa-lower}
\end{equation}
If no stage $k>J$ has been completed, the hidden suffix leaves objective gap at
least $8\varepsilon$.  To verify the PL property, let $j_\star$ be the first
gate that is not fully open.  Appendix~\ref{app:pl} shows that the gradient at
this stage controls its scale, while the geometric weights control all
remaining stages:
\begin{equation}
 \|\nabla R(x)\|^2\gtrsim w_{j_\star}^2,
 \qquad
 \sum_{k\ge j_\star}w_k^2
 \lesssim\frac{w_{j_\star}^2}{1-r_{\rm w}}.
 \label{eq:main-weighted-barrier-tail}
\end{equation}
Together with the strong convexity of the preceding fully open blocks, these
bounds give the global PL inequality.  Global rescaling leaves the condition
number unchanged.
The two-scale construction covers
$\sqrt n\le\kappa_{\max}<c_{\rm w}\sqrt n$.

\section{Restarted PAGE}
\label{sec:restarted-page}

\paragraph{What the lower bound suggests.}
On the small-$\kappa$ hard instance, revealing a stage costs $\Omega(n)$ calls
while the gap shrinks by a factor of $\Theta((\kappa_{\max}/\sqrt n)^2)$.  This leads to
PAGE epochs costing $\Theta(n)$ calls in expectation when
$\kappa_{\rm ms}\le\sqrt n$, with conditional expected gap at most
$q_{\rm ep}=\Theta(\kappa_{\rm ms}/\sqrt n)$ times the starting gap.

Algorithm~\ref{alg:restarted-page} implements this schedule.  Its expected IFO
complexity is
\begin{equation}
 \begin{cases}
 \displaystyle O\!\left(n+\frac{n\log(\Delta/\varepsilon)}
 {1+\log(\sqrt n/\kappa_{\rm ms})}\right),
 &1\le\kappa_{\rm ms}\le\sqrt n,\\[3mm]
 \displaystyle O\!\left(n+\kappa_{\rm ms}\sqrt n
 \log\frac{\Delta}{\varepsilon}\right),
 &\kappa_{\rm ms}\ge\sqrt n.
 \end{cases}
 \label{eq:restarted-page-main-rate}
\end{equation}
\begingroup
\setlength{\intextsep}{7pt plus 1pt minus 1pt}
\begin{algorithm}[!ht]
\caption{Restarted PAGE}
\label{alg:restarted-page}
\footnotesize
\begin{algorithmic}[1]
\Require $x^{(0)}$, $n$, a mean-squared smoothness bound $L_{\rm ms}$, a PL
  constant $\mu$, a gap bound $F(x^{(0)})-F^*\le\Delta$, and a target accuracy
  $\varepsilon\in(0,\Delta)$;
  $\kappa_{\rm ms}:=L_{\rm ms}/\mu\ge1$
\State $p\gets1/(n+1)$; $\eta\gets1/[L_{\rm ms}(1+\sqrt n)]$
\State $A_{\rm ep}\gets\kappa_{\rm ms}(1+\sqrt n)$;
  $T\gets\lceil4(n+A_{\rm ep})\rceil$; $q_{\rm ep}\gets A_{\rm ep}/T$
\State $S_{\rm ep}\gets\left\lceil
  \log(\Delta/\varepsilon)/\log(1/q_{\rm ep})\right\rceil$
\For{$s=0,\ldots,S_{\rm ep}-1$}
  \State $x_0\gets x^{(s)}$;
    $g_0\gets n^{-1}\sum_{i=1}^n\nabla f_i(x_0)$
  \For{$t=0,\ldots,T-1$}
    \State $x_{t+1}\gets x_t-\eta g_t$
    \If{$t<T-1$}
    \State With probability $p$, set
      $g_{t+1}\gets n^{-1}\sum_{i=1}^n\nabla f_i(x_{t+1})$
    \State Otherwise draw $i_t\sim\operatorname{Unif}[n]$ and set
      $g_{t+1}\gets g_t+\nabla f_{i_t}(x_{t+1})-\nabla f_{i_t}(x_t)$
    \EndIf
  \EndFor
  \State Draw $U_s\sim\operatorname{Unif}\{0,\ldots,T-1\}$ and set
    $x^{(s+1)}\gets x_{U_s}$
\EndFor
\State \Return $x^{(S_{\rm ep})}$
\end{algorithmic}
\end{algorithm}
\endgroup

Lemma~\ref{lem:page-epoch-contraction} shows that one epoch costs at most
$n+3T$ IFO calls in expectation and that
\begin{equation}
 \E\!\left[F(x^{(s+1)})-F^*\mid x^{(s)}\right]
 \le q_{\rm ep}\bigl(F(x^{(s)})-F^*\bigr).
 \label{eq:page-main-contraction}
\end{equation}
When $\kappa_{\rm ms}\le\sqrt n$, each epoch costs $O(n)$, and
\begin{equation*}
 q_{\rm ep}=\Theta\!\left(\frac{\kappa_{\rm ms}}{\sqrt n}\right),
 \qquad
 S_{\rm ep}=O\!\left(
 1+\frac{\log(\Delta/\varepsilon)}
 {1+\log(\sqrt n/\kappa_{\rm ms})}
 \right).
\end{equation*}
When $\kappa_{\rm ms}\ge\sqrt n$, each epoch costs
$O(\kappa_{\rm ms}\sqrt n)$ calls in expectation and contracts the conditional
expected gap by a constant factor.
These two cases give~\eqref{eq:restarted-page-main-rate}, matching the upper
bounds in Theorem~\ref{thm:pl-main}.  Appendix~\ref{app:page-upper} proves the
contraction and IFO cost; Appendix~\ref{app:stopping} converts this expected
guarantee to a fixed call budget with any fixed constant success probability.

\section{Conclusion and Limitations}
\label{sec:limitations}

We establish tight lower bounds under individual smoothness for nonconvex
optimization and for both $\kappa_{\max}$ regimes under PL.  Restarted PAGE
matches the PL rates, and the small-$\kappa_{\max}$ construction has a globally
strongly convex average objective.  Open directions include lower bounds in a
fixed ambient dimension, noisy finite-sum oracle models, and the range
$\kappa<3$, where the condition number is bounded by a constant.  Our lower
bounds also assume that $n$ is above a universal threshold.

\clearpage
\bibliography{references}
\bibliographystyle{iclr2027_conference}

\section*{AI Use Statement}
During the preparation of this manuscript, the authors used large language
model assistants to support the development and writing of mathematical proofs
and to improve the clarity and readability of the manuscript.

\clearpage
\appendix
\paragraph{Appendix dependencies.}
The diagram separates the shared argument for the lower bounds from the nonconvex and PL
branches and from the Restarted PAGE analysis.  Solid arrows indicate proof
dependencies.

\begin{figure}[H]
\centering
\includegraphics[width=\linewidth]{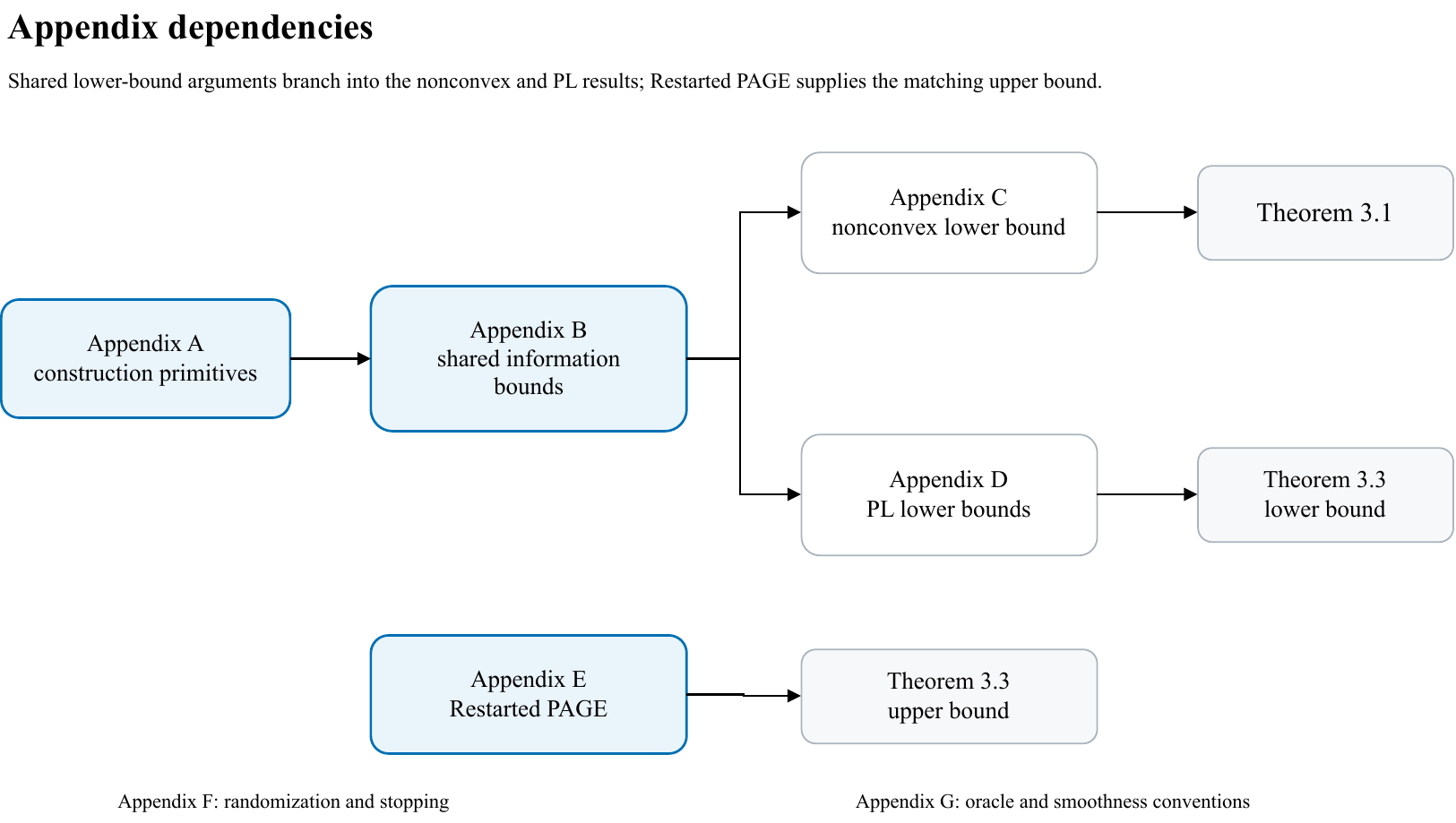}
\caption{Proof dependencies for the main results.  Appendices~\ref{app:stopping}
and~\ref{app:smoothness-calibration} supply the transfers and conventions used
by the branches proving the lower and upper bounds.}
\label{fig:appendix-proof-map}
\end{figure}

\section{Common Definitions and Calculus}
\label{app:gate-code}

This appendix defines the three objects shared by the lower bounds: a smooth
threshold gate, a fixed encoding table, and a bounded radial map.
Random objects for stage $j$ carry the superscript $(j)$.  The indices
$i\in[n]$ and $\ell\in[D]$ refer to a component row and a coordinate.

\subsection{Constructions and basic properties}

\begin{lemma}[Explicit smooth threshold gate]
\label{lem:explicit-flat-gate}
Set
\begin{equation}
 a=\frac5{32},
 \qquad
 b=\frac9{32}.
 \label{eq:gate-thresholds}
\end{equation}
There is an explicit nondecreasing $G\in C^\infty(\R)$ with $0\le G\le1$
such that
\begin{align}
 G^{(r)}(s)&=0
 \quad\text{for }s\le a,\ r\ge0,
 \label{eq:lower-flat}\\
 G(s)&=1,\qquad G^{(r)}(s)=0
 \quad\text{for }s\ge b,\ r\ge1.
 \label{eq:upper-flat-positive-orders}
\end{align}
Its first two derivatives satisfy the dimension-free bounds
\begin{equation}
 \|G'\|_\infty\le17,
 \qquad
 \|G''\|_\infty\le259.
 \label{eq:scaled-gate-derivatives}
\end{equation}
\end{lemma}

\paragraph{Encoding table and admissible bias grid.}

Define the admissible bias grid
\begin{equation}
 \mathcal R_n:=\left\{\frac mn:
 m\in\mathbb Z_{>0},\ m\equiv n\pmod 2,\ m\le\frac n8\right\}.
 \label{eq:admissible-bias-grid}
\end{equation}
Fix any $\rho\in\mathcal R_n$; equivalently, choose a positive integer $m$ with
the same parity as $n$ and $1\le m\le n/8$, and set
\begin{equation}
 \rho:=\frac mn.
 \label{eq:general-m-rho}
\end{equation}
The construction and the information bound are uniform over every such
choice.  Consecutive values of $\rho$ in $\mathcal R_n$ are $2/n$ apart, and no
relation between $m$ and $\sqrt n$ is assumed here.

Before oracle interaction, sample the complete stage pairs independently.  For
each $j$, draw $\Theta^{(j)}\sim\operatorname{Unif}\{\pm1\}^D$ and, conditional on
$\Theta^{(j)}$, form the matrix
$S^{(j)}=(S^{(j)}_{i\ell})_{i\in[n],\ell\in[D]}\in\{\pm1\}^{n\times D}$.
Its $i$th row $S^{(j)}_{i:}\in\{\pm1\}^D$ is stored by component $i$.
For each coordinate $\ell$, write its column as
$S^{(j)}_\ell=(S^{(j)}_{1\ell},\ldots,S^{(j)}_{n\ell})\in\{\pm1\}^n$.
Conditional on $\Theta^{(j)}$, draw the columns
$S^{(j)}_1,\ldots,S^{(j)}_D$ independently and uniformly subject to
\begin{equation}
 \mathbf 1^\top S^{(j)}_\ell
 =\sum_{i=1}^nS^{(j)}_{i\ell}
 =m\Theta^{(j)}_\ell.
 \label{eq:column-constraint}
\end{equation}

Write the complete stage variable as
\begin{equation}
    \cZ^{(j)}=(\Theta^{(j)},S^{(j)}).
\end{equation}
\begin{lemma}[Properties of the encoding table]
\label{lem:encoding-table-properties}
For every stage $j$:
\begin{enumerate}[label=\textup{(\roman*)},leftmargin=*,itemsep=2pt]
 \item $S^{(j)}$ is sampled once before oracle interaction.  Conditional on
 $\Theta^{(j)}$, its columns are independent, and the pairs
 $\cZ^{(j)}=(\Theta^{(j)},S^{(j)})$ are independent across stages.
 \item The component rows have the exact average
 \begin{equation}
  \frac1n\sum_{i=1}^nS^{(j)}_{i:}=\rho\Theta^{(j)}.
  \label{eq:exact-code-average}
 \end{equation}
 \item After any $d$ distinct rows have been chosen causally, with each new
 index fixed before its row is observed, the remaining rows are uniformly
 assigned to the unseen labels subject to the residual column sums.
\end{enumerate}
\end{lemma}

\begin{corollary}[Optimized admissible bias]
\label{cor:optimized-near-balanced-bias}
Let $n\ge1024$, let $m$ be the least integer not smaller than $\sqrt n$ that
has the parity of $n$, and set $\rho=m/n$.  Then
\begin{equation}
 \sqrt n\le m<\sqrt n+2,
 \qquad \rho=\frac mn,
 \qquad
 0\le\rho-\frac1{\sqrt n}<\frac2n,
 \qquad
 \rho<\frac1{\sqrt n}+\frac2n
 \le\frac{17}{16\sqrt n}
 \le\frac{17}{512}
 =:\overline\rho.
 \label{eq:m-rho}
\end{equation}
For $g_n(\rho):=\min\{n\rho,\rho^{-1}\}$, this admissible grid point satisfies
\begin{equation}
 g_n(\rho)=\rho^{-1}
 >\frac{\sqrt n}{1+2/\sqrt n}
 \ge\frac{16}{17}\sqrt n.
 \label{eq:appendix-discrete-envelope-value}
\end{equation}
\end{corollary}

\paragraph{Bounds for the radial map.}

For $x\in\R^D$, define
\begin{equation}
 \Psi(x)=\frac{x}{\sqrt{1+\|x\|^2}}.
 \label{eq:radial-saturation}
\end{equation}
For any unit vector $u$, let
\begin{equation}
    q_u(x)=\langle u,\Psi(x)\rangle.
\end{equation}
Suppressing the stage index, for a hidden direction $\theta$ and the rows
$S_1,\ldots,S_n$ of the preceding encoding table, set
\begin{equation}
 u_\theta=\frac\theta{\sqrt D},
 \qquad s_i=\frac{S_i}{\sqrt D},
 \qquad q_\theta=q_{u_\theta},
 \qquad \widetilde q_i(x)=\rho^{-1}\langle s_i,\Psi(x)\rangle.
 \label{eq:radial-global-component}
\end{equation}
\begin{lemma}[Properties of the radial coordinate]
\label{lem:radial-first-second}
For every unit $u$, one has $\|\Psi(x)\|<1$ and, with
$r_x=(1+\|x\|^2)^{1/2}$,
\begin{equation}
 \nabla q_u(x)=\frac{u}{r_x}
 -\frac{\langle u,x\rangle x}{r_x^3},
 \qquad
 \|\nabla q_u(x)\|\le1,
 \label{eq:radial-gradient}
\end{equation}
and
\begin{equation}
 \nabla^2q_u(x)
 =-\frac{ux^\top+xu^\top+\langle u,x\rangle I}{r_x^3}
 +\frac{3\langle u,x\rangle xx^\top}{r_x^5},
 \qquad
 \|\nabla^2q_u(x)\|_{\op}\le2.
 \label{eq:radial-hessian}
\end{equation}
Consequently,
\begin{equation}
 |\widetilde q_i(x)|\le\frac1\rho,
 \qquad
 \|\nabla\widetilde q_i(x)\|\le\frac1\rho,
 \qquad
 \|\nabla^2\widetilde q_i(x)\|_{\op}\le\frac2\rho.
 \label{eq:component-radial-bounds}
\end{equation}
All constants are independent of $D$.
Moreover, the identity for the fixed table gives
\begin{equation}
 \frac1n\sum_{i=1}^n\widetilde q_i(x)=q_\theta(x)
 \qquad\forall x.
 \label{eq:radial-exact-average}
\end{equation}
\end{lemma}

For any normalized stage coordinate $z_j$ used below,
$V_j=\Psi(z_j)$ satisfies the unit ball condition in
Lemma~\ref{lem:tunable-bias-kl}.

\subsection{Proofs of the shared primitives}

\begin{proof}[Proof of Lemma~\ref{lem:explicit-flat-gate}]
We use the following explicit construction.  Let
\begin{equation}
 \delta_{\rm g}:=\frac1{512},
 \qquad
 w_{\rm g}:=1-2\delta_{\rm g}=\frac{255}{256},
\end{equation}
and define the triangular density of unit mass
\begin{equation}
 p_{\rm g}(t):=
 \begin{cases}
 \dfrac{4(t-\delta_{\rm g})}{w_{\rm g}^2},
    &\delta_{\rm g}\le t\le\dfrac12,\\[2mm]
 \dfrac{4(1-\delta_{\rm g}-t)}{w_{\rm g}^2},
    &\dfrac12\le t\le1-\delta_{\rm g},\\[2mm]
 0,&\text{otherwise}.
 \end{cases}
 \label{eq:triangular-density}
\end{equation}
Define the normalized $C^\infty$ bump
\begin{equation}
 Z_{\rm g}:=\int_{-\delta_{\rm g}}^{\delta_{\rm g}}
 \exp\!\left(-\frac1{\delta_{\rm g}^2-u^2}\right)du,
 \qquad
 \phi_{\rm g}(s):=
 \begin{cases}
 Z_{\rm g}^{-1}\exp\!\left(-\dfrac1{\delta_{\rm g}^2-s^2}\right),
    &|s|<\delta_{\rm g},\\[2mm]
 0,&|s|\ge\delta_{\rm g}.
 \end{cases}
 \label{eq:gate-smoothing-bump}
\end{equation}
With convolution denoted by $*$, set
\begin{equation}
 \begin{aligned}
 g_{\rm g}(t)&:=(p_{\rm g}*\phi_{\rm g})(t),\\
 H_0(t)&:=\int_{-\infty}^t g_{\rm g}(u)\,du,
 \qquad
 G(s):=H_0\!\left(8s-\frac54\right).
 \end{aligned}
 \label{eq:explicit-infinite-gate}
\end{equation}

\emph{Plateaus.}
The density $p_{\rm g}$ is nonnegative, integrates to one, and satisfies
\begin{equation}
 \|p_{\rm g}\|_\infty=\frac2{w_{\rm g}},
 \qquad
 \operatorname{Lip}(p_{\rm g})=\frac4{w_{\rm g}^2}.
 \label{eq:triangular-bounds}
\end{equation}
The function $\phi_{\rm g}$ is nonnegative, has unit mass, belongs to
$C_c^\infty(\R)$, and is supported on
$[-\delta_{\rm g},\delta_{\rm g}]$.  Hence $g_{\rm g}$ is
nonnegative, has unit mass, and is supported on $[0,1]$.  Thus $H_0$ is
nondecreasing, equals zero on $(-\infty,0]$, and equals one on $[1,\infty)$.

\emph{Derivative bounds.}
Convolution with a nonnegative kernel of unit mass and
$g_{\rm g}'=p_{\rm g}'*\phi_{\rm g}$ give
\begin{equation}
 \|H_0'\|_\infty\le\frac2{w_{\rm g}}<\frac{17}{8},
 \qquad
 \|H_0''\|_\infty\le\frac4{w_{\rm g}^2}<\frac{259}{64}.
 \label{eq:base-gate-derivatives}
\end{equation}
The affine rescaling in~\eqref{eq:explicit-infinite-gate} gives the two
plateaus.  It also gives $G'(s)=8H_0'(8s-5/4)$ and
$G''(s)=64H_0''(8s-5/4)$, which yield the stated derivative bounds.  On the
fully open region, $G=1$ and only the derivatives of positive order vanish.
\end{proof}

\begin{proof}[Proof of Lemma~\ref{lem:encoding-table-properties}]
The first claim is part of the sampling rule.  Summing each column and using
\eqref{eq:column-constraint} proves the second.  For the third, conditioning on
revealed positions in a uniformly ordered column with fixed sign counts leaves a uniform
ordering of its residual signs; induction applies because each new label is
chosen before its row is revealed.
\end{proof}

\begin{proof}[Proof of Corollary~\ref{cor:optimized-near-balanced-bias}]
Consecutive parity-compatible integers differ by two.  For $n\ge1024$, the
chosen $m$ also satisfies $m<n/8$.  The displayed bounds follow directly from
$\sqrt n\le m<\sqrt n+2$ and $\rho\ge n^{-1/2}$.
\end{proof}

\begin{proof}[Proof of Lemma~\ref{lem:radial-first-second}]
\emph{Derivative formulas.}
The formulas follow by differentiation.  The Jacobian of $\Psi$ has operator
norm at most one.  With $s=\|x\|$ and $|\langle u,x\rangle|\le s$, the Hessian
norm is at most
\begin{equation}
 3\frac{s}{(1+s^2)^{3/2}}
 +3\frac{s^3}{(1+s^2)^{5/2}}<2.
\end{equation}
After setting $t=s^2$, the squared inequality is equivalent to
\begin{equation}
\begin{aligned}
 4(1+t)^5-9t(1+2t)^2
 &=4+11t+4t^2+4t^3+20t^4+4t^5>0.
\end{aligned}
\end{equation}

Since $\|s_i\|=1$, multiplying the same bounds by $1/\rho$ proves
\eqref{eq:component-radial-bounds}; the value bound follows from
$\|\Psi(x)\|<1$.  Finally, \eqref{eq:exact-code-average} gives
\[
 \frac1n\sum_{i=1}^n\widetilde q_i(x)
 =\frac1\rho\left\langle
   \frac1n\sum_{i=1}^n\frac{S_i}{\sqrt D},\Psi(x)
  \right\rangle
 =\left\langle\frac\theta{\sqrt D},\Psi(x)\right\rangle
 =q_\theta(x),
\]
which proves~\eqref{eq:radial-exact-average}.
\end{proof}

\subsection{Derivative bounds}

The following two lemmas give the local Hessian and radial gradient bounds used
in both hard constructions.

Set
\begin{equation}
 M_1:=\|G'\|_\infty,
 \qquad
 M_2:=\|G''\|_\infty.
 \label{eq:shared-gate-derivative-norms}
\end{equation}
Let $h:\R^{d_x}\to\R$ and $P:\R^{d_y}\to\R$ be twice continuously
differentiable, let $c\in\R$, and suppose
\begin{equation}
 |P|\le B_0,\quad
 \|\nabla P\|\le B_1,\quad
 \|\nabla^2P\|_{\op}\le B_2,\qquad
 \|\nabla h\|\le g_1,\quad
 \|\nabla^2h\|_{\op}\le g_2.
 \label{eq:gated-link-premises}
\end{equation}
For
\begin{equation}
 \mathcal L(x,y)=-G(h(x))(P(y)+c),
\end{equation}

\begin{lemma}[Hessian bounds for a gated link and its block rows]
\label{lem:gated-link-hessian}
Under these assumptions, the Hessian blocks satisfy
\begin{align}
 \|\nabla^2_{yy}\mathcal L\|_{\op}
 &\le B_2,
 \label{eq:gated-link-current}\\
 \|\nabla^2_{xy}\mathcal L\|_{\op}
 &\le M_1g_1B_1,
 \label{eq:gated-link-cross}\\
 \|\nabla^2_{xx}\mathcal L\|_{\op}
 &\le(B_0+|c|)(M_2g_1^2+M_1g_2).
 \label{eq:gated-link-previous}
\end{align}
Moreover, for any symmetric block matrix $H=(H_{jk})$,
\begin{equation}
 \|H\|_{\op}
 \le\max_j\sum_k\|H_{jk}\|_{\op}.
 \label{eq:block-row-operator-bound}
\end{equation}
\end{lemma}

\begin{lemma}[Radial gradient lower bound]
\label{lem:radial-alignment-barrier}
Set $\lambda=33/16$ and
\begin{equation}
 \chi:=
 \frac{(1-b^2)^2-\lambda b}{\sqrt{1-b^2}}
 >\frac5{18}.
 \label{eq:explicit-radial-chi}
\end{equation}
For every unit vector $u$, every $x\in\R^D$, and every $c\ge1$ such that
$q_u(x)\le b$,
\begin{equation}
 \|\lambda x-c\nabla q_u(x)\|
 \ge\|\lambda x-\nabla q_u(x)\|
 \ge\chi.
 \label{eq:radial-norm-comparison}
\end{equation}
The unique stationary alignment for $c=1$ is a scalar $q_*>b$ satisfying
$(1-q_*^2)^2-\lambda q_*=0$.
\end{lemma}

\subsection{Proofs of the derivative bounds}

\begin{proof}[Proof of Lemma~\ref{lem:gated-link-hessian}]
Direct differentiation gives
\begin{align*}
 \nabla^2_{yy}\mathcal L
 &=-G(h)\nabla^2P,\qquad
 \nabla^2_{xy}\mathcal L
 =-G'(h)\nabla h(\nabla P)^\top,\\
 \nabla^2_{xx}\mathcal L
 &=-(P+c)\bigl[G''(h)\nabla h\nabla h^\top
                  +G'(h)\nabla^2h\bigr].
\end{align*}
The three displayed bounds follow from $0\le G\le1$ and the derivative
bounds on $G$.

For the block statement, set $A_{jk}=\|H_{jk}\|_{\op}$.  If
$z_k=\|v_k\|$, then the vector of block norms of $Hv$ is coordinatewise at
most $Az$, so $\|Hv\|\le\|A\|_2\|v\|$.  Symmetry gives
\begin{equation}
 \|A\|_2\le\sqrt{\|A\|_1\|A\|_\infty}=\|A\|_\infty,
\end{equation}
which is~\eqref{eq:block-row-operator-bound}.
\end{proof}

\begin{proof}[Proof of Lemma~\ref{lem:radial-alignment-barrier}]
Write $x=tu+y$, where $y\perp u$, and set
\begin{equation}
 q=q_u(x),\qquad
 \zeta=(1+\|x\|^2)^{-1}.
\end{equation}
Lemma~\ref{lem:radial-first-second} gives
\begin{equation}
 \langle x,\nabla q_u(x)\rangle=q\zeta,
 \qquad
 \|\nabla q_u(x)\|^2
 =\zeta(1-q^2-q^2\zeta),
 \qquad
 0<\zeta\le1-q^2.
 \label{eq:radial-scalar-identities}
\end{equation}

\emph{Case 1: $0\le q\le b$.}
First expand the dependence on $c$:
\begin{equation}
 \|\lambda x-c\nabla q_u(x)\|^2
 =\lambda^2(\zeta^{-1}-1)-2\lambda cq\zeta
  +c^2\zeta(1-q^2-q^2\zeta).
\label{eq:radial-c-quadratic}
\end{equation}
Its derivative with respect to $c$ is
\begin{equation}
 2\zeta\bigl[-\lambda q+c(1-q^2-q^2\zeta)\bigr].
 \label{eq:radial-c-derivative}
\end{equation}
Setting this derivative to zero gives the unconstrained minimizer
\begin{equation}
 c_0=\frac{\lambda q}{1-q^2-q^2\zeta}
 \le\frac{\lambda q}{(1-q^2)^2}<1.
\end{equation}
The strict inequality follows because
$h(q):=(1-q^2)^2-\lambda q$ is decreasing on $[0,1]$ and
\begin{equation}
 h(b)>\frac4{15}>0.
 \label{eq:stationary-margin}
\end{equation}
Moreover, $\sqrt{1-b^2}<24/25$, so the definition of $\chi$ gives
$\chi>5/18$.

Thus the minimum over $c\ge1$ occurs at $c=1$.  For this value,
\begin{equation}
 \|\lambda x-\nabla q_u(x)\|^2
 =\lambda^2(\zeta^{-1}-1)-2\lambda q\zeta
 +\zeta(1-q^2-q^2\zeta).
\end{equation}
The derivative of the right-hand side is
\begin{equation}
 -\frac{\lambda^2}{\zeta^2}-2\lambda q+(1-q^2)-2q^2\zeta
 \le-\lambda^2+1<0,
\end{equation}
where $q\ge0$ and $\zeta\le1$.  Hence the minimum
occurs at $\zeta=1-q^2$, equivalently $y=0$.  There,
\begin{equation}
 \|\lambda x-\nabla q_u(x)\|
 =\frac{(1-q^2)^2-\lambda q}{\sqrt{1-q^2}},
\end{equation}
which decreases on $[0,b]$ and is therefore at least $\chi$.

\emph{Case 2: $q\le0$.}
Equation~\eqref{eq:radial-c-derivative} is nonnegative for $c\ge1$,
so the squared norm is minimized at $c=1$.  At that value,
\begin{equation}
\begin{aligned}
 \|\lambda x-\nabla q_u(x)\|^2
 ={}&\lambda^2\|x\|^2+\|\nabla q_u(x)\|^2
 -2\lambda\langle x,\nabla q_u(x)\rangle\\
 \ge{}&\lambda^2\|x\|^2+\|\nabla q_u(x)\|^2,
\end{aligned}
\end{equation}
because $\langle x,\nabla q_u(x)\rangle=q\zeta\le0$.

Moreover, $\nabla q_u(x)=J_\Psi(x)^\top u$, and the smallest singular value
of $J_\Psi(x)$ is $(1+\|x\|^2)^{-3/2}$.  Since $\|u\|=1$,
\begin{equation}
 \|\nabla q_u(x)\|^2\ge(1+\|x\|^2)^{-3}.
\end{equation}
Consequently,
\begin{equation}
 \|\lambda x-\nabla q_u(x)\|^2
 \ge\lambda^2\|x\|^2+(1+\|x\|^2)^{-3}
 \ge1>\chi^2.
\end{equation}

The second inequality uses $1-(1+s)^{-3}\le3s$ for $s\ge0$ and
$\lambda^2>3$.  This proves~\eqref{eq:radial-norm-comparison}.  Finally,
the stationary equation for $c=1$ is $h(q_*)=0$; monotonicity and
\eqref{eq:stationary-margin} give the stated uniqueness and $q_*>b$.
\end{proof}

\clearpage
\section{Dense Weak Hiding: One-Stage Information and Sequential Revelation}
\label{app:information}

The one-stage argument proves that alignment above the gate threshold requires
$\Omega(\min\{n,\rho^{-2}\})$ rows, which is $\Omega(n)$ at the optimized
bias.  The sequential argument then gives an $\Omega(rJ)$ cost for
$J$ stages.  Lemmas~\ref{lem:tunable-bias-kl}--%
\ref{lem:precompletion-extension} prove the row cost and cover the completion
query.  Lemmas~\ref{lem:future-alignment}--\ref{lem:pathwise-coupling} and
Theorem~\ref{thm:sequential} give stage independence, coupling, and additivity.

\subsection{One-stage information bound}

\paragraph{One-stage setting.}
For this one-stage argument, suppress the stage superscript and write
$\Theta=\Theta^{(j)}$ and $S=S^{(j)}$.  Write the fixed table as
$S=(S_1,\ldots,S_D)\in\{\pm1\}^{n\times D}$, with entries $S_{i\ell}$,
column $S_\ell=(S_{1\ell},\ldots,S_{n\ell})\in\{\pm1\}^n$, and row
$S_{i:}\in\{\pm1\}^D$.  Column $S_\ell$ encodes the hidden coordinate
$\Theta_\ell$.  Throughout the general result, let $m$ be a positive integer
with the parity of $n$ and $1\le m\le n/8$, and set
$\rho=m/n$, so $\rho\in\mathcal R_n$.  The relation $m\asymp\sqrt n$ is used
only in the optimized result below.

\paragraph{Column sampling under $P$.}
The hard distribution is
\begin{equation}
 \begin{aligned}
 P:\qquad
 &\Theta\sim\operatorname{Unif}(\{\pm1\}^D),\\
 &S_\ell\mid\Theta
 \sim\operatorname{Unif}\!\left\{s\in\{\pm1\}^n:
             \mathbf 1^\top s=m\Theta_\ell\right\},
 \qquad \ell\in[D],
\end{aligned}
\label{eq:hard-distribution-P}
\end{equation}
Equivalently, column $\ell$ contains $(n+m\Theta_\ell)/2$ plus signs and
$(n-m\Theta_\ell)/2$ minus signs, uniformly permuted across its rows.  The
columns are conditionally independent, but entries within a column are not.
The table is sampled once before interaction.  Querying a row does not
resample it.

\paragraph{One-stage protocol and alignment event.}
Let $W\perp(\Theta,S)$ collect the random tape and all data independent of
$(\Theta,S)$ that are used in replies before stage completion.

For the information bound, consider the stronger protocol in which every query
also reveals its complete encoding row.  Let $\mathcal H_t$ be the history after
response $t$ in this stronger protocol.  $\mathcal H_0$ includes $W$ and one
optional initialization row whose index is fixed from $W$ in advance.  At
query $t\le B$, $\mathcal H_{t-1}$ determines an index $I_t\in[n]$ and an
alignment vector $V_t$ with $\|V_t\|\le1$; repeats reveal no new row.  This
strengthening is valid because every reply before stage completion is
measurable in $W$, the queries, and the selected rows.  The alignment condition
is evaluated before the response and is not revealed to the algorithm.

Let $K_t$ count the distinct rows returned before query $t$, including the
initialization row when present.  With $a=5/32$, define
\begin{equation}
 E_r=\left\{\exists t\le B:K_t<r,\quad
 \frac1{\sqrt D}\langle\Theta,V_t\rangle>a\right\}.
 \label{eq:stopped-discovery-event}
\end{equation}

Both one-stage results use the dimension condition
\begin{equation}
 D\ge512\log(B+2)+2048.
 \label{eq:explicit-row-dimension}
\end{equation}

\begin{lemma}[One-stage row lower bound for a tunable bias]
\label{lem:tunable-bias-kl}
Let $m$ have the parity of $n$, let $1\le m\le n/8$, set $\rho=m/n$, and
define
\begin{equation}
 r_\rho:=\left\lfloor
 \frac1{4096}\min\{n,\rho^{-2}\}
 \right\rfloor.
 \label{eq:tunable-row-threshold}
\end{equation}
If $r_\rho\ge2$, $B\ge1$, and $D$ satisfies
\eqref{eq:explicit-row-dimension}, then
\begin{equation}
 P(E_{r_\rho}\mid W)\le\frac18
 \qquad\text{almost surely}.
 \label{eq:restated-tunable-one-stage-conclusion}
\end{equation}
\end{lemma}

\begin{lemma}[Optimized one-stage row lower bound]
\label{lem:optimized-one-stage}
Choose the least parity-compatible $m\ge\sqrt n$.  Then
$\sqrt n\le m<\sqrt n+2$ and $\rho=m/n$.  For every $n\ge1024$, $B\ge1$,
and $D$ satisfying~\eqref{eq:explicit-row-dimension}, the sharper threshold
\begin{equation}
 r=\left\lfloor\frac n{512}\right\rfloor
 \label{eq:optimized-row-threshold}
\end{equation}
satisfies
\begin{equation}
 P(E_r\mid W)\le\frac18
 \qquad\text{almost surely}.
 \label{eq:restated-one-stage-conclusion}
\end{equation}
At the optimized bias, this sharpens the tunable threshold to
$r=\lfloor n/512\rfloor$ while preserving
$P(E_r\mid W)\le1/8$.
\end{lemma}

\begin{figure}[H]
\centering
\includegraphics[width=\linewidth]{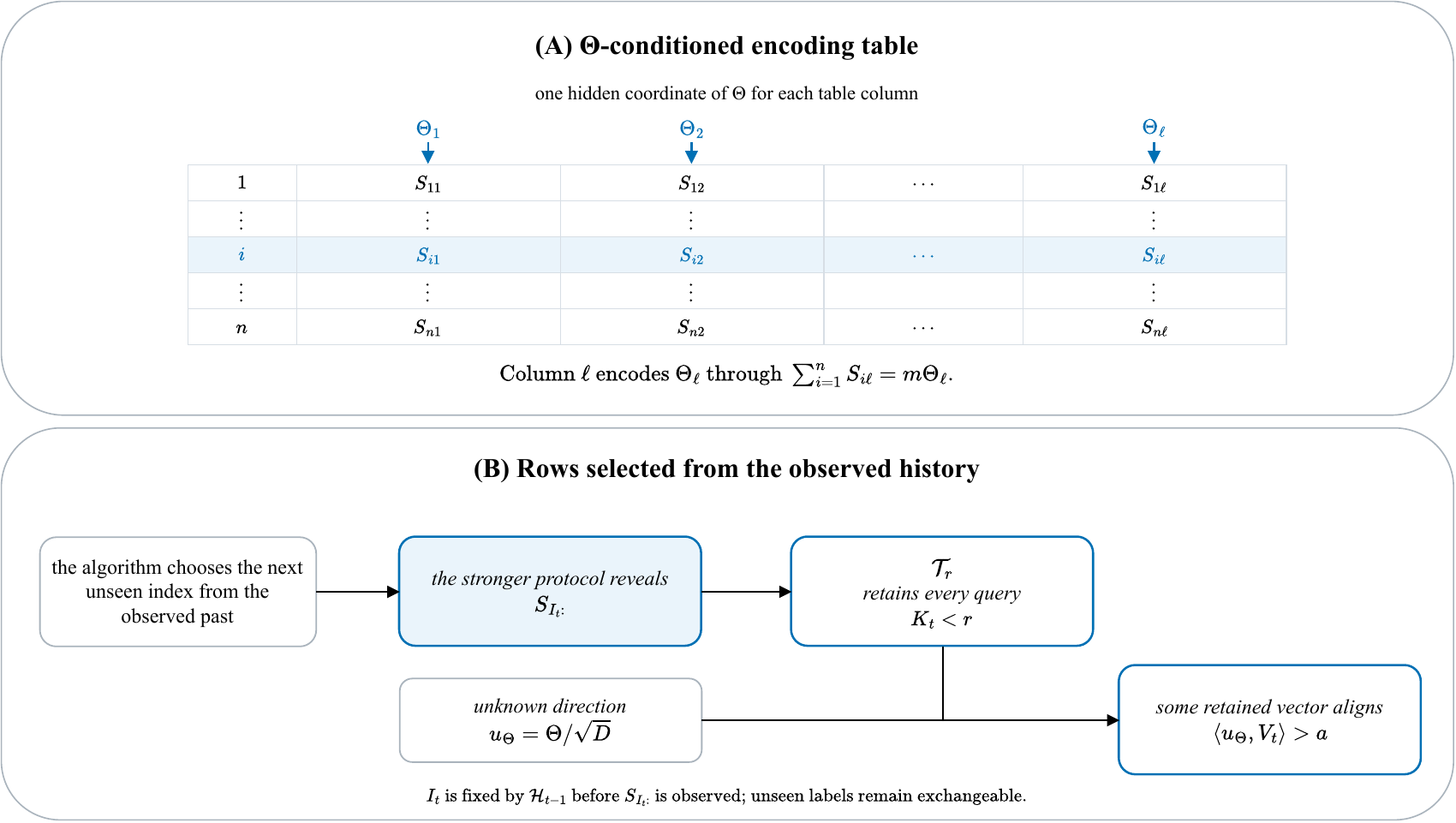}
\caption{One-stage encoding when row indices are selected from the observed
history.  Column $\ell$ of the fixed table encodes $\Theta_\ell$.  In the
stronger protocol used for the information bound, a new component index is
chosen before its row is observed, so exchangeability determines the
conditional law of the next unseen row.}
\label{fig:sign-table-direction}
\end{figure}

\paragraph{Information under sequential row selection.}
The one-stage argument has three steps: bound the information in one unseen
row, sum over the first $r$ distinct rows, and convert low mutual information
into a low probability of early alignment.
Throughout the proof of Lemma~\ref{lem:optimized-one-stage}, fix an admissible
value $W=w$.  All probabilities and mutual informations in this part are
conditional on $W=w$, and we suppress this conditioning.  The bounds are
uniform in $w$.

At every query, $I_t$ is fixed by $\mathcal H_{t-1}$ before its row is
observed.  If $I_t$ is new and $K_t=d$, Lemma~\ref{lem:encoding-table-properties}
therefore makes $S_{I_t:}$ a uniformly chosen remaining row conditional on
$(\Theta,\mathcal H_{t-1})$.  If $I_t$ is repeated, its row is already in the
history and contributes no new information.

Let $\mathcal T_r$ be this history, including the optional initialization row,
stopped immediately after the response that returns the
$r$th distinct row, or after query $B$ if fewer than $r$ rows are returned.
Suppose the $r$th row is returned.  The query selecting it is retained and has
$K_t=r-1$, and every later query has $K_t\ge r$.  If the $r$th row is never
returned, $\mathcal T_r$ contains all $B$ queries.  In both cases,
$\mathcal T_r$ retains exactly every query that can witness $E_r$, so $E_r$ is
measurable with respect to $(\Theta,\mathcal T_r)$.

\paragraph{Information in newly revealed rows.}
\begin{lemma}[Information in the first $r$ distinct rows]
\label{lem:first-use-row-information}
Let $1\le r<n$ and assume
\begin{equation}
 8m+9r\le n.
 \label{eq:row-information-interior}
\end{equation}
For the stopped history $\mathcal T_r$,
\begin{equation}
 \mathrm I(\Theta;\mathcal T_r)
 \le
 \frac{32D}{63}
 \sum_{d=0}^{r-1}
 \left(\frac{m}{n-d}\right)^2
 \le
 \frac{32D}{63}\,
 r\left(\frac{m}{n-r}\right)^2.
 \label{eq:first-use-row-information}
\end{equation}
\end{lemma}

\begin{proof}[Proof of Lemma~\ref{lem:first-use-row-information}]
\emph{Information in a newly queried row.}
Consider a query $t$ retained in $\mathcal T_r$ and condition on a realization
$h$ of $\mathcal H_{t-1}$.  Both $I_t$ and $K_t$ are then fixed.  If $I_t$ is
repeated, its row and the corresponding reply are already determined by $h$,
so the conditional information increment is zero.  Suppose instead that
$I_t$ is unseen and $K_t=d<r$.  For column $\ell$, let $\Sigma_{d,\ell}$ be the
sum of that coordinate over the $d$ distinct rows already revealed in $h$.
Conditional exchangeability leaves $n-d$ uniformly assigned signs with
residual sum $m\Theta_\ell-\Sigma_{d,\ell}$.  Hence
\begin{equation}
 p_{\theta,\ell}^{(d)}
 :=
 \Prob(S_{I_t\ell}=1
 \mid \Theta_\ell=\theta,\mathcal H_{t-1}=h)
 =
 \frac{n+m\theta-d-\Sigma_{d,\ell}}{2(n-d)}.
 \label{eq:conditional-unseen-sign}
\end{equation}
The optional initialization row obeys the same formula with $d=0$, because its
index is fixed by $W$ before the row is observed.

Expose the coordinates of the new row in the order $\ell=1,\ldots,D$.
Conditional independence of the column sequences gives
\[
 \Prob(S_{I_t\ell}=1
 \mid\Theta,\mathcal H_{t-1}=h,S_{I_t,<\ell})
 =p_{\Theta_\ell,\ell}^{(d)}.
\]
Thus, under this conditioning, the current sign depends on the hidden
direction only through $\Theta_\ell$.  The argument does not require
independence among the posterior coordinates of $\Theta$.
More explicitly, conditional on $\mathcal H_{t-1}=h$ and
$S_{I_t,<\ell}$,
\[
 \Theta\longrightarrow\Theta_\ell\longrightarrow S_{I_t\ell}
\]
is a Markov chain.  Therefore,
\[
 \mathrm I(\Theta;S_{I_t\ell}
 \mid\mathcal H_{t-1}=h,S_{I_t,<\ell})
 =
 \mathrm I(\Theta_\ell;S_{I_t\ell}
 \mid\mathcal H_{t-1}=h,S_{I_t,<\ell}).
\]

Condition~\eqref{eq:row-information-interior} implies
\[
 \left|p_{\theta,\ell}^{(d)}-\frac12\right|
 \le
 \frac{m+d}{2(n-d)}
 \le\frac1{16}.
\]
Hence both $p_{+,\ell}^{(d)}$ and $p_{-,\ell}^{(d)}$ lie in
$[7/16,9/16]$.  For $Y\sim\operatorname{Bernoulli}(p)$, its Shannon entropy is
$H(Y)=h_{\rm b}(p):=-p\log p-(1-p)\log(1-p)$.  Here $\log$ is the natural
logarithm, so entropy and mutual information are measured in nats.  On the
preceding interval,
\[
 -h_{\rm b}''(p)=\frac1{p(1-p)}
 \le\frac{256}{63}
 \qquad (p\in[7/16,9/16]).
\]

Fix also a realization of $S_{I_t,<\ell}$ and write
\[
 \pi:=\Prob(\Theta_\ell=1\mid
   \mathcal H_{t-1}=h,S_{I_t,<\ell}),
 \qquad
 \bar p:=\pi p_{+,\ell}^{(d)}+(1-\pi)p_{-,\ell}^{(d)}.
\]
Under this conditioning, $S_{I_t\ell}$ is Bernoulli with parameter $\bar p$;
after additionally conditioning on $\Theta_\ell=\theta$, its parameter is
$p_{\theta,\ell}^{(d)}$.  Thus the first conditional entropy below is
$h_{\rm b}(\bar p)$, while the entropy after conditioning on $\Theta_\ell$ is
the posterior average of $h_{\rm b}(p_{+,\ell}^{(d)})$ and
$h_{\rm b}(p_{-,\ell}^{(d)})$.  The identity
$\mathrm I(X;Y\mid Z)=H(Y\mid Z)-H(Y\mid X,Z)$ gives
\begin{align}
 \mathrm I(
 \Theta;S_{I_t\ell}
 \mid \mathcal H_{t-1}=h,S_{I_t,<\ell})
 &=h_{\rm b}(\bar p)
   -\pi h_{\rm b}\!\left(p_{+,\ell}^{(d)}\right)
   -(1-\pi)h_{\rm b}\!\left(p_{-,\ell}^{(d)}\right)
   \notag\\
 &\le \frac12\frac{256}{63}\,
   \pi(1-\pi)
   \left(p_{+,\ell}^{(d)}-p_{-,\ell}^{(d)}\right)^2
   \notag\\
 &\le\frac{32}{63}
   \left(\frac{m}{n-d}\right)^2.
 \label{eq:one-coordinate-mutual-information}
\end{align}
The first inequality is Taylor's theorem for the concave function $h_{\rm b}$:
its Jensen gap is at most one half of the curvature bound times the variance of
the two-point parameter.  The last inequality uses
\[
 \pi(1-\pi)\le\frac14,
 \qquad
 p_{+,\ell}^{(d)}-p_{-,\ell}^{(d)}=\frac{m}{n-d}.
\]

\emph{Chain rule along the actual history.}
Because the coordinate bound is uniform in $S_{I_t,<\ell}$, the chain rule
within the new row gives
\begin{align}
 \mathrm I(\Theta;S_{I_t:}\mid\mathcal H_{t-1}=h)
 &=\sum_{\ell=1}^D
   \mathrm I(\Theta;S_{I_t\ell}
   \mid\mathcal H_{t-1}=h,S_{I_t,<\ell})
   \notag\\
 &\le \frac{32D}{63}\left(\frac{m}{n-d}\right)^2.
 \label{eq:one-new-row-information}
\end{align}
Pad the stopped history after its stopping time with deterministic null
query--response symbols to a total of $B$ rounds.  The chain rule then has one
increment in conditional mutual information per round.  The chosen query is
$\mathcal H_{t-1}$-measurable and therefore adds no information by itself; a
response from a repeated row, or a null response, is already determined by the
preceding history and also adds zero.
The optional initialization row, when present, is the $d=0$ increment before
the first query.  Each response from a new row increases the number of distinct rows, so
every $d=0,\ldots,r-1$ contributes at most once.  (Without an initialization
row, the first response from a new row corresponds to the $d=0$ increment.)  Applying
\eqref{eq:one-new-row-information} at these increments gives
\begin{align*}
 \mathrm I(\Theta;\mathcal T_r)
 &\le \frac{32D}{63}
 \sum_{d=0}^{r-1}\left(\frac{m}{n-d}\right)^2
 \le \frac{32D}{63}\,
 r\left(\frac{m}{n-r}\right)^2.
\end{align*}
Conditioning on $S_{I_t,<\ell}$ in
\eqref{eq:one-coordinate-mutual-information} makes posterior independence
unnecessary.  Finally, $n-d\ge n-r$ for $d\le r-1$.  This
proves~\eqref{eq:first-use-row-information}.
\end{proof}

\paragraph{From row information to early alignment.}
For any $r$ above, define the product distribution
\[
 \overline P_r:=P_\Theta\otimes P_{\mathcal T_r}.
\]
Under $\overline P_r$, every alignment vector belonging to a query retained in
$\mathcal T_r$ is independent of $\Theta$.  Conditional on any retained vector
$v$ with $\|v\|\le1$, Hoeffding's inequality gives
\[
 \overline P_r\!\left(
  \frac1{\sqrt D}\langle\Theta,v\rangle>a\ \middle|\ v
 \right)
 \le \exp\!\left(-\frac{a^2D}{2}\right).
\]
There are at most $B$ retained queries, so a union bound gives
\begin{equation}
q_r:=\overline P_r(E_r)
 \le B\exp\!\left(-\frac{a^2D}{2}\right)
 =B\exp\!\left(-\frac{25D}{2048}\right)
 <\frac18,
 \label{eq:product-alignment-tail}
\end{equation}
where the final inequality follows because
$D\ge512\log(B+2)+2048$ as required in
\eqref{eq:explicit-row-dimension}.

We next transfer this bound to the hard distribution.  If $q_r=0$, absolute continuity of
$P_{\Theta,\mathcal T_r}$ with respect to
$P_\Theta\otimes P_{\mathcal T_r}$ gives
$P(E_r)=0$.  Hence assume below that $q_r>0$.  Since $E_r$ is measurable from
$(\Theta,\mathcal T_r)$, apply the map
$(\Theta,\mathcal T_r)\mapsto\mathbf 1_{E_r}$.  The KL divergence between the
joint law and the product law equals mutual information, so data processing
gives the binary divergence
\[
 \mathrm I(\Theta;\mathcal T_r)
 \ge d_{\rm bin}\!\left(P(E_r)\middle\|q_r\right).
\]
For fixed $q_r<1/8$, the derivative
\[
 \frac{\partial}{\partial p}d_{\rm bin}(p\|q_r)
 =
 \log\frac{p(1-q_r)}{q_r(1-p)}
\]
is positive for every $p\ge1/8$.  Therefore, if $P(E_r)>1/8$,
\begin{align}
 \mathrm I(\Theta;\mathcal T_r)
 &\ge d_{\rm bin}\!\left(\frac18\middle\|q_r\right)\nonumber\\
 &\ge
 \frac{25}{16384}D
 -\frac18\log B
 -h_{\rm b}\!\left(\frac18\right).
 \label{eq:mutual-information-event-lower}
\end{align}
The last step follows by expanding the binary divergence:
\[
 d_{\rm bin}\!\left(\frac18\middle\|q\right)
 =\frac18\log\frac1q-h_{\rm b}\!\left(\frac18\right)
  -\frac78\log(1-q)
 \ge \frac18\log\frac1q-h_{\rm b}\!\left(\frac18\right),
\]
and then substituting~\eqref{eq:product-alignment-tail}.

\begin{proof}[Proof of Lemma~\ref{lem:tunable-bias-kl}]
Because $r_\rho\ge2$ and the coefficient in
\eqref{eq:tunable-row-threshold} is $1/4096$,
\[
 n\ge8192,\qquad
 \rho\le\frac1{\sqrt{8192}}<\frac1{32},
 \qquad
 r_\rho\le\frac n{4096}.
\]
Thus $8m+9r_\rho<n/4+9n/4096<n$, so
Lemma~\ref{lem:first-use-row-information} applies.  Moreover,
$r_\rho\rho^2\le1/4096$, and hence
\begin{align}
 \mathrm I(\Theta;\mathcal T_{r_\rho})
 &\le
 \frac{32D}{63\cdot4096}
 \left(\frac{4096}{4095}\right)^2
 =
 \frac{131072}{1056448575}D
 <\frac D{8000}.
 \label{eq:tunable-mutual-information-upper}
\end{align}

Suppose, for contradiction, that $P(E_{r_\rho})>1/8$.  Then
\eqref{eq:mutual-information-event-lower} gives
\begin{equation*}
 \mathrm I(\Theta;\mathcal T_{r_\rho})
 \ge \frac{25D}{16384}-\frac18\log B
      -h_{\rm b}\!\left(\frac18\right).
\end{equation*}

The numerical margins satisfy
\[
 \frac{25}{16384}-\frac1{8000}>\frac1{1024}
\]
and
\[
 \frac D{1024}
 \ge\frac12\log(B+2)+2
 >
 \frac18\log B+h_{\rm b}\!\left(\frac18\right),
\]
Together these inequalities make the lower bound contradict
\eqref{eq:tunable-mutual-information-upper}.  Therefore
$P(E_{r_\rho})\le1/8$.
\end{proof}

\begin{proof}[Proof of Lemma~\ref{lem:optimized-one-stage}]
Assume $n\ge1024$ and use the optimized choice
$m<\sqrt n+2\le(17/16)\sqrt n$.  With
$r=\lfloor n/512\rfloor$,
\[
 m\le\frac{17n}{512},
 \qquad
 8m+9r<\frac{145n}{512}<n.
\]
Lemma~\ref{lem:first-use-row-information} gives
\begin{align}
 \mathrm I(\Theta;\mathcal T_r)
 &\le
 \frac{32D}{63}\,
 \frac n{512}
 \left[
 \frac{(17/16)\sqrt n}{n(511/512)}
 \right]^2
 =
 \frac{18496}{16450623}D.
 \label{eq:optimized-mutual-information-upper}
\end{align}

Suppose, for contradiction, that $P(E_r)>1/8$.  Then
\begin{equation*}
 \mathrm I(\Theta;\mathcal T_r)
 \ge \frac{25D}{16384}-\frac18\log B
      -h_{\rm b}\!\left(\frac18\right).
\end{equation*}

The coefficient margin satisfies
\[
 \frac{25}{16384}
 -\frac{18496}{16450623}
 >\frac1{2500},
\]
whereas
\[
 \frac D{2500}
 \ge\frac{512}{2500}\log(B+2)+\frac{2048}{2500}
 >
 \frac18\log B+h_{\rm b}\!\left(\frac18\right).
\]
This is a contradiction.  Thus $P(E_r)\le1/8$, completing the proof of
Lemma~\ref{lem:optimized-one-stage}.
\end{proof}

\begin{lemma}[Coupling up to the first stage completion]
\label{lem:precompletion-extension}
For an independent $(\Theta,S)$, let $W\perp(\Theta,S)$ contain the history before
initialization, an initialization index fixed in advance, the random tape, and
all other data needed for replies.  Suppose each IFO reply before completion,
including its value and all gradient blocks, has the form
\begin{equation}
 Y_t^{\rm pre}=\Phi_t\!\left(
 W,(I_s,X_s)_{s\le t},(S_{I_s:})_{s\le t}
 \right),
 \label{eq:precompletion-reply-form}
\end{equation}
where $\Phi_t$ does not depend on the alignment condition and reveals neither
$\Theta$ nor data from the next stage.  At query $t$, suppose completion requires
$D^{-1/2}\langle\Theta,V_t\rangle>a$ before the response, with $V_t$ fixed by
the preceding history and query.  A length-$B$ continuation can then be coupled so
that
\begin{equation}
 \left\{
 \exists t\le B:\ \text{completion at $t$},\ K_t<r
 \right\}
 \subseteq E_r
 \label{eq:first-alignment-containment}
\end{equation}
and the histories and queries agree through the first completion query.
Consequently, whenever $P(E_r\mid W)\le p$ almost surely,
\begin{equation}
 P\!\left(\exists t\le B:\ \text{completion at $t$},\ K_t<r\mid W\right)
 \le p
 \qquad\text{almost surely}.
 \label{eq:precompletion-probability-bound}
\end{equation}
\end{lemma}

\begin{proof}
Construct a $B$-query continuation that remains at the same stage.  After
alignment, keep using the same maps $\Phi_t$ and the same fixed table, and set
$V_t=0$ after the $r$th distinct row is returned.  Couple the two
continuations using the same $W$, stage data, and optional initialization row.

Until stage completion or the return of the $r$th distinct row, the coupled
interactions have the same history before each query and hence choose the same
query; every earlier reply that does not complete the stage agrees by
\eqref{eq:precompletion-reply-form}.  If completion occurs with $K_t<r$, that
query is made before the $r$th distinct row is returned and therefore belongs
to $E_r$.  The completion event is determined before that response, so the
response is not needed for the inclusion.
\end{proof}

\subsection{Sequential revelation}

\paragraph{Process and notation.}
At stage $j$, write
\begin{equation}
 \operatorname{align}_j(x_j):=\frac1{\sqrt D}
 \langle\Theta^{(j)},V_j(x_j)\rangle,
 \qquad \|V_j(x_j)\|\le1.
 \label{eq:generic-radial-alignment}
\end{equation}
This is $q_j$ for the constant-scale chain, $u_j$ for the weighted chain, and
$h_j$ for the two-scale chain.  Stage $j$ completes when $\operatorname{align}_j>a$ first
holds, checked before the response.
Assign each IFO call to the unique stage that is current when its query is
made.  The response completing stage $j-1$ may also include one stage-$j$
initialization row, but this costs no additional call and is not assigned again.
Let $\nu_t$ be the number of stages completed after response $t$, with
$\nu_0=0$.  The revelation process advances by at most one stage per response;
premature alignment at a later stage is included in the event below.
For an $N$-query interaction, define
\begin{equation}
\begin{aligned}
 \cEfuture
 :={}&\{\exists t\le N,\ k>\nu_{t-1}+1:
       \operatorname{align}_k(X_{t,k})>a\}\\
 &{}\cup\{\exists k>\nu_N+1:
       \operatorname{align}_k(\widehat X_k)>a\}.
\end{aligned}
\label{eq:future-alignment-event}
\end{equation}

After stage $j$ is completed, we reveal $\Theta^{(j)}$ and one stage-$(j+1)$
initialization row as part of the \emph{augmented interaction}, without charging
another IFO call.  The \emph{truncated interaction} retains this revealed
information but omits links beyond the completed stages and the current stage.
Crossing $a$ leaves the region where
$G=G'=0$; full opening occurs only at $b$.

\proofguidespace
\begin{theorem}[Sequential revelation: row costs add]
\label{thm:sequential}
Fix the first $J$ stages, a budget of $N$ IFO calls, and constants
$r\ge2$, $p\in[0,1/2)$, and $\delta\in[0,1-2p)$.  Assume:
\begin{enumerate}[label=(S\arabic*),leftmargin=*,itemsep=1pt,topsep=2pt]
 \item conditional on the history immediately before its initialization row is
 returned, or on the initial history for stage $1$, every stage once it becomes current
 completes after fewer than $r$ calls assigned to it with probability at most $p$;
 \item the event $\cEfuture$ satisfies $\Prob(\cEfuture)\le\delta$; and
 \item on $\cEfuture^c$, the augmented and truncated interactions have identical
 histories and final outputs.
\end{enumerate}
If $N\le rJ/2$, then
\begin{equation}
 \Prob(\nu_N\ge J)\le2p,
 \qquad
 \Prob(\nu_N<J,\ \cEfuture^c)\ge1-2p-\delta.
 \label{eq:restated-sequential-conclusion}
\end{equation}
On the second event, the output of the original IFO interaction satisfies
$\operatorname{align}_k(\widehat X_k)\le a$ for every $k>J$.
\end{theorem}

\paragraph{Conditions for applications to radial chains.}
The optimized corollary below requires:
\begin{enumerate}[label=(R\arabic*),leftmargin=*,itemsep=2pt,topsep=2pt]
 \item The pairs $\cZ^{(j)}=(\Theta^{(j)},S^{(j)})$ are independent samples from $P$
 with the least parity-compatible $m\ge\sqrt n$ and $a=5/32$.
 \item Immediately before an optional initialization row is returned, let $W$
 contain its index (fixed in advance), the history, seed, and data from
 completed stages used in replies.  Then $W\perp\cZ^{(j)}$, and every full IFO
 reply that does not complete the stage has the form
 \[
  \Phi_t\!\left(W,(I_s,X_s)_{s\le t},
                    (S^{(j)}_{I_s:})_{s\le t}\right),
 \]
 including its value and all gradient blocks.  The map does not depend on the
 alignment condition and reveals neither $\Theta^{(j)}$ nor data from later stages.
 \item Completion uses~\eqref{eq:generic-radial-alignment}.  If the current stage and
 all stages not yet current have alignment at most $a$, every omitted link has
 $G=G'=0$.  The final output is included without an oracle call.
 \item Assign each call to the stage current when its query is made.  A
 completion response may include the next initialization row without an
 additional IFO call and is assigned only to the stage it completes.  Before
 the $q$th call assigned to
 stage $j$, at most its initialization row and $q-1$ response rows have been
 seen.
\end{enumerate}

\begin{corollary}[Sequential revelation for the optimized radial chains]
\label{cor:optimized-radial-chain}
Consider a radial chain satisfying \emph{(R1)--(R4)} above.  Let $n\ge1024$,
$1\le J<M_{\rm tot}$, and set
\begin{equation}
 r=\left\lfloor\frac n{512}\right\rfloor.
\end{equation}
If $N\ge1$, $N\le rJ/2$, and
\begin{equation}
 D\ge\max\left\{
 2048[1+\log(N+3)],\quad
 \frac{2048}{25}\log\bigl(128(N+1)M_{\rm tot}\bigr)
 \right\},
 \label{eq:shared-radial-dimension}
\end{equation}
then
\begin{equation}
 \Prob(\nu_N\ge J)\le\frac14,
 \qquad
 \Prob(\nu_N<J,\ \cEfuture^c)\ge\frac{95}{128}.
 \label{eq:shared-radial-revelation}
\end{equation}
On the second event, the original IFO output satisfies
$\operatorname{align}_k(\widehat X_k)\le a$ for every $k>J$.
\end{corollary}

The lemmas below establish stage independence and exact coupling; the counting
argument in the next subsection then proves the corollary.

\begin{figure}[H]
\centering
\includegraphics[width=\linewidth]{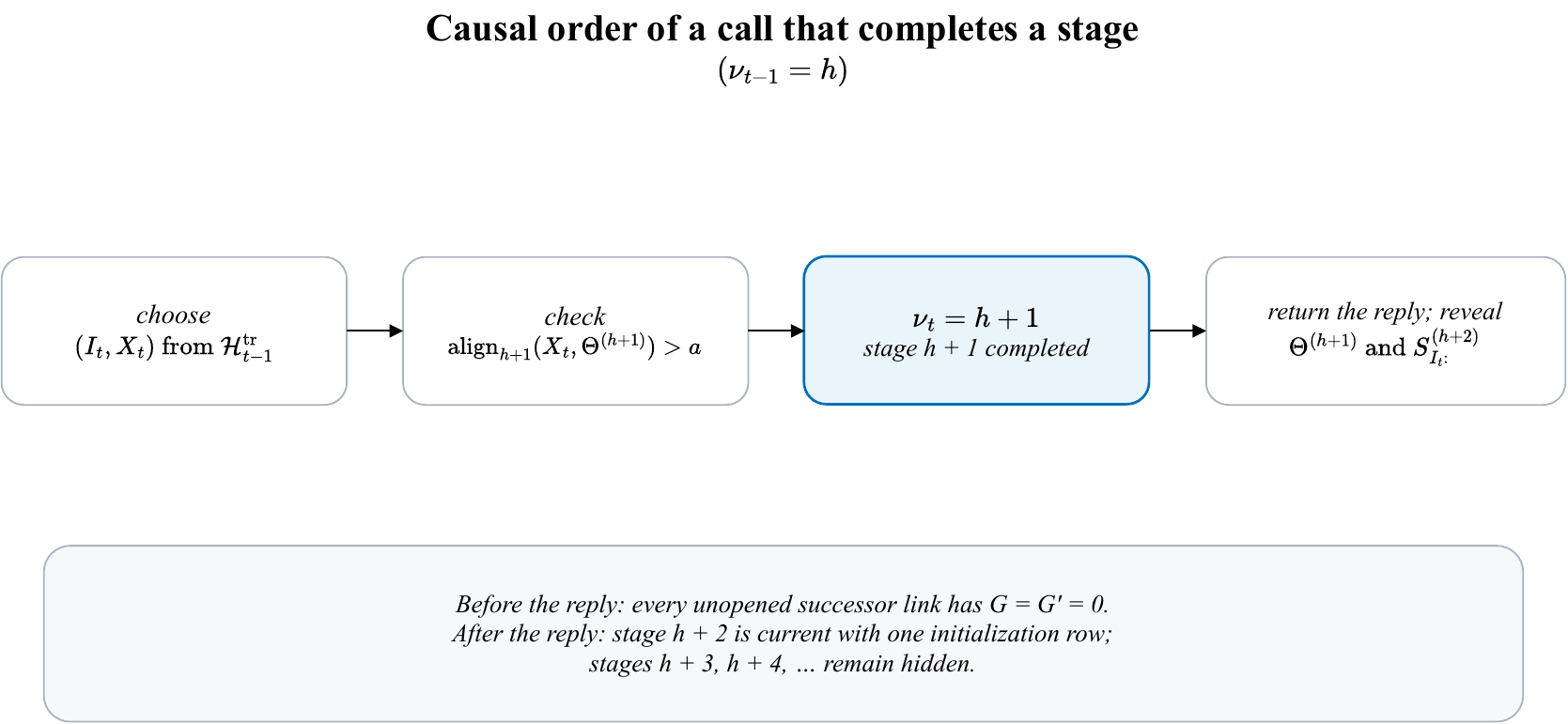}
\caption{The completion condition is evaluated on a query fixed from the preceding history.
Only afterward can the augmented interaction reveal the completed direction
and an initialization row for the next stage.}
\label{fig:stage-unlocking}
\end{figure}

The complete history of the truncated interaction after response $t$ is denoted
by $\mathcal H_t^{\rm tr}$ and includes the information returned with completion
responses and the private seed
$\xi$.  Including
the output in $\cEfuture$ reveals nothing to the algorithm and leaves $\nu_N$
unchanged.

\begin{lemma}[Independence before stage revelation and control of future alignment]
\label{lem:future-alignment}
Under \emph{(R1)--(R4)}, the history of the truncated interaction before a
stage's initialization row is independent of that stage.  Thus a stage, once
current, completes after fewer than $r$ calls assigned to it with conditional
probability at most $p$ whenever the one-stage bound has parameters $(r,p)$.
At query $t$ with $\nu_{t-1}=h$, every $\Theta^{(k)}$, $k>h+1$, is independent of
$X_{t,k}$; the same holds for $\widehat X_k$.  Consequently, for a chain of
$M_{\rm tot}$ stages,
\begin{equation}
 D\ge\frac{2048}{25}
 \log\bigl(128(N+1)M_{\rm tot}\bigr)
 \quad\Longrightarrow\quad
 \Prob(\cEfuture)\le\frac1{128}.
 \label{eq:future-alignment-dimension}
\end{equation}
This joint probability bound covers every query paired with every future stage,
as well as the final output.  Replacing $128$ above by $1/\delta_{\rm f}$ gives any
target $\delta_{\rm f}\in(0,1)$, with only logarithmic growth in $D$.
\end{lemma}

\begin{lemma}[The augmented and truncated interactions agree exactly]
\label{lem:pathwise-coupling}
Under \emph{(R1)--(R4)}, for the same fixed instance and seed, on
$\cEfuture^c$ the two interactions have identical component values, all
gradient blocks, revealed data, queries, and final outputs, pathwise.  Here
$\cEfuture$ is evaluated in the truncated interaction.
\end{lemma}

\begin{proof}[Proof of Lemma~\ref{lem:future-alignment}]
\emph{Independence before revelation.}
Maintain the filtration invariant
\begin{equation}
 \{\nu_t=g\}
 \quad\Longrightarrow\quad
 \sigma(\mathcal H_t^{\rm tr})\subseteq
 \sigma(\cZ^{(1)},\ldots,\cZ^{(g+1)},\xi).
 \label{eq:active-filtration-invariant}
\end{equation}
If $\nu_{t-1}=h$, the query uses only
$\cZ^{(1)},\ldots,\cZ^{(h+1)}$ and $\xi$.  A reply that does not complete the
current stage reveals nothing from the next stage; a completion reply reveals
only one initialization row from $\cZ^{(h+2)}$.  This
proves the invariant and leaves every $\cZ^{(k)}$, $k>\nu_t+1$, independent of
the observed history.

\emph{Conditional one-stage bound.}
For $j\ge2$, condition on the history immediately before its initialization
row is returned, including the fixed query, row index, seed, and reply data
from completed stages.
The invariant gives $W\perp\cZ^{(j)}$, while $G=G'=0$ below the lower threshold
puts the full reply in
\eqref{eq:precompletion-reply-form} with rows $S^{(j)}_{I_s:}$.
Lemma~\ref{lem:precompletion-extension}
therefore applies; for $j=1$, use the initial data and no initialization row.

Because alignment is checked before the response, completion on the $q$th call
assigned to that stage gives
\[
 K_t\le q-1+\mathbf 1\{j\ge2\},
 \qquad q<r\Longrightarrow K_t<r.
\]
Thus completion after fewer than $r$ assigned calls is contained in $E_r$ and has
conditional probability at most $\Prob(E_r\mid W)\le p$.

Each direction at a stage not yet current is independent of its query or output
block.  A
Rademacher tail bound over at most $(N+1)M_{\rm tot}$ pairs $(t,k)$ gives
\[
 \Prob(\cEfuture)
 \le (N+1)M_{\rm tot}e^{-a^2D/2}
 \le\frac1{128}.
\]
\end{proof}

\begin{proof}[Proof of Lemma~\ref{lem:pathwise-coupling}]
We argue by induction, starting from a common history before the query with
progress $h$.  If the stage
does not complete, $\operatorname{align}_{h+1}\le a$; after completion,
$\cEfuture^c$ also controls the new successor.
Every omitted link therefore has $G=G'=0$, so its scalar value and every
gradient contribution vanish.  The replies, revealed data, histories, and next
queries agree, and the common seed gives the same output.  Completed links
remain in both interactions, even if a completed coordinate later falls below
$b$; the output term in $\cEfuture^c$ is used only to control the final output.
\end{proof}

\subsection{Counting calls across stages}

Each IFO call is assigned to the stage that is current when its query is made.  A
completion response may supply the next initialization row without an
additional IFO call, but the same call is assigned only to the stage it
completes.  Hence calls assigned to
different stages are disjoint.

\begin{proof}[Proof of Theorem~\ref{thm:sequential}]
Let $Z$ count the stages among $1,\ldots,J$ that are reached and finish in
fewer than $r$ calls assigned to them.  Conditioning when stage $j$ is reached,
\emph{(S1)} and the tower property give $\E Z\le pJ$.

If all first $J$ stages are completed by $N\le rJ/2$ calls, then
necessarily $Z\ge J/2$: otherwise more than $J/2$ stages would each use at
least $r$ calls, exceeding the budget.  Thus Markov's inequality, without
assuming independence across stages, gives
\begin{equation}
\Prob(\nu_N\ge J)
\le\Prob\!\left(Z\ge\frac J2\right)
\le2p.
\end{equation}

The assumed bound on $\cEfuture$ and a union bound yield
\begin{equation}
 \Prob\bigl(\nu_N<J,\ \cEfuture^c\bigr)
 \ge1-2p-\delta.
\end{equation}

On this event, \emph{(S3)} makes the auxiliary histories and outputs
identical; the augmented interaction only adds information, so the conclusion
transfers to the original IFO.

It remains to check the final output.  On the joint event, let $h=\nu_N<J$.
Every $k>J$ also satisfies $k>h+1$;
the output term in the definition of $\cEfuture$ therefore gives
$\operatorname{align}_k(\widehat X_k)\le a$.
\end{proof}

\begin{proof}[Proof of Corollary~\ref{cor:optimized-radial-chain}]
The first dimension term, Lemma~\ref{lem:optimized-one-stage}, and
Lemma~\ref{lem:precompletion-extension} give $p=1/8$ for every stage once it
becomes current,
because completion in $q<r$ calls assigned to that stage uses fewer than $r$
distinct rows, including one optional initialization row.  The second dimension requirement
and Lemma~\ref{lem:future-alignment} give $\delta=1/128$ jointly for all $N$
queries and the final output.  The condition $G=G'=0$ for future links invokes
Lemma~\ref{lem:pathwise-coupling}; the disjoint call assignments verify the
remaining hypothesis of Theorem~\ref{thm:sequential}, which now yields
\eqref{eq:shared-radial-revelation} and the claimed bound on the original IFO
output.
\end{proof}

\subsection{Auxiliary quadratic for the additive \texorpdfstring{$n$}{n} term}

This subsection proves the standalone $\Omega(n)$ IFO lower bound that gives
the additive $n$ term in Theorems~\ref{thm:nonconvex-main}
and~\ref{thm:pl-main}.
The sequential argument above uses the optimized bias.  The separate additive
$n$ term uses instead the smallest positive parity-compatible bias and the
fixed dimension $D_0=4$.

\begin{lemma}[Auxiliary quadratic lower bound]
\label{lem:dense-quadratic-row-bound}
Fix $n\ge n_0=1024$ and $0\le N\le\lfloor n/4\rfloor$.  Set
\begin{equation}
 m_0:=
 \begin{cases}
  1,&n\ \text{odd},\\
  2,&n\ \text{even},
 \end{cases}
 \qquad
 \rho_0:=\frac{m_0}{n},
 \qquad
 D_0:=4.
 \label{eq:low-bias-quadratic-parameters}
\end{equation}
Draw $(\Theta,S)$ from $P$ in~\eqref{eq:hard-distribution-P} with
$(m,D)=(m_0,D_0)$, and write
$u=\Theta/\sqrt{D_0}$ and $s_i=S_{i:}/\sqrt{D_0}$.  For any
$\gamma,\tau>0$, define
\begin{equation}
 f_i^{\rm quad}(x):=\frac\gamma2\|x\|^2
 -\frac\tau{\rho_0}\langle s_i,x\rangle,
 \qquad
 F_{\rm quad}(x):=\frac1n\sum_{i=1}^nf_i^{\rm quad}(x)
 =\frac\gamma2\|x\|^2-\tau\langle u,x\rangle.
\end{equation}
Every randomized IFO algorithm making at most $N$ exact calls and returning
$\widehat X$ satisfies
\begin{equation}
 \Prob\bigl(\|\gamma\widehat X-\tau u\|\le\tau/4\bigr)
 \le\frac5{16}.
 \label{eq:low-bias-quadratic-success}
\end{equation}
\end{lemma}

\begin{proof}
Each response reveals the selected row because
\begin{equation}
 \nabla f_i^{\rm quad}(x)=\gamma x-\frac\tau{\rho_0}s_i,
 \qquad
 s_i=\frac{\rho_0}\tau
 \bigl(\gamma x-\nabla f_i^{\rm quad}(x)\bigr),
\end{equation}
and, given $x$, the row determines both the value and gradient.

We next bound the information in all revealed rows.  Let
$\mathcal T_N^{\rm quad}$ be the complete interaction history, including the private seed
$\xi$, and let $\mathcal H^{\rm quad}_{t-1}$ be the history before query $t$.
The display above shows that every response determines its selected encoding
row.  Conditional on $\mathcal H^{\rm quad}_{t-1}$, the index $I_t$ is fixed;
querying a previously selected component reveals no new encoding row.  When a
new row is queried with $K_t=d$, let $\Sigma_{d,\ell}$ be the
sum in column $\ell$ over the $d$ distinct rows already revealed.  Conditional
exchangeability gives
\begin{equation}
 p_{\theta,\ell}^{(d)}
 =\frac{n+m_0\theta-d-\Sigma_{d,\ell}}{2(n-d)}.
\end{equation}
Since
$m_0\le2$, $d<N\le n/4$, and $n\ge1024$,
\begin{equation}
 \left|p_{\theta,\ell}^{(d)}-\frac12\right|
 \le\frac{m_0+d}{2(n-d)}
 \le\frac3{16}.
\end{equation}
Thus the conditional Bernoulli parameters lie in $[5/16,11/16]$, where
$-h_{\rm b}''\le256/55$.  The entropy identity and Taylor bound used in
\eqref{eq:one-coordinate-mutual-information} now give, for each coordinate,
\[
 \mathrm I(\Theta;S_{I_t\ell}\mid
   \mathcal H^{\rm quad}_{t-1},S_{I_t,<\ell})
 \le \frac12\frac{256}{55}\frac14
       \left(\frac{m_0}{n-d}\right)^2
 =\frac{32}{55}\left(\frac{m_0}{n-d}\right)^2.
\]
Summing over the $D_0$ coordinates yields
\begin{equation}
 \mathrm I(\Theta;S_{I_t:}\mid\mathcal H^{\rm quad}_{t-1})
 \le\frac{32D_0}{55}\left(\frac{m_0}{n-d}\right)^2.
\end{equation}
Since $\xi\perp(\Theta,S)$, queries are predictable from their histories, and
each value of $d$ occurs at most once when a new row is queried, the
chain rule for mutual information gives
\begin{equation}
 \mathrm I(\Theta;\mathcal T_N^{\rm quad})
 \le\frac{32D_0N}{55}
       \left(\frac{m_0}{n-N}\right)^2
 \le\frac{2048}{495n}
 \le\frac2{495}<\frac1{200}.
 \label{eq:low-bias-quadratic-information}
\end{equation}

Finally, the $16$ vectors $u_\theta=\theta/2$ have pairwise distance at least one.  The
nearest-codeword decoder
\begin{equation}
 \widehat\Theta_{\rm dec}
 \in\operatorname*{argmin}_{\vartheta\in\{\pm1\}^{D_0}}
 \left\|\frac{\gamma\widehat X}{\tau}-\frac\vartheta2\right\|
\end{equation}
is therefore correct whenever
$\|\gamma\widehat X/\tau-u_\Theta\|\le1/4$.  Fano's inequality
\citep{yu1997assouad} and
$\log2>1/2$ give
\begin{equation}
 \Prob(\widehat\Theta_{\rm dec}=\Theta)
 \le\frac{\mathrm I(\Theta;\mathcal T_N^{\rm quad})+\log2}{\log16}
 <\frac14+\frac1{400}<\frac5{16},
\end{equation}
which proves~\eqref{eq:low-bias-quadratic-success}.
\end{proof}

\section{Nonconvex First-Order Lower Bound}
\label{app:nonconvex}

\paragraph{Restatement of Theorem~\ref{thm:nonconvex-main}.}
There are universal constants $c_0,c_1>0$ such that the following holds.
Fix $n\ge1024$, $L_{\max},\Delta>0$, and
$0<\varepsilon^2\le c_0\Delta L_{\max}$.  For every randomized IFO algorithm
making at most
\begin{equation}
 N\le c_1\left(n+\sqrt n\,
       \frac{\Delta L_{\max}}{\varepsilon^2}\right)
 \label{eq:restated-nonconvex-budget}
\end{equation}
calls, there exist a dimension $d$ and a finite sum
$F\in\mathcal F_{\rm ind}^{n,d}(\Delta,L_{\max})$ such that
\begin{equation}
 \Pr_\xi\!\left(\|\nabla F(\widehat X)\|\ge4\varepsilon\right)
 \ge\frac{11}{16}.
 \label{eq:restated-nonconvex-failure}
\end{equation}
The algorithm may choose arbitrary queries from its observed history and return
any output measurable with respect to that history.

\subsection{Construction}

Let
\begin{equation}
 \mathcal R_n:=\left\{\frac mn:
 m\in\mathbb Z_{>0},\quad m\equiv n\pmod2,\quad m\le\frac n8\right\}.
 \label{eq:nonconvex-local-bias-grid}
\end{equation}
Fix $\rho=m/n\in\mathcal R_n$, $T\ge1$, and write
$x=(x_1,\ldots,x_T)$ with $x_j\in\R^D$.  At stage
$j$, let $\Theta^{(j)}\in\{\pm1\}^D$ be the hidden direction and let
$S^{(j)}\in\{\pm1\}^{n\times D}$ be its fixed encoding table with bias
$\rho=m/n$; row $S^{(j)}_{i:}$ belongs to component $i$.  The table satisfies
$n^{-1}\sum_{i=1}^nS^{(j)}_{i:}=\rho\Theta^{(j)}$.  For $z\in\R^D$, define
\begin{equation}
 \Psi(z)=\frac{z}{\sqrt{1+\|z\|^2}},\qquad
 q_j(z)=\frac1{\sqrt D}\langle\Theta^{(j)},\Psi(z)\rangle,
 \qquad
 \widetilde q_{ij}(z)=\frac1\rho
 \left\langle\frac{S^{(j)}_{i:}}{\sqrt D},\Psi(z)\right\rangle.
 \label{eq:nonconvex-local-radial-coordinates}
\end{equation}
Here $q_j$ is the scalar alignment with the stage-$j$ hidden direction, and
$\widetilde q_{ij}$ is the corresponding component term determined by the
selected row.  The displayed table average gives
$n^{-1}\sum_{i=1}^n\widetilde q_{ij}(z)=q_j(z)$ for every $z$.
Fix
\begin{equation}
    C=1,
    \qquad \lambda=\frac{33}{16},
\end{equation}
and use the gate thresholds
\begin{equation}
 a=\frac5{32},
 \qquad
 b=\frac9{32}.
 \label{eq:nonconvex-local-gate-thresholds}
\end{equation}
Thus $s\le a$ implies $G(s)=G'(s)=0$, whereas $s\ge b$ implies
$G(s)=1$ and $G'(s)=0$.  In the link from stage $j-1$ to stage $j$,
$q_{j-1}$ is the gate coordinate and $\widetilde q_{ij}$ is component $i$'s
stage contribution.  Define
\begin{equation}
\begin{aligned}
 R_i(x)
 ={}&-\widetilde q_{i1}(x_1)
 -\sum_{j=2}^T
 G(q_{j-1}(x_{j-1}))
 \bigl(\widetilde q_{ij}(x_j)+C\bigr)\\
 &+\frac\lambda2\sum_{j=1}^T\|x_j\|^2.
\end{aligned}
\label{eq:uniform-component}
\end{equation}
Averaging over components and using
$n^{-1}\sum_i\widetilde q_{ij}=q_j$ gives
\begin{equation}
\begin{aligned}
 R(x):=\frac1n\sum_{i=1}^nR_i(x)
 ={}&-q_1(x_1)
 -\sum_{j=2}^T
 G(q_{j-1}(x_{j-1}))(q_j(x_j)+C)\\
 &+\frac\lambda2\sum_{j=1}^T\|x_j\|^2.
\end{aligned}
\label{eq:uniform-average}
\end{equation}

For progress index $h\in\{0,\ldots,T\}$, let $R_i^{[h]}$
be~\eqref{eq:uniform-component} with the sum of chain links truncated at
$\min\{h+1,T\}$.

Set
\begin{equation}
 M_1:=\|G'\|_\infty,
 \qquad
 M_2:=\|G''\|_\infty,
 \qquad
 \mathsf H(\rho):=
 \lambda+2+(1+C\rho)(M_2+2M_1)+2M_1.
 \label{eq:nonconvex-general-H}
\end{equation}

\begin{proposition}[Properties of the nonconvex radial chain]
\label{prop:nonconvex-chain-properties}
Let $\alpha,\beta>0$.  Define
$f_i(y):=\alpha R_i(\beta y)$ and
$F(y):=n^{-1}\sum_i f_i(y)=\alpha R(\beta y)$.  Then:
\begin{enumerate}[label=\textup{(\roman*)},leftmargin=*,itemsep=2pt]
 \item\label{item:nc-property-smoothness}
 Each component satisfies
 \begin{equation}
  \operatorname{Lip}(\nabla f_i)
  \le\frac{\alpha\beta^2\mathsf H(\rho)}{\rho}.
  \label{eq:nc-property-smoothness}
 \end{equation}
 \item\label{item:nc-property-gap}
 With $C_0:=C+1=2$,
 \begin{equation}
  F(0)-\inf F\le\alpha C_0T.
  \label{eq:nc-property-gap}
 \end{equation}
 \item\label{item:nc-property-barrier}
 For the constant $\chi>0$ in
 Lemma~\ref{lem:radial-alignment-barrier},
 \begin{equation}
  \min_jq_j(\beta y_j)<b
  \quad\Longrightarrow\quad
  \|\nabla F(y)\|\ge\alpha\beta\chi.
  \label{eq:nc-property-barrier}
 \end{equation}
 \item\label{item:nc-property-revelation}
 For $0\le h<T$, a response that does not complete stage $h+1$ uses only
 $S^{(h+1)}_{i:}$ from that stage, while the response that completes it may
 also use $S^{(h+2)}_{i:}$ when
 $h+2\le T$.  For any $0\le h\le T$, if every omitted gate input is at most
 $a$, then
 \begin{equation}
  \bigl(R_i(x),\nabla R_i(x)\bigr)
  =\bigl(R_i^{[h]}(x),\nabla R_i^{[h]}(x)\bigr).
  \label{eq:nc-property-revelation}
 \end{equation}
\end{enumerate}
\end{proposition}

\subsection{Initial gap and individual smoothness}

\begin{lemma}[Initial gap]
\label{lem:uniform-gap}
With $C_0=C+1$,
\begin{equation}
    R(0)=0,
    \qquad R(0)-\inf_xR(x)\le C_0T.
    \label{eq:uniform-gap}
\end{equation}
\end{lemma}

\begin{proof}
At the origin all $q_j=0$ and $G(0)=0$, so $R(0)=0$.  Since $|q_j|<1$, $0\le G\le1$, and the quadratic is nonnegative,
$R(x)\ge-1-(T-1)(C+1)\ge-C_0T$.
\end{proof}

\begin{lemma}[Individual smoothness]
\label{lem:uniform-smoothness}
For every $\rho\in\mathcal R_n$,
\begin{equation}
 \sup_x\|\nabla^2R_i(x)\|_{\op}
 \le\frac {\mathsf H(\rho)}\rho
 \qquad\text{for every }i.
 \label{eq:general-component-smoothness}
\end{equation}
For the least parity-compatible $m\ge\sqrt n$, one has
$\rho=m/n<17/(16\sqrt n)\le\overline\rho:=17/512$, and
\begin{equation}
 H_{\rm main}:=341.
 \label{eq:explicit-component-H}
\end{equation}
Then every component satisfies
\begin{equation}
 \sup_x\|\nabla^2R_i(x)\|_{\op}
 \le\frac {H_{\rm main}}\rho.
 \label{eq:uniform-component-smoothness}
\end{equation}
The bound is independent of $T,D,$ and $n$, except through $1/\rho$.
\end{lemma}

\begin{proof}
\emph{One chain link.}
Consider one chain link
\begin{equation}
 A_{ij}(x_{j-1},x_j)
 =G(q_{j-1}(x_{j-1}))
 \bigl(\widetilde q_{ij}(x_j)+C\bigr).
\end{equation}
Apply Lemma~\ref{lem:gated-link-hessian} to $-A_{ij}$ with
\begin{equation}
 (B_0,B_1,B_2)=\frac1\rho(1,1,2),
 \qquad (g_1,g_2)=(1,2),
 \qquad c=C.
\end{equation}
By Lemma~\ref{lem:explicit-flat-gate},
\begin{equation}
 M_1\le17,
 \qquad
 M_2\le259.
 \label{eq:nonconvex-local-gate-derivatives}
\end{equation}
It gives, respectively, the current diagonal, cross, and previous diagonal
block bounds
\begin{equation}
 \frac2\rho,
 \qquad
 \frac{M_1}\rho,
 \qquad
 \frac{(1+C\rho)(M_2+2M_1)}\rho.
 \label{eq:uniform-link-hessian-bounds}
\end{equation}
For the last bound we used
$1/\rho+C=(1+C\rho)/\rho$, not the looser factor $(C+1)/\rho$.

\emph{Global block row.}
The first component term has Hessian norm at most $2/\rho$, and the
quadratic contributes $\lambda I$.  A global block row contains at most one
current diagonal term, one diagonal term induced by its successor, and two adjacent
cross terms.  Therefore the block row clause of
Lemma~\ref{lem:gated-link-hessian} gives
\begin{equation}
\begin{aligned}
 \max_j\sum_k\|\nabla^2_{x_jx_k}R_i(x)\|_{\op}
 &\le \lambda+\frac2\rho
 +\frac{(1+C\rho)(M_2+2M_1)}\rho
 +\frac{2M_1}\rho\\
 &\le\frac{\mathsf H(\rho)}{\rho}.
\end{aligned}
\end{equation}
The second line uses $\rho\le1$ and the definition of $\mathsf H(\rho)$.
This proves~\eqref{eq:general-component-smoothness} with no factor
depending on $T$ or $D$.  For the optimized bias,
$\rho\le\overline\rho$ and substitution of the displayed derivative bounds
gives $\mathsf H(\rho)\le H_{\rm main}$, proving
\eqref{eq:uniform-component-smoothness}.
\end{proof}

\subsection{Gradient lower bound and scaling}

\begin{lemma}[Nonconvex gradient lower bound at a constant scale]
\label{lem:radial-barrier}
Let $\chi>0$ be the constant in
Lemma~\ref{lem:radial-alignment-barrier}.  Then, for the average objective
$R$ in~\eqref{eq:uniform-average},
\begin{equation}
 \min_{1\le j\le T}q_j(x_j)<b
 \quad\Longrightarrow\quad
 \|\nabla R(x)\|\ge\chi.
 \label{eq:radial-barrier}
\end{equation}
\end{lemma}

\begin{proof}
Let $j$ be the first stage with $q_j(x_j)<b$.  All previous gates equal one.
The $j$th gradient block is
\begin{equation}
 \nabla_{x_j}R(x)
 =\lambda x_j-\gamma_j\nabla q_j(x_j),
 \qquad \gamma_j\ge1.
 \label{eq:radial-block-field}
\end{equation}
For $j<T$, the successor chain link adds
$G'(q_j)(q_{j+1}+C)\ge0$ to $\gamma_j$; for $j=T$,
$\gamma_j=1$.  Positivity uses $C=1$ and $q_{j+1}>-1$.
Lemma~\ref{lem:radial-alignment-barrier}, applied with
$u=u_{\Theta^{(j)}}$ and $c=\gamma_j$, now gives
$\|\nabla_{x_j}R(x)\|\ge\chi$, proving the claim.
\end{proof}

\begin{proof}[Proof of Proposition~\ref{prop:nonconvex-chain-properties}]
Item~\ref{item:nc-property-smoothness} follows from
Lemma~\ref{lem:uniform-smoothness} and the chain rule.
Items~\ref{item:nc-property-gap} and~\ref{item:nc-property-barrier} follow from
Lemmas~\ref{lem:uniform-gap} and~\ref{lem:radial-barrier}, respectively, after
the same rescaling.  For item~\ref{item:nc-property-revelation}, inspection of
\eqref{eq:uniform-component} shows that the current stage contribution uses only
$S^{(h+1)}_{i:}$ and that a newly active successor term uses only
$S^{(h+2)}_{i:}$, when that stage exists.  Since $G=G'=0$ below the lower
threshold, the displayed equality holds: every omitted value term contains
$G$, and every omitted gradient term contains $G$ or $G'$.
\end{proof}

\paragraph{Tunable scaling.}
For every $\rho\in\mathcal R_n$, set
\begin{equation}
 H_{\rm tun}:=368,
 \label{eq:tunable-component-H}
\end{equation}
and, for a sufficiently small universal $c_T>0$, define
\begin{equation}
 \beta_\rho=\frac{L_{\max}\rho\chi}{8H_{\rm tun}\varepsilon},
 \qquad
 \alpha_\rho=\frac{64H_{\rm tun}\varepsilon^2}
 {L_{\max}\rho\chi^2},
 \qquad
 T_\rho=\left\lfloor
 c_T\frac{\Delta L_{\max}\rho}{\varepsilon^2}
 \right\rfloor.
 \label{eq:tunable-analytic-parameters}
\end{equation}

\begin{lemma}[Analytic bounds for a tunable bias]
\label{lem:tunable-bias-analytic}
With the preceding definitions, $H_{\rm tun}$ is universal and
\begin{equation}
 \sup_x\|\nabla^2R_i(x)\|_{\op}\le\frac{H_{\rm tun}}{\rho}
 \qquad\text{for every }i.
 \label{eq:tunable-component-smoothness}
\end{equation}
Moreover, define
$f_i(y):=\alpha_\rho R_i(\beta_\rho y)$ and
$F(y):=n^{-1}\sum_{i=1}^nf_i(y)=\alpha_\rho R(\beta_\rho y)$.
Then $f_i$ is $L_{\max}$-smooth for every $i$,
$F(0)-\inf F\le\Delta$, and
$\min_jq_j(\beta_\rho y_j)<b$ implies
$\|\nabla F(y)\|\ge8\varepsilon$.  In particular,
$T_\rho=\Theta(\Delta L_{\max}\rho/\varepsilon^2)$ whenever $T_\rho\ge2$.
\end{lemma}

\begin{proof}
The proof of Lemma~\ref{lem:uniform-smoothness} based on block row sums uses the bias
only through $1+C\rho$.  Since $\rho\le1/8$, the bounds
$M_1\le17$ and $M_2\le259$ give
$\mathsf H(\rho)\le\mathsf H(1/8)<368=H_{\rm tun}$, which proves
\eqref{eq:tunable-component-smoothness}.

\emph{Smoothness.}
For $f_i(y)=\alpha_\rho R_i(\beta_\rho y)$, the chain rule and
\eqref{eq:tunable-component-smoothness} give
\begin{equation}
 \|\nabla^2 f_i(y)\|_{\op}
 =\alpha_\rho\beta_\rho^2
   \|\nabla^2R_i(\beta_\rho y)\|_{\op}
 \le\frac{\alpha_\rho\beta_\rho^2H_{\rm tun}}\rho
 =L_{\max}.
\end{equation}

\emph{Gradient lower bound.}
If $\min_jq_j(\beta_\rho y_j)<b$, Lemma~\ref{lem:radial-barrier} gives
$\|\nabla R(\beta_\rho y)\|\ge\chi$.  Hence
\begin{equation}
 \|\nabla F(y)\|
 =\alpha_\rho\beta_\rho\|\nabla R(\beta_\rho y)\|
 \ge\alpha_\rho\beta_\rho\chi
 =8\varepsilon.
\end{equation}

\emph{Initial gap and stage count.}
Lemma~\ref{lem:uniform-gap} and the definition of $T_\rho$ imply
\begin{equation}
 F(0)-\inf F
 \le\alpha_\rho C_0T_\rho
 \le\frac{64C_0H_{\rm tun}c_T}{\chi^2}\,\Delta
 \le\Delta,
\end{equation}
after choosing $c_T\le\chi^2/(64C_0H_{\rm tun})$.  Finally, when
$T_\rho\ge2$, the floor in~\eqref{eq:tunable-analytic-parameters} loses at
most a factor of two, proving
$T_\rho=\Theta(\Delta L_{\max}\rho/\varepsilon^2)$.
\end{proof}

\proofguidespace
\begin{proposition}[Bias--smoothness tradeoff]
\label{prop:bias-smoothness-tradeoff}
There are universal constants $c>0$ and $p_\star<1/2$ such that, for every
$n\ge n_0$, $L_{\max},\Delta,\varepsilon>0$, and admissible $\rho$ for which
the integer row threshold $r_\rho$ and stage count $T_\rho$ are both at least
two, every randomized IFO algorithm using at most
\begin{equation*}
 c\frac{\Delta L_{\max}}{\varepsilon^2}
 \min\{n\rho,\rho^{-1}\}
\end{equation*}
calls has, in some dimension $d$, an instance
$F_\rho\in\mathcal F_{\rm ind}^{n,d}(\Delta,L_{\max})$ for which
$\Pr_\xi(\|\nabla F_\rho(\widehat X)\|\le\varepsilon)\le p_\star$.
\end{proposition}

\begin{proof}
We keep $\rho$ arbitrary rather than immediately selecting
$\rho\asymp n^{-1/2}$.  The corresponding row threshold and number of stages are
\begin{equation}
 r_\rho=\left\lfloor\frac1{4096}
          \min\{n,\rho^{-2}\}\right\rfloor,
 \qquad
 T_\rho=\left\lfloor c_T
          \frac{\Delta L_{\max}\rho}{\varepsilon^2}\right\rfloor.
 \label{eq:restated-bias-stage-parameters}
\end{equation}
Assume $r_\rho,T_\rho\ge2$ and choose the block dimension large enough for
the one-stage bound and the bound on $\cEfuture$, equivalently
\begin{equation}
 D\ge c_D\left[1+\log\bigl((N+2)(T_\rho+1)\bigr)\right]
 \label{eq:restated-bias-dimension}
\end{equation}
for a universal $c_D$.

\emph{Assembly.}
Lemma~\ref{lem:tunable-bias-kl} gives the row cost $r_\rho$,
and Lemma~\ref{lem:tunable-bias-analytic} gives $T_\rho$ individually
smooth hidden stages and an $8\varepsilon$ gradient lower bound at the final stage.
Apply Theorem~\ref{thm:sequential} to the first $J=T_\rho-1$ stages.  Its joint event has
probability at least $95/128$ whenever
$N\le r_\rho(T_\rho-1)/2$, and on that event the gradient lower bound at the final stage
applies.
Since $r_\rho,T_\rho\ge2$, the floor losses are only
universal factors and
\begin{equation}
 r_\rho(T_\rho-1)
 =\Omega\!\left(
 \frac{\Delta L_{\max}}{\varepsilon^2}
 \min\{n\rho,\rho^{-1}\}\right).
 \label{eq:restated-bias-product}
\end{equation}
Consequently, a sufficiently small universal choice of $c$ makes the call
budget in the proposition at most $r_\rho(T_\rho-1)/2$.

The conversion in Section~\ref{sec:deterministic-to-randomized} now gives the
following statement for a fixed instance of the construction at accuracy
$\varepsilon$:
\begin{equation}
 \Pr_\xi\!\left(\|\nabla F_\rho(\widehat X)\|
                 \le4\varepsilon\right)
 \le\frac{33}{128}<\frac12.
 \label{eq:restated-bias-failure}
\end{equation}
Since the event with threshold $\varepsilon$ is contained in the event with
threshold $4\varepsilon$, this proves the proposition with $p_\star=33/128$.
\end{proof}

\subsection{Optimized bias}

When the optimized stage count is at least two, the construction yields the
$\Omega(\sqrt n\,\Delta L_{\max}/\varepsilon^2)$ lower bound, with the last
stage reserved for the algorithm's output.

Here $m$ is the least parity-compatible integer not smaller than $\sqrt n$ and
\begin{equation}
 \sqrt n\le m<\sqrt n+2,
 \qquad \rho=\frac mn,
 \qquad \frac1{\sqrt n}\le\rho<\frac{17}{16\sqrt n}.
 \label{eq:nonconvex-local-optimized-bias}
\end{equation}
The constants used below are $H_{\rm main}=341$ from
Lemma~\ref{lem:uniform-smoothness}, $\chi>0$ from
Lemma~\ref{lem:radial-alignment-barrier}, $C_0=C+1=2$, and the one-stage row
threshold $r=\lfloor n/512\rfloor$.

Set
\begin{equation}
    f_i(y)=\alpha R_i(\beta y),
    \qquad
    F(y):=\frac1n\sum_{i=1}^nf_i(y)=\alpha R(\beta y),
\end{equation}
where
\begin{equation}
    \beta=\frac{L_{\max}\rho\chi}{8H_{\rm main}\varepsilon},
 \qquad
    \alpha=\frac{64H_{\rm main}\varepsilon^2}{L_{\max}\rho\chi^2}.
 \label{eq:nonconvex-scales}
\end{equation}
Then
\begin{equation}
 \alpha\beta^2\frac {H_{\rm main}}\rho=L_{\max},
 \qquad
 \alpha\beta\chi=8\varepsilon.
 \label{eq:scale-equations}
\end{equation}

Let
\begin{equation}
 T=\left\lfloor
 c_T\frac{\Delta L_{\max}\rho}{\varepsilon^2}
 \right\rfloor,
 \label{eq:chain-length}
\end{equation}
where $c_T$ is small enough that $\alpha C_0T\le\Delta/2$.
Proposition~\ref{prop:nonconvex-chain-properties} and
\eqref{eq:scale-equations} give
\begin{equation}
 \max_i\operatorname{Lip}(\nabla f_i)\le L_{\max},
 \qquad F(0)-F^*\le\Delta,
 \qquad
 \min_jq_j(\beta y_j)<b\Longrightarrow
 \|\nabla F(y)\|\ge8\varepsilon.
 \label{eq:optimized-nonconvex-three-properties}
\end{equation}

Assume $T\ge2$; the case $T<2$ is handled next.  Reserve the
last block by setting
\begin{equation}
 J:=T-1.
 \label{eq:terminal-prefix}
\end{equation}
With $V_j(y_j)=\Psi(\beta y_j)$,
Proposition~\ref{prop:nonconvex-chain-properties}, item~\ref{item:nc-property-revelation},
gives \emph{(R2)--(R3)}.  Independent stage tables give \emph{(R1)}, and the
rule for assigning calls, including the free initialization row convention, gives
\emph{(R4)}.

Choose the integer $D$ so that
\begin{equation}
 D\ge\max\left\{
 2048[1+\log(N+3)],\quad
 \frac{2048}{25}\log\bigl(128(N+1)T\bigr)
 \right\}.
 \label{eq:nonconvex-local-dimension-requirements}
\end{equation}
The case $N=0$ is immediate, so assume $N\ge1$.
The corollary, with $M_{\rm tot}=T$, now gives
\begin{equation}
 \Prob\bigl(\nu_N<T-1,\ \cEfuture^c\bigr)\ge\frac{95}{128}
 \quad\text{whenever}\quad
 N\le\frac{r(T-1)}2.
 \label{eq:high-precision-cost}
\end{equation}
On this event, Corollary~\ref{cor:optimized-radial-chain} gives
$q_T(\beta\widehat X_T)\le a<b$.  Therefore
\eqref{eq:optimized-nonconvex-three-properties} yields
\begin{equation}
 \|\nabla F(\widehat X)\|\ge8\varepsilon.
\end{equation}
Since $T-1\ge T/2$ and $n\rho=m=\Theta(\sqrt n)$, the query threshold in
\eqref{eq:high-precision-cost} is
$\Omega(\sqrt n\,\Delta L_{\max}/\varepsilon^2)$.  The reserved last stage is
what makes the conclusion valid for an output that was never queried.

Finally, $D=O(\log((N+2)(T+2)))$ under the displayed requirement.
Since $d=TD$, substituting~\eqref{eq:chain-length} and
\eqref{eq:high-precision-cost} gives
\begin{equation}
 d=O\!\left[
 \left(1+\frac{\Delta L_{\max}}{\sqrt n\,\varepsilon^2}\right)
 \log\!\left(2+n+\sqrt n\,\frac{\Delta L_{\max}}{\varepsilon^2}\right)
 \right].
 \label{eq:nonconvex-dimension}
\end{equation}

\subsection{The linear sample-size term}

When $T<2$, one auxiliary quadratic gives the full bound.

Because the optimized bias satisfies $\rho=\Theta(n^{-1/2})$, the condition
$T\ge2$ in~\eqref{eq:chain-length} is equivalent, up to universal constants, to
\begin{equation}
 \varepsilon^2\lesssim\frac{\Delta L_{\max}}{\sqrt n}
 \quad\Longleftrightarrow\quad
 \sqrt n\,\frac{\Delta L_{\max}}{\varepsilon^2}\gtrsim n.
 \label{eq:nonconvex-regime-split}
\end{equation}
Thus the radial chain applies in the regime where its second complexity term is
at least a constant multiple of $n$.

For the case $T<2$, instantiate the auxiliary quadratic already defined
in Lemma~\ref{lem:dense-quadratic-row-bound} with
\begin{equation}
 \tau:=16\varepsilon,
 \qquad
 \gamma:=\frac{\tau^2}{2\Delta}
 =\frac{128\varepsilon^2}{\Delta}.
\end{equation}
Its components $f_i^{\rm quad}$ have Hessian $\gamma I$, and its average obeys
\begin{equation}
 F_{\rm quad}(0)-F_{\rm quad}^*
 =\frac{\tau^2}{2\gamma}=\Delta,
 \qquad
 \nabla F_{\rm quad}(x)=\gamma x-\tau u.
\end{equation}
Under the standing assumption
$\varepsilon^2\le c_0\Delta L_{\max}$, choose the theorem's universal
constant so that $c_0\le1/128$.  Then
$\gamma\le L_{\max}$, so every component is $L_{\max}$-smooth.  The same lemma
gives, after reducing the universal budget constant so that
$N\le\lfloor n/4\rfloor$,
\begin{equation}
 \Prob(\|\nabla F_{\rm quad}(\widehat X)\|\le4\varepsilon)
 \le\frac5{16}.
 \label{eq:coarse-failure}
\end{equation}

In the complementary regime,
$\sqrt n\,\Delta L_{\max}/\varepsilon^2=O(n)$.  Hence the target budget
$n+\sqrt n\,\Delta L_{\max}/\varepsilon^2$ is $O(n)$, so
\eqref{eq:coarse-failure} gives the full bound after reducing constants.
This quadratic uses $D_0=4$ in the
coarse regime, within the dimension bound
in~\eqref{eq:nonconvex-dimension}.  Combining the two regimes proves the
fixed budget claim in Theorem~\ref{thm:nonconvex-main}.

\section{Proofs for the Global PL Bounds}
\label{app:pl}

\paragraph{Restatement of the lower bound in
Theorem~\ref{thm:pl-main}.}
Fix $n\ge1024$, $L_{\max},\mu,\Delta>0$, and
$\kappa_{\max}=L_{\max}/\mu\ge3$.  There are universal constants
$c,c_\varepsilon>0$ such that, whenever
$0<\varepsilon\le c_\varepsilon\Delta$, every randomized IFO algorithm using at most
\begin{equation}
 N\le c\begin{cases}
 \displaystyle
 n+\frac{n\log(\Delta/\varepsilon)}
 {1+\log(\sqrt n/\kappa_{\max})},
 &3\le\kappa_{\max}\le\sqrt n,\\[2mm]
 \displaystyle
 n+\kappa_{\max}\sqrt n\log\frac{\Delta}{\varepsilon},
 &\kappa_{\max}\ge\sqrt n
 \end{cases}
 \label{eq:restated-pl-lower-budget}
\end{equation}
calls has a fixed instance
$F\in\mathcal F_{\mathrm{ind},\mathrm{PL}}^{n,d}
(\Delta,L_{\max},\mu)$, in some dimension $d$, for which
\begin{equation}
 \Pr_\xi\!\left(F(\widehat X)-F^*>\varepsilon\right)
 \ge\frac{11}{16}.
 \label{eq:restated-pl-lower-failure}
\end{equation}
Because the hard instances are individually smooth, the same lower bound also
holds for the mean-squared class.
Appendix~\ref{app:stopping} gives the conversion to a fixed hard instance and
the extension to expected call budgets.  The two-scale chain proves the first
regime and the bounded transition interval.  The geometrically weighted chain
proves the large-$\kappa_{\max}$ regime.  We analyze the weighted construction
first because its block decomposition of the gap isolates the global PL
argument.  The two-scale construction then treats the range in which the
weighted geometric ratio is not admissible.

\paragraph{Motivation for geometric stage scales.}
Let $j_\star$ be the first unfinished stage.  With equal stage scales, the
total contribution of stages $j_\star,\ldots,M$ grows like
$M-j_\star+1$, while the available lower bound on the squared gradient norm remains
constant.  Hence the PL constant would shrink as the chain gets longer.
Geometric scales instead give
\begin{equation*}
\begin{aligned}
 &\sum_{k=j_\star}^Mw_k^2=O(w_{j_\star}^2),
 &\|\nabla F(x)\|^2&=\Omega(w_{j_\star}^2).
\end{aligned}
\end{equation*}
Thus the first unfinished stage controls the total contribution of all
remaining stages without forcing $\mu$ to decrease with the chain length.  The
two constructions differ in how quickly the stages may shrink and still keep
every component $L_{\max}$-smooth.

\paragraph{Scale choices at a glance.}
Let $r_{\rm w}$ be the ratio of consecutive squared weights,
$\vartheta$ the ratio of consecutive amplitudes, and
$\bar\kappa_{\max}:=\min\{\kappa_{\max},\sqrt n\}$.  At the optimized bias, the
parameter choices that determine the two rates are
\begin{equation*}
\rho\asymp n^{-1/2},\qquad
1-r_{\rm w}\asymp\frac{\sqrt n}{\kappa_{\max}},\qquad
\vartheta\asymp\frac{R_j}{W_j}
\asymp\frac{\bar\kappa_{\max}}{\sqrt n}.
\end{equation*}
The first relation makes one hidden stage cost $\Theta(n)$ rows.  The second is
the admissible geometry using one scale in the large-$\kappa$ branch; the third is the
separated geometry used below the transition.  The explicit numerical constants
introduced later serve only to verify individual smoothness, the PL inequality,
and admissibility.  They do not affect these scale relations or the final rates.

\paragraph{Shared encoding.}
All staged constructions use the smooth threshold gate and independent fixed tables
$S^{(j)}\in\{\pm1\}^{n\times D}$ sampled conditional on
$\Theta^{(j)}\in\{\pm1\}^D$.  Fix $\rho=m/n\in\mathcal R_n$ as in
Appendix~\ref{app:gate-code}.
The row $S^{(j)}_{i:}$ belongs to component $i$ and
$n^{-1}\sum_iS^{(j)}_{i:}=\rho\Theta^{(j)}$.  We use
\begin{equation*}
 \begin{aligned}
  \Psi(z)&:=\frac{z}{\sqrt{1+\|z\|^2}},
  &u_{\Theta^{(j)}}&:=\frac{\Theta^{(j)}}{\sqrt D},
  \\
  q_{\Theta^{(j)}}(z)&:=\frac1{\sqrt D}
     \langle\Theta^{(j)},\Psi(z)\rangle,
  &q_j(z)&:=q_{\Theta^{(j)}}(z),\\
  \widetilde q_{ij}(z)&:=
  \frac1\rho\left\langle\frac{S^{(j)}_{i:}}{\sqrt D},\Psi(z)\right\rangle.
 \end{aligned}
\end{equation*}
Thus $q_j$ is the alignment coordinate with the hidden direction at stage $j$,
$\widetilde q_{ij}$ is the component term determined by row $i$, and
$n^{-1}\sum_i\widetilde q_{ij}=q_j$ exactly.

Throughout this appendix, $a=5/32$ and $b=9/32$ are the gate thresholds:
$s\le a$ implies $G(s)=G'(s)=0$, whereas $s\ge b$ implies
$G(s)=1$ and $G'(s)=0$.

All three constructions used for the lower bounds share the same fixed encoding table.  The
nonconvex chain has a constant stage scale, the large-$\kappa$ PL chain uses one
geometric scale $w_j$ per stage, and the small-$\kappa$ PL chain separates the
gate scale $R_j$ from the contribution scale $W_j$.

\paragraph{Part I: weighted chain for the large-$\kappa$ branch.}
For the large-$\kappa$ branch, we use one geometric scale $w_j$ for both the gate and
the stage contribution.  We say that a gate is \emph{fully open} when
$u_j\ge b$, so $G(u_j)=1$ and $G'(u_j)=0$.  We reserve \emph{stage completion}
for the information threshold $u_j>a$ used in
Appendix~\ref{app:information}.

\subsection{Weighted construction}

Fix an integer $M\ge2$ and positive nonincreasing weights
\begin{equation}
 1=w_1\ge w_2\ge\cdots\ge w_M>0.
 \label{eq:weights}
\end{equation}
For any stage $j$, the remaining squared weight is
$\sum_{k=j}^M w_k^2$.  Its largest ratio to the squared weight of the current
stage is
\begin{equation}
 \mathcal T(w):=\max_{1\le j\le M}
 \frac{\sum_{k=j}^M w_k^2}{w_j^2}.
 \label{eq:weight-tail-ratio}
\end{equation}
For each stage define
\begin{equation}
 u_j(x_j)=q_{\Theta^{(j)}}(x_j/w_j),
 \qquad
 Q_j(x_j)=w_j^2u_j(x_j),
 \label{eq:weighted-radial-global}
\end{equation}
and
\begin{equation}
 \widetilde Q_{ij}(x_j)
 =\frac{w_j^2}{\rho}
 \left\langle\frac{S^{(j)}_{i:}}{\sqrt D},
 \Psi(x_j/w_j)\right\rangle.
 \label{eq:weighted-radial-component}
\end{equation}
Then $n^{-1}\sum_i\widetilde Q_{ij}=Q_j$ exactly.  Here $u_j$ is the gate
coordinate, $\widetilde Q_{ij}$ is the stage contribution in component $i$, and $Q_j$ is
its component average.  In the link from stage $j-1$ to stage $j$, the factor
$G(u_{j-1})$ decides whether the current stage contribution is present.  The
condition $G=G'=0$ makes this link and its gradients in both adjacent blocks vanish
whenever $u_{j-1}\le a$.  Fix
\begin{equation}
 C=1,
 \qquad
 \lambda=\frac{33}{16},
 \qquad
 \sigma:=\lambda-2=\frac1{16},
 \label{eq:pl-analytic-constants}
\end{equation}
The component hard instance is
\begin{equation}
\begin{aligned}
 R_i(x)
 ={}&-\widetilde Q_{i1}(x_1)
 -\sum_{j=2}^M
 G(u_{j-1}(x_{j-1}))
 \bigl(\widetilde Q_{ij}(x_j)+Cw_j^2\bigr)\\
 &+\frac\lambda2\sum_{j=1}^M\|x_j\|^2,
\end{aligned}
\label{eq:weighted-component}
\end{equation}
and its average is $R:=n^{-1}\sum_{i=1}^nR_i$; equivalently, averaging the
displayed formula replaces every $\widetilde Q_{ij}$ by $Q_j$.

\begin{proposition}[Properties of the weighted chain]
\label{prop:weighted-chain-properties}
For $M\ge2$, $\rho\in\mathcal R_n$ with $0<\rho\le17/512$, and
$1=w_1\ge\cdots\ge w_M>0$, there are universal constants
$H_{\rm main},C_{\rm PL}^{\rm gap}>0$ such that:
\begin{enumerate}[label=\textup{(\roman*)},leftmargin=*,itemsep=2pt]
 \item\label{item:weighted-property-smoothness}
 Every component is $H_{\rm main}/\rho$-smooth.
 \item\label{item:weighted-property-gap}
 The initial gap satisfies
 \begin{equation}
  R(0)-R^*\le2\mathcal T(w).
  \label{eq:weighted-property-gap}
 \end{equation}
 \item\label{item:weighted-property-pl}
 The average satisfies
 \begin{equation}
  R(x)-R^*
  \le C_{\rm PL}^{\rm gap}\mathcal T(w)\|\nabla R(x)\|^2.
  \label{eq:weighted-property-pl}
 \end{equation}
 Equivalently, it is PL with constant
 $\bigl(2C_{\rm PL}^{\rm gap}\mathcal T(w)\bigr)^{-1}$.
 \item\label{item:weighted-property-residual}
 If $\mathcal I\subseteq\{2,\ldots,M\}$ and
 $G(u_{j-1}(x_{j-1}))=0$ for every $j\in\mathcal I$, then
 \begin{equation}
  R(x)-R^*\ge\sum_{j\in\mathcal I}w_j^2.
  \label{eq:weighted-property-residual}
 \end{equation}
\end{enumerate}
\end{proposition}

The offset $Cw_j^2$ carries no hidden direction.  After averaging, it makes
$Q_j+Cw_j^2=w_j^2(u_j+C)>0$.  In the decomposition below, a link
whose predecessor coefficient is zero then leaves a fixed positive residual.

Lemma~\ref{lem:radial-first-second} gives
\begin{align}
 \|\nabla Q_j\|&\le w_j,
 &\|\nabla^2Q_j\|_{\op}&\le2,\nonumber\\
 \|\nabla\widetilde Q_{ij}\|&\le\frac{w_j}{\rho},
 &\|\nabla^2\widetilde Q_{ij}\|_{\op}&\le\frac2\rho,\nonumber\\
 \|\nabla u_j\|&\le\frac1{w_j},
 &\|\nabla^2u_j\|_{\op}&\le\frac2{w_j^2}.
 \label{eq:weighted-derivatives}
\end{align}

\subsection{Decomposition of the gap by blocks}

Define the predecessor coefficient
\begin{equation}
 c_1:=1,
 \qquad
 c_j:=G(u_{j-1})\quad(2\le j\le M).
 \label{eq:weighted-predecessor-coefficient}
\end{equation}
The potential for one block,
\begin{equation}
 \phi_j(z):=\frac\lambda2\|z\|^2-q_{\Theta^{(j)}}(z)
 \label{eq:weighted-unit-block}
\end{equation}
is $\sigma$-strongly convex by~\eqref{eq:weighted-derivatives}.  Its minimum
does not depend on $j$ or $\Theta^{(j)}$ by rotational symmetry.  Write
\begin{equation}
 d_*:=-\min_z\phi_j(z).
 \label{eq:weighted-unit-depth}
\end{equation}
The minimizer is $z_j^*=t_*u_{\Theta^{(j)}}$, where
$\lambda t_*=(1+t_*^2)^{-3/2}$.  Its alignment
$q_*=t_*/\sqrt{1+t_*^2}$ satisfies $q_*>b$ by
\eqref{eq:stationary-margin}.  Moreover, $0<d_*<1$: the directional derivative
of $\phi_j(tu_{\Theta^{(j)}})$ at zero is negative, while at its finite minimizer
$z_j^*$ one has
$\phi_j(z_j^*)\ge-q_{\Theta^{(j)}}(z_j^*)>-1$.

For $c\in[0,1]$, define the block gap by the exact mixture
\begin{equation}
 e_{j,c}(z)
 :=c\bigl(\phi_j(z)+d_*\bigr)
 +(1-c)\left(\frac\lambda2\|z\|^2+1+d_*\right)
 \ge0.
 \label{eq:weighted-block-gap-mixture}
\end{equation}
Thus $e_{j,c}$ interpolates between two nonnegative terms.  Expanding
the mixture gives the equivalent formula
\begin{equation}
 e_{j,c}(z)
 =\frac\lambda2\|z\|^2-c\bigl(q_{\Theta^{(j)}}(z)+1\bigr)+1+d_*.
 \label{eq:weighted-block-gap-definition}
\end{equation}
At $c=1$, $e_{j,1}=\phi_j+d_*$ is the fully open, strongly convex block gap;
at $c=0$, $e_{j,0}(z)=\lambda\|z\|^2/2+1+d_*\ge1$ is a fixed residual.  Since
the minimizer of the first gap has alignment $q_*>b$, all block gaps can vanish
simultaneously, while every closed gate contributes objective error.

\begin{lemma}[Exact decomposition of the gap by blocks and consequences]
\label{lem:weighted-gap-decomposition}
For every $x$, with $z_j=x_j/w_j$ and $c_j$ from
\eqref{eq:weighted-predecessor-coefficient},
\begin{equation}
 R(x)-R^*
 =\sum_{j=1}^Mw_j^2e_{j,c_j}(z_j),
 \label{eq:weighted-gap-decomposition}
\end{equation}
where
\begin{equation}
 R^*=-d_*w_1^2-(1+d_*)\sum_{j=2}^Mw_j^2.
 \label{eq:weighted-optimal-value}
\end{equation}
At the origin,
\begin{equation}
 R(0)-R^*
 =d_*w_1^2+(1+d_*)\sum_{j=2}^Mw_j^2
 \le2\mathcal T(w).
 \label{eq:pl-initial-gap}
\end{equation}
If $c_j=0$ for every $j$ in an index set $\mathcal I$, then
\begin{equation}
 R(x)-R^*\ge\sum_{j\in\mathcal I}w_j^2.
 \label{eq:zero-coefficient-residual}
\end{equation}
\end{lemma}

\begin{proof}
The first stage contribution in~\eqref{eq:weighted-component} has no offset,
whereas every later stage contribution has offset $w_j^2$.  Consequently,
\begin{equation}
 R(x)=w_1^2+\sum_{j=1}^Mw_j^2
 \left[\frac\lambda2\|z_j\|^2-c_j(q_{\Theta^{(j)}}(z_j)+1)\right].
 \label{eq:weighted-objective-coefficient-form}
\end{equation}
Substituting~\eqref{eq:weighted-block-gap-definition} into this identity gives
\eqref{eq:weighted-gap-decomposition} with the value in
\eqref{eq:weighted-optimal-value}.  Every summand is
nonnegative by~\eqref{eq:weighted-block-gap-mixture}.  Conversely, setting
$z_j=z_j^*$ for every $j$ makes every predecessor gate fully open,
so $c_j=1$ and every summand vanishes.  This proves both the decomposition and
the displayed value of $R^*$.

At zero, $q_j(0)=0\le a$, so $c_1=1$ and $c_j=G(0)=0$ for $j\ge2$.
Equation~\eqref{eq:pl-initial-gap} follows from $0<d_*<1$, $w_1=1$, and the
definition of $\mathcal T(w)$.  Finally, if $c_j=0$, then
$e_{j,0}(z_j)=\lambda\|z_j\|^2/2+1+d_*\ge1$; retaining these terms in
\eqref{eq:weighted-gap-decomposition} proves
\eqref{eq:zero-coefficient-residual}.
\end{proof}

\subsection{Weighted gradient lower bound}

Write
\begin{equation}
 z_j=\frac{x_j}{w_j},
 \qquad
 u_j=q_{\Theta^{(j)}}(z_j),
 \qquad
 v_j=\nabla q_{\Theta^{(j)}}(z_j).
\end{equation}
Define the coefficient
\begin{equation}
 \gamma_j
 :=c_j+\ind\{j<M\}G'(u_j)
 \left(\frac{w_{j+1}}{w_j}\right)^2
 \bigl(q_{\Theta^{(j+1)}}(z_{j+1})+1\bigr).
 \label{eq:block-coefficient}
\end{equation}
Since $q_{\Theta^{(j+1)}}(z_{j+1})+1>0$, one has
$\gamma_j\ge c_j\ge0$.  Direct
differentiation gives
\begin{equation}
 \nabla_{x_j}R=w_j(\lambda z_j-\gamma_jv_j).
 \label{eq:weighted-block-gradient}
\end{equation}

\begin{lemma}[Weighted radial gradient lower bound]
\label{lem:weighted-barrier}
If $j_\star$ is the first stage with $u_{j_\star}(x_{j_\star})<b$, then
\begin{equation}
 \|\nabla R(x)\|\ge\chi w_{j_\star},
 \label{eq:weighted-barrier}
\end{equation}
where $\chi>0$ is the constant from
Lemma~\ref{lem:radial-alignment-barrier}.
\end{lemma}

\begin{proof}
If $j_\star=1$, then $c_{j_\star}=1$ by definition.  If $j_\star>1$, every
preceding gate is fully open and again $c_{j_\star}=1$.  Hence
$\gamma_{j_\star}\ge1$ in both cases.  The block
formula~\eqref{eq:weighted-block-gradient} and
Lemma~\ref{lem:radial-alignment-barrier} imply
\begin{equation}
 \|\nabla_{x_{j_\star}}R\|
 \ge w_{j_\star}\|\lambda z_{j_\star}-v_{j_\star}\|
 \ge\chi w_{j_\star}.
\end{equation}
\end{proof}

By definition of $\mathcal T(w)$, every $j_0$ obeys
\begin{equation}
 \sum_{j=j_0}^Mw_j^2
 \le\mathcal T(w)w_{j_0}^2.
 \label{eq:weight-tail}
\end{equation}
Figure~\ref{fig:weighted-pl-proof} displays the two estimates used for the
global PL inequality: the exact decomposition of the gap by blocks and the gradient bound
at the first gate that is not fully open.

\begin{figure}[H]
\centering
\includegraphics[width=\linewidth]{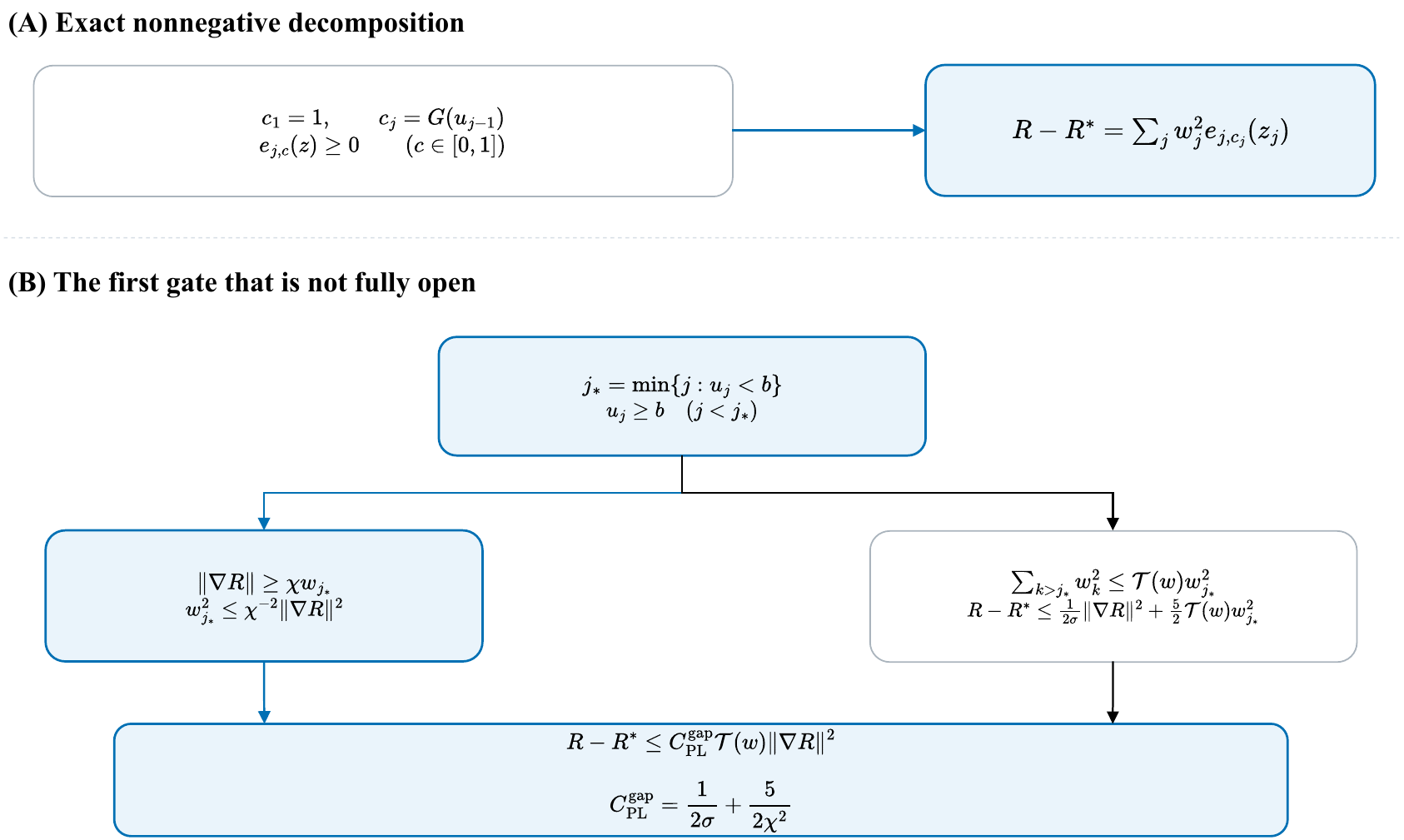}
\caption{Weighted PL inequality.  The decomposition of the gap by blocks gives
nonnegative terms.  If every gate is fully open,
strong convexity of one block directly controls the gap.  Otherwise, the first
gate that is not fully open gives the gradient barrier, while
$\mathcal T(w)$ bounds all later contributions.  The threshold $b$ for full
opening is distinct from the completion threshold $a$.}
\label{fig:weighted-pl-proof}
\end{figure}

\subsection{Global PL inequality from the block gaps}

Set
\begin{equation}
 C_{\rm PL}^{\rm gap}
 :=\frac1{2\sigma}+\frac{5}{2\chi^2}.
 \label{eq:block-gap-pl-constant}
\end{equation}

We use the standard consequence of differentiable $\sigma$-strong convexity
\begin{equation}
 f(x)-f^*\le\frac1{2\sigma}\|\nabla f(x)\|^2.
 \label{eq:strong-convex-gradient-domination}
\end{equation}
It follows by minimizing the right-hand side of
$f(y)\ge f(x)+\langle\nabla f(x),y-x\rangle
+(\sigma/2)\|y-x\|^2$ over $y$.

\begin{lemma}[Gap bound for one block]
\label{lem:weighted-block-gap}
For every block $j$ and every $x$,
\begin{equation}
 w_j^2e_{j,c_j}(z_j)
 \le\frac1{2\sigma}\|\nabla_{x_j}R(x)\|^2
 +\frac52w_j^2\mathbf 1\{c_j<1\}.
 \label{eq:weighted-block-gap-bound}
\end{equation}
\end{lemma}

\begin{proof}
\emph{Case 1: $\gamma_j\ge1$.}
If $u_j<b$, the radial
comparison~\eqref{eq:radial-norm-comparison} and the strong convexity of
$\phi_j$ apply because
$e_{j,1}(z_j)=\phi_j(z_j)+d_*=\phi_j(z_j)-\min_z\phi_j(z)$, and give
\begin{equation}
 e_{j,1}(z_j)
 \le\frac1{2\sigma}\|\lambda z_j-v_j\|^2
 \le\frac1{2\sigma}\|\lambda z_j-\gamma_jv_j\|^2.
\end{equation}
Since
$e_{j,c_j}=e_{j,1}+(1-c_j)(q_{\Theta^{(j)}}(z_j)+1)$ and
$q_{\Theta^{(j)}}(z_j)+1\le2$, this proves~\eqref{eq:weighted-block-gap-bound}.
If instead $u_j\ge b$, then $G'(u_j)=0$ and
$\gamma_j=c_j\le1$.  Thus $\gamma_j\ge1$ forces
$c_j=\gamma_j=1$, and the same bound from strong convexity applies without slack.

\emph{Case 2: $0\le\gamma_j<1$.}
Here $c_j\le\gamma_j<1$, so the indicator in
\eqref{eq:weighted-block-gap-bound} equals one.  Since
$q_{\Theta^{(j)}}(z_j)+1\ge0$ and $d_*<1$,
\begin{equation}
 e_{j,c_j}(z_j)\le\frac\lambda2\|z_j\|^2+2.
\end{equation}
Write $g_j:=\lambda z_j-\gamma_jv_j$.  Because
$\lambda z_j=g_j+\gamma_jv_j$, $\gamma_j<1$, and $\|v_j\|\le1$,
\begin{equation}
 \lambda^2\|z_j\|^2\le2\|g_j\|^2+2.
\end{equation}
Consequently,
\begin{equation}
 e_{j,c_j}(z_j)
 \le\frac1\lambda\|g_j\|^2+2+\frac1\lambda
 \le\frac1{2\sigma}\|g_j\|^2+\frac52.
\end{equation}
Multiplying by $w_j^2$ and using
$\nabla_{x_j}R=w_jg_j$ proves the claim.
\end{proof}

\begin{proof}[Proof of Proposition~\ref{prop:weighted-chain-properties},
item~\ref{item:weighted-property-pl}]
\emph{Case 1: every gate is fully open.}
Then $c_j=1$ for every $j$.
Summing Lemma~\ref{lem:weighted-block-gap} and using
\eqref{eq:weighted-gap-decomposition} gives
\begin{equation}
 R-R^*\le\frac1{2\sigma}\|\nabla R\|^2,
\end{equation}
which implies~\eqref{eq:weighted-property-pl} because $\mathcal T(w)\ge1$.

\emph{Case 2: a gate is not fully open.}
Let $j_\star$ be the first index with $u_{j_\star}<b$.  Every
predecessor gate for a block $j\le j_\star$ is fully open, so
$c_j=1$ for these blocks.  Summing Lemma~\ref{lem:weighted-block-gap} therefore
gives
\begin{equation}
\begin{aligned}
 R-R^*
 &\le\frac1{2\sigma}\|\nabla R\|^2
 +\frac52\sum_{j=j_\star+1}^Mw_j^2\\
 &\le\frac1{2\sigma}\|\nabla R\|^2
 +\frac52\mathcal T(w)w_{j_\star}^2.
\end{aligned}
\end{equation}
Lemma~\ref{lem:weighted-barrier} gives
$w_{j_\star}^2\le\chi^{-2}\|\nabla R\|^2$.  Therefore
\begin{equation}
\begin{aligned}
 R-R^*
 &\le\left(\frac1{2\sigma}
 +\frac{5}{2\chi^2}\mathcal T(w)\right)
 \|\nabla R\|^2\\
 &\le C_{\rm PL}^{\rm gap}\mathcal T(w)\|\nabla R\|^2,
\end{aligned}
\end{equation}
where the second line uses $\mathcal T(w)\ge1$ and
\eqref{eq:block-gap-pl-constant}.  This is~\eqref{eq:weighted-property-pl}.

Lemma~\ref{lem:radial-alignment-barrier} gives $\chi>5/18$ and hence
$\chi^2>1/13$.  Substitution into~\eqref{eq:block-gap-pl-constant} yields
\begin{equation}
 C_{\rm PL}^{\rm gap}
 <8+\frac{65}{2}=\frac{81}{2}<41.
\end{equation}
\end{proof}

\subsection{Verification of individual smoothness}

\begin{lemma}[Explicit weighted Hessian bound]
\label{lem:weighted-smoothness}
For every $\rho\in\mathcal R_n$ satisfying
$0<\rho\le\overline\rho=17/512$, let $H_{\rm main}=341$ be the constant
defined in~\eqref{eq:explicit-component-H}.
Then every component satisfies
\begin{equation}
 \sup_x\|\nabla^2R_i(x)\|_{\op}\le\frac {H_{\rm main}}\rho.
 \label{eq:weighted-component-smoothness}
\end{equation}
The bound is independent of $M,D,\mathcal T(w)$, and the
smallest weight.
\end{lemma}

\begin{proof}
Apply Lemma~\ref{lem:gated-link-hessian} to one link
$-G(u_{j-1})(\widetilde Q_{ij}+Cw_j^2)$ with
\begin{equation}
 (B_0,B_1,B_2)
 =\left(\frac{w_j^2}{\rho},\frac{w_j}{\rho},\frac2\rho\right),
 \quad
 (g_1,g_2)=\left(\frac1{w_{j-1}},\frac2{w_{j-1}^2}\right),
 \quad c=Cw_j^2.
\end{equation}
For this calculation,
\begin{equation}
 M_1:=\|G'\|_\infty\le17,
 \qquad
 M_2:=\|G''\|_\infty\le259.
 \label{eq:pl-local-gate-derivatives}
\end{equation}
Since $\rho\le\overline\rho$ and $w_j/w_{j-1}\le1$, the current diagonal,
cross, and predecessor diagonal bounds are at most
\begin{equation}
 \frac2\rho,\qquad
 \frac{M_1}\rho,\qquad
 \frac{(1+C\overline\rho)(M_2+2M_1)}\rho.
\end{equation}
The first stage contribution is at most $2/\rho$, and the quadratic contributes
$\lambda I$.  Since $\rho\le1$, the block row bound is
\begin{equation}
 \max_j\sum_k\|\nabla^2_{x_jx_k}R_i(x)\|_{\op}
 \le\frac{\lambda+2+(1+C\overline\rho)(M_2+2M_1)+2M_1}{\rho}
 \le\frac{H_{\rm main}}{\rho},
\end{equation}
which proves~\eqref{eq:weighted-component-smoothness}.
\end{proof}

\begin{proof}[Proof of the remaining items of
Proposition~\ref{prop:weighted-chain-properties}]
Item~\ref{item:weighted-property-smoothness} is
Lemma~\ref{lem:weighted-smoothness}.  Items~\ref{item:weighted-property-gap}
and~\ref{item:weighted-property-residual} follow from
Lemma~\ref{lem:weighted-gap-decomposition}.
\end{proof}

\subsection{Scaling, condition number, and information cost}

With the optimized bias and $r=\lfloor n/512\rfloor$, this subsection
verifies
\begin{equation*}
\begin{aligned}
 \rho&=\Theta(n^{-1/2}),
 &1-r_{\rm w}&=\Theta\!\left(\frac{\sqrt n}{\kappa_{\max}}\right),\\
 J&=\Theta\!\left(\frac{\kappa_{\max}}{\sqrt n}
     \log\frac{\Delta}{\varepsilon}\right),
 &rJ&=\Theta\!\left(\kappa_{\max}\sqrt n
     \log\frac{\Delta}{\varepsilon}\right).
\end{aligned}
\end{equation*}

\subsubsection{Weight geometry and stage count}

\emph{Geometric weights.}
Let $\kappa_{\max}=L_{\max}/\mu$ and set
\begin{equation}
 \Gamma:=2H_{\rm main}C_{\rm PL}^{\rm gap}.
 \label{eq:pl-gamma}
\end{equation}
The bounds $H_{\rm main}=341$ and $C_{\rm PL}^{\rm gap}<41$ give
\begin{equation}
 \Gamma<28000.
 \label{eq:gamma-upper-bound}
\end{equation}
Use throughout the weighted branch the optimized choice in
\eqref{eq:m-rho}.  Thus
\begin{equation}
 \rho=\frac mn,
 \qquad
 \sqrt n\le m<\sqrt n+2,
 \qquad
 \rho<\frac{17}{16\sqrt n}\le\overline\rho.
 \label{eq:weighted-fixed-bias}
\end{equation}
Assume
\begin{equation}
 \kappa_{\max}\ge\kappa_{\rm w}(n):=2\Gamma\sqrt n.
 \label{eq:weighted-kappa-threshold}
\end{equation}
This explicit threshold comes from the smoothness and PL bounds above.  The
two-scale construction covers the bounded transition interval, so the theorem
still changes regimes at $\kappa_{\max}=\sqrt n$.
Define the ratio of consecutive squared weights
\begin{equation}
 r_{\rm w}:=1-\frac{\Gamma}{\kappa_{\max}\rho}.
 \label{eq:weighted-geometric-ratio}
\end{equation}
Because $\rho\ge n^{-1/2}$, the threshold above gives
\begin{equation}
 \frac12\le r_{\rm w}<1,
 \qquad
 0<1-r_{\rm w}\le\frac12.
 \label{eq:weighted-ratio-range}
\end{equation}

\emph{Accuracy scale and stage count.}
We may decrease the universal $c_\varepsilon$ in the theorem so that
\begin{equation}
 \bar\varepsilon:=\frac{64\varepsilon}{\Delta}\le\frac14
 \qquad\text{and}\qquad
 \log\frac1{\bar\varepsilon}
 =\Theta\!\left(\log\frac{\Delta}{\varepsilon}\right).
 \label{eq:weighted-accuracy-range}
\end{equation}
The first inequality follows from $\varepsilon\le\Delta/256$.
Choose
\begin{equation}
 J:=\left\lfloor
 \frac{\log(1/\bar\varepsilon)}{-\log r_{\rm w}}
 \right\rfloor,
 \qquad
 M:=2J+1,
 \qquad
 w_j:=r_{\rm w}^{(j-1)/2}.
 \label{eq:weighted-stage-choice}
\end{equation}
The quantity inside the floor is at least two, so $J\ge2$.  The definition also
gives
\begin{equation}
 \bar\varepsilon\le r_{\rm w}^J
 <\frac{\bar\varepsilon}{r_{\rm w}}
 \le2\bar\varepsilon\le\frac12.
 \label{eq:weighted-floor-window}
\end{equation}

\emph{Value of $\mathcal T(w)$.}
For this geometric sequence,
\begin{equation}
 \mathcal T(w)
 \le\sum_{k=0}^{\infty}r_{\rm w}^k
 =\frac1{1-r_{\rm w}}
 =\frac{\kappa_{\max}\rho}{\Gamma}.
 \label{eq:weighted-geometric-tail-ratio}
\end{equation}

\subsubsection{Scaling, the gap from an unrevealed suffix, and stage count}

\emph{Rescale to the target gap and smoothness.}
A global rescaling $F(y)=\alpha R(\beta y)$ multiplies both the smoothness and
PL constants by $\alpha\beta^2$, and therefore preserves their ratio:
\begin{equation}
 L_F=\alpha\beta^2L_R,
 \qquad
 \mu_F=\alpha\beta^2\mu_R,
 \qquad
 \frac{L_F}{\mu_F}=\frac{L_R}{\mu_R}.
 \label{eq:global-scaling-condition-number}
\end{equation}
Thus the choice of $r_{\rm w}$ fixes the ratio between the PL and smoothness
constants.  The two global scales below enforce the initial gap and individual
smoothness bounds.  Define
\begin{equation}
 f_i(y)=\alpha R_i(\beta y),
 \qquad
 \alpha=\frac{\Delta(1-r_{\rm w})}{2},
 \qquad
 \beta^2=\frac{L_{\max}\rho}{\alpha H_{\rm main}}.
 \label{eq:pl-scaling}
\end{equation}
Let
\begin{equation}
 F(y):=\frac1n\sum_{i=1}^nf_i(y)=\alpha R(\beta y).
 \label{eq:weighted-scaled-average}
\end{equation}
Lemma~\ref{lem:weighted-smoothness} makes every component
$L_{\max}$-smooth, while~\eqref{eq:pl-initial-gap} and
\eqref{eq:weighted-geometric-tail-ratio} give $F(0)-F^*\le\Delta$.

\emph{PL constant after scaling.}
To track the PL constant under scaling, use
$\nabla F(y)=\alpha\beta\nabla R(\beta y)$ and
\begin{equation}
\begin{aligned}
 F(y)-F^*
 &=\alpha\bigl(R(\beta y)-R^*\bigr)\\
 &\le\frac{\alpha}{2\mu_0}\|\nabla R(\beta y)\|^2
 =\frac1{2\alpha\beta^2\mu_0}\|\nabla F(y)\|^2.
\end{aligned}
\label{eq:pl-scaling-identity}
\end{equation}
Proposition~\ref{prop:weighted-chain-properties},
item~\ref{item:weighted-property-pl}, and
\eqref{eq:weighted-geometric-tail-ratio} give
$\mu_0\ge(1-r_{\rm w})/(2C_{\rm PL}^{\rm gap})$.  Therefore the scaled PL
constant is at
least
\begin{equation}
 \alpha\beta^2\mu_0
 \ge\frac{L_{\max}\rho(1-r_{\rm w})}
 {2H_{\rm main}C_{\rm PL}^{\rm gap}}
 =\frac{L_{\max}}{\kappa_{\max}}
 =\mu,
 \label{eq:scaled-pl-constant}
\end{equation}
where the middle identity uses
$1-r_{\rm w}=\Gamma/(\kappa_{\max}\rho)$ and
$\Gamma=2H_{\rm main}C_{\rm PL}^{\rm gap}$.

\emph{Gap from the unrevealed suffix and stage count.}
Stage $J+1$ separates the first $J$ stages used in the call count from the
hidden suffix $j=J+2,\ldots,2J+1$, whose total weight is
\begin{equation}
 \sum_{j=J+2}^{2J+1}w_j^2
 =\frac{r_{\rm w}^{J+1}(1-r_{\rm w}^J)}{1-r_{\rm w}}.
 \label{eq:terminal-suffix-weight}
\end{equation}
If $u_k\le a$ for every $k>J$, then the predecessor of every block in this
suffix is below the lower gate threshold.  Thus $c_j=0$ throughout the hidden
suffix.  Hence the exact decomposition of the gap by blocks gives
$R-R^*\ge\sum_{j=J+2}^{2J+1}w_j^2$.  Substituting
$\alpha=\Delta(1-r_{\rm w})/2$ and
$\bar\varepsilon\le r_{\rm w}^J<2\bar\varepsilon\le1/2$, and using
$r_{\rm w}\ge1/2$, yields
\begin{equation}
 F(y)-F^*
 \ge\frac{\Delta}{2}r_{\rm w}^{J+1}(1-r_{\rm w}^J)
 \ge\frac{\Delta}{8}\bar\varepsilon
 =8\varepsilon.
 \label{eq:terminal-suffix-residual-scale}
\end{equation}

For $0<1-r_{\rm w}\le1/2$,
\begin{equation}
 1-r_{\rm w}\le-\log r_{\rm w}\le2(1-r_{\rm w}).
\end{equation}
Since the quantity inside the floor in~\eqref{eq:weighted-stage-choice} is at
least two, flooring loses at most a factor of two.  Consequently,
\begin{equation}
 \frac{\kappa_{\max}\rho}{4\Gamma}\log\frac1{\bar\varepsilon}
 \le J
 \le\frac{\kappa_{\max}\rho}{\Gamma}\log\frac1{\bar\varepsilon},
 \qquad
 J=\Theta\!\left(
 \kappa_{\max}\rho\log\frac{\Delta}{\varepsilon}
 \right).
 \label{eq:weighted-stage-count}
\end{equation}

\subsubsection{Sequential revelation, query cost, and dimension}

\emph{Information cost.}
The one-stage result at the optimized bias now applies directly.  Set
$r=\lfloor n/512\rfloor$.  Since $n\ge1024$, one has $r\ge n/1024$;
therefore~\eqref{eq:weighted-fixed-bias} and the lower bound in
\eqref{eq:weighted-stage-count} give
\begin{equation}
 rJ
 \ge\frac{\kappa_{\max}m}{4096\Gamma}\log\frac1{\bar\varepsilon}
 \ge\frac{\kappa_{\max}\sqrt n}{4096\Gamma}\log\frac1{\bar\varepsilon}.
 \label{eq:weighted-row-cost}
\end{equation}

For $h\in\{0,\ldots,M\}$, let $R_i^{[h]}$ keep the first stage contribution and full
quadratic sum in~\eqref{eq:weighted-component}, but truncate the sum of gated links
at $\min\{h+1,M\}$.  Thus future blocks retain their deterministic quadratic
gradients while omitted gated stage contributions are removed.  With
$V_j(y_j)=\Psi(\beta y_j/w_j)$, table independence, the dependence of each
component on its assigned table row, the condition $G=G'=0$ for omitted gates,
and the call accounting rule give \emph{(R1)--(R4)} of
Corollary~\ref{cor:optimized-radial-chain}.  Take an integer $D$
satisfying
\begin{equation}
 D\ge\max\left\{
 2048[1+\log(N+3)],\quad
 \frac{2048}{25}\log\bigl(128(N+1)M\bigr)
 \right\}.
 \label{eq:weighted-dimension-requirement}
\end{equation}
The case $N=0$ is immediate, so assume $N\ge1$.
The corollary, with $M_{\rm tot}=M=2J+1$, gives
\begin{equation}
 \Prob(\nu_N<J,\ \cEfuture^c)\ge\frac{95}{128}
 \quad\text{whenever}\quad
 N\le\frac{rJ}{2}.
 \label{eq:weighted-query-budget}
\end{equation}
On this event, Corollary~\ref{cor:optimized-radial-chain} gives
\begin{equation}
 u_j(\beta\widehat X_j)\le a
 \qquad\forall j>J.
\end{equation}
Equation~\eqref{eq:terminal-suffix-residual-scale} then gives
\begin{equation}
 F(\widehat X)-F^*\ge8\varepsilon.
\end{equation}
Together with~\eqref{eq:weighted-row-cost}, this proves the
$\kappa_{\max}\sqrt n\log(\Delta/\varepsilon)$ term.

Equations~\eqref{eq:weighted-fixed-bias} and
\eqref{eq:weighted-stage-count} give
$M=O((\kappa_{\max}/\sqrt n)\log(\Delta/\varepsilon))$.  Since $d=MD$,
and both $N+3$ and $M$ are bounded by a universal multiple of the quantity
inside the logarithm below, the two logarithmic requirements can be absorbed
into one:
\begin{equation}
 d=O\!\left[
 \frac{\kappa_{\max}}{\sqrt n}
 \log\frac{\Delta}{\varepsilon}
 \log\!\left(
 2+n+\kappa_{\max}\sqrt n
 \log\frac{\Delta}{\varepsilon}
 \right)
 \right].
 \label{eq:pl-dimension}
\end{equation}

\paragraph{Conclusion for the weighted chain.}
The preceding calculation establishes, after scaling,
\begin{equation*}
\begin{aligned}
 &\max_i\operatorname{Lip}(\nabla f_i)\le L_{\max},
 \qquad F(0)-F^*\le\Delta,
 \qquad F\ \text{is $\mu$-PL},\\
 &N\le rJ/2
 \ \Longrightarrow\ 
 \Prob\bigl(F(\widehat X)-F^*\ge8\varepsilon\bigr)\ge\frac{95}{128},
 \qquad
 rJ=\Omega\!\left(\kappa_{\max}\sqrt n\log\frac\Delta\varepsilon\right).
\end{aligned}
\end{equation*}

\subsection{Auxiliary \texorpdfstring{$\Omega(n)$}{Omega(n)} bound}

\begin{lemma}[PL auxiliary quadratic lower bound]
\label{lem:pl-sample-size}
For every $n\ge n_0=1024$ and integer $0\le N\le\lfloor n/4\rfloor$, let
$\mu=L_{\max}/\kappa_{\max}$, where $\kappa_{\max}\ge1$, and suppose
$0<\varepsilon\le\Delta/128$.  There exists a distribution over individually
$L_{\max}$-smooth finite sums in dimension $4$ whose averages are globally $\mu$-PL and have
initial gap $\Delta$ such that every algorithm making at most $N$ calls
satisfies
\begin{equation}
 \Prob\bigl(F(\widehat X)-F^*<8\varepsilon\bigr)\le\frac5{16}.
 \label{eq:pl-sample-size-success}
\end{equation}
\end{lemma}

\begin{proof}
Instantiate the quadratic $f_i^{\rm quad}$ from
Lemma~\ref{lem:dense-quadratic-row-bound} with
$\gamma=\mu$ and $\tau=(2\mu\Delta)^{1/2}$, and write
$F:=F_{\rm quad}$.  Each component has Hessian
$\mu I\preceq L_{\max}I$, while its average satisfies
\begin{equation}
 F(0)-F^*=\frac{\tau^2}{2\mu}=\Delta,
 \qquad
 \|\nabla F(x)\|^2=2\mu(F(x)-F^*).
\end{equation}
Completing the square gives
\begin{equation}
 F(x)-F^*=\frac1{2\mu}\|\mu x-\tau u\|^2.
\end{equation}
Hence the success event implies
\begin{equation}
 \|\mu\widehat X-\tau u\|^2
 <16\mu\varepsilon
 \le\frac{\tau^2}{16},
\end{equation}
where the last inequality uses $\tau^2=2\mu\Delta$ and
$\varepsilon\le\Delta/128$.  Lemma~\ref{lem:dense-quadratic-row-bound}, with
these parameters, bounds this success probability by $5/16$ and proves
\eqref{eq:pl-sample-size-success}.
\end{proof}

\paragraph{Lower bound from the weighted chain.}
The weighted chain has failure probability $95/128$, and
Lemma~\ref{lem:pl-sample-size} has failure probability $11/16$.
The larger of their two query thresholds is at least a universal constant
times their sum.  Thus, for
$\kappa_{\max}\ge\kappa_{\rm w}(n):=2\Gamma\sqrt n$, where $\Gamma<28000$,
the weighted
construction proves
\begin{equation}
 \Omega\!\left(
 n+\kappa_{\max}\sqrt n
 \log\frac{\Delta}{\varepsilon}
 \right)
 \label{eq:weighted-pl-summary}
\end{equation}
using the one-stage bound at the optimized bias for the chain and obtaining the
linear term from an auxiliary quadratic at small bias.

\paragraph{Part II: two-scale chain for the small-$\kappa$ branch.}
The weighted construction requires a geometric ratio $r_{\rm w}\in(0,1)$ with
$1-r_{\rm w}=\Theta(\sqrt n/\kappa_{\max})$.  This choice is not admissible when
$\kappa_{\max}<\sqrt n$.  In this range, we instead separate the gate scale
$R_j$ from the contribution scale $W_j$.  The construction can then use the
amplitude ratio $\vartheta=\Theta(\kappa_{\max}/\sqrt n)$ while preserving
sequential hiding and individual smoothness.

\subsection{Two-scale chain: construction}
\label{app:two-scale-construction}

Fix an integer $M_{\rm ts}\ge2$; the choice depending on accuracy is made only
after the analytic properties of this family have been established.

\paragraph{Bias and effective curvature.}
Use the optimized admissible bias from Corollary~\ref{cor:optimized-near-balanced-bias}:
$m$ is the least parity-compatible integer not smaller than $\sqrt n$, so
\begin{equation}
 \sqrt n\le m<\sqrt n+2,
 \qquad \rho=\frac mn,
 \qquad \frac1{\sqrt n}\le\rho<\frac{17}{16\sqrt n}.
 \label{eq:two-scale-local-optimized-bias}
\end{equation}
For $\kappa_{\max}\ge3$, define
\begin{equation}
 \bar\kappa_{\max}:=\min\{\kappa_{\max},\sqrt n\},
 \qquad
 \bar\mu:=\frac{L_{\max}}{\bar\kappa_{\max}}.
 \label{eq:two-scale-effective-parameters}
\end{equation}
Since $\bar\mu\ge\mu=L_{\max}/\kappa_{\max}$, every globally
$\bar\mu$-strongly convex objective is globally $\mu$-PL.

\paragraph{Geometrically decreasing amplitudes and two spatial scales.}
Define the positive stage amplitudes by
\begin{equation}
 \vartheta:=\frac{\bar\kappa_{\max}\rho}{8192},
 \qquad
 \zeta_1:=\sqrt{2\bar\mu\Delta},
 \qquad
 \zeta_j:=\zeta_1\vartheta^{j-1}.
 \label{eq:two-scale-amplitudes}
\end{equation}
For every stage, define the gate scale $R_j$ and contribution scale $W_j$ by
\begin{equation}
 R_j:=\frac{\zeta_j}{\bar\mu},
 \qquad
 W_j:=\frac{64\zeta_j}{L_{\max}\rho}.
 \label{eq:two-scale-radii}
\end{equation}
The ratio $R_j/W_j$ is independent of $j$ and satisfies
\begin{equation}
 \delta_{\rm rad}:=\frac{R_j}{W_j}
 =\frac{\bar\kappa_{\max}\rho}{64}
 <\frac{17}{1024},
 \label{eq:two-scale-radius-ratio}
\end{equation}
where we used $\bar\kappa_{\max}\le\sqrt n$ and the local bias bound
in~\eqref{eq:two-scale-local-optimized-bias}.

\paragraph{Gate, component term, and chain.}
For $x=(x_1,\ldots,x_{M_{\rm ts}})$ with $x_j\in\R^D$, define the gate
coordinate
\begin{equation}
 h_j(x_j):=q_{\Theta^{(j)}}(x_j/R_j)
 \label{eq:two-scale-gate-coordinate}
\end{equation}
and the component term
\begin{equation}
 P_{ij}(x_j)
 :=\frac{\zeta_jW_j}{\rho}
 \left\langle\frac{S^{(j)}_{i:}}{\sqrt D},
 \Psi(x_j/W_j)\right\rangle.
 \label{eq:two-scale-component-payload}
\end{equation}
By~\eqref{eq:exact-code-average}, its average is exactly
\begin{equation}
 P_j(x_j):=\frac1n\sum_{i=1}^nP_{ij}(x_j)
 =\zeta_jW_j q_{\Theta^{(j)}}(x_j/W_j).
 \label{eq:two-scale-average-payload}
\end{equation}
With $M_{\rm ts}$ stages, define
\begin{equation}
\begin{aligned}
 f_i^{\rm ts}(x)
 ={}&\bar\mu\sum_{j=1}^{M_{\rm ts}}\|x_j\|^2-P_{i1}(x_1)\\
 &-\sum_{j=2}^{M_{\rm ts}}
 G(h_{j-1}(x_{j-1}))P_{ij}(x_j),
\end{aligned}
\label{eq:two-scale-component-chain}
\end{equation}
and let $F_{\rm ts}=n^{-1}\sum_i f_i^{\rm ts}$.  The exact table average gives
\begin{equation}
 F_{\rm ts}(x)
 =\bar\mu\sum_{j=1}^{M_{\rm ts}}\|x_j\|^2-P_1(x_1)
 -\sum_{j=2}^{M_{\rm ts}}
 G(h_{j-1}(x_{j-1}))P_j(x_j).
 \label{eq:two-scale-average-chain}
\end{equation}
All components are $C^\infty$.

\Needspace{34\baselineskip}
\begin{proposition}[Properties of the two-scale chain]
\label{prop:two-scale-chain-properties}
For every $M_{\rm ts}\ge2$, the construction above satisfies:
\begin{enumerate}[label=\textup{(\roman*)},leftmargin=*,itemsep=2pt]
 \item\label{item:two-scale-property-smoothness}
 Every component $f_i^{\rm ts}$ is individually $L_{\max}$-smooth.
 \item\label{item:two-scale-property-convexity}
 The average satisfies
 \begin{equation}
  \nabla^2F_{\rm ts}(x)
  \succeq\bar\mu I,
  \label{eq:two-scale-property-convexity}
 \end{equation}
 and hence is globally $\bar\mu$-strongly convex and $\mu$-PL.  When
 $3\le\kappa_{\max}\le\sqrt n$, one has $\bar\mu=\mu$, so
 $F_{\rm ts}$ is $\mu$-strongly convex.
 \item\label{item:two-scale-property-minimizer}
 Its unique minimizer is
 \begin{equation}
  x_j^*=R_jz_*u_{\Theta^{(j)}},
  \qquad
  2z_*=(1+\delta_{\rm rad}^2z_*^2)^{-3/2},
  \label{eq:two-scale-property-minimizer}
 \end{equation}
 where $z_*>1/3$ is the unique positive solution; every optimal gate satisfies
 $h_j(x_j^*)>b$.
 \item\label{item:two-scale-property-gaps}
 The initial gap and the gap at the final stage satisfy
 \begin{equation}
  F_{\rm ts}(0)-F_{\rm ts}^*\le\Delta,
  \label{eq:two-scale-property-initial-gap}
 \end{equation}
 and
 \begin{equation}
  h_{M_{\rm ts}}(x_{M_{\rm ts}})\le a
  \quad\Longrightarrow\quad
  F_{\rm ts}(x)-F_{\rm ts}^*
  \ge\frac{\zeta_{M_{\rm ts}}^2}{128\bar\mu}.
  \label{eq:two-scale-property-terminal-gap}
 \end{equation}
\end{enumerate}
\end{proposition}

For $3\le\kappa_{\max}\le\sqrt n$, one has
$\bar\kappa_{\max}=\kappa_{\max}$ and $\rho\asymp n^{-1/2}$, so
$R_j/W_j\asymp\kappa_{\max}/\sqrt n$.  If the stage contribution used the
short gate scale $R_j$ as well, its basic component Hessian estimate would be
$2\zeta_j/(\rho R_j)=2L_{\max}/(\bar\kappa_{\max}\rho)>L_{\max}$, already
exceeding the smoothness budget for each component.  The larger contribution scale
$W_j$ reduces this Hessian contribution to the required scale.  The choice
$R_j=\zeta_j/\bar\mu$ sets the displacement needed to activate the next stage.
Figure~\ref{fig:two-scale-radii} shows the conflict between hiding and
individual smoothness, and how the separation $R_j\ll W_j$ resolves it.

\begin{figure}[H]
\centering
\includegraphics[width=\linewidth]{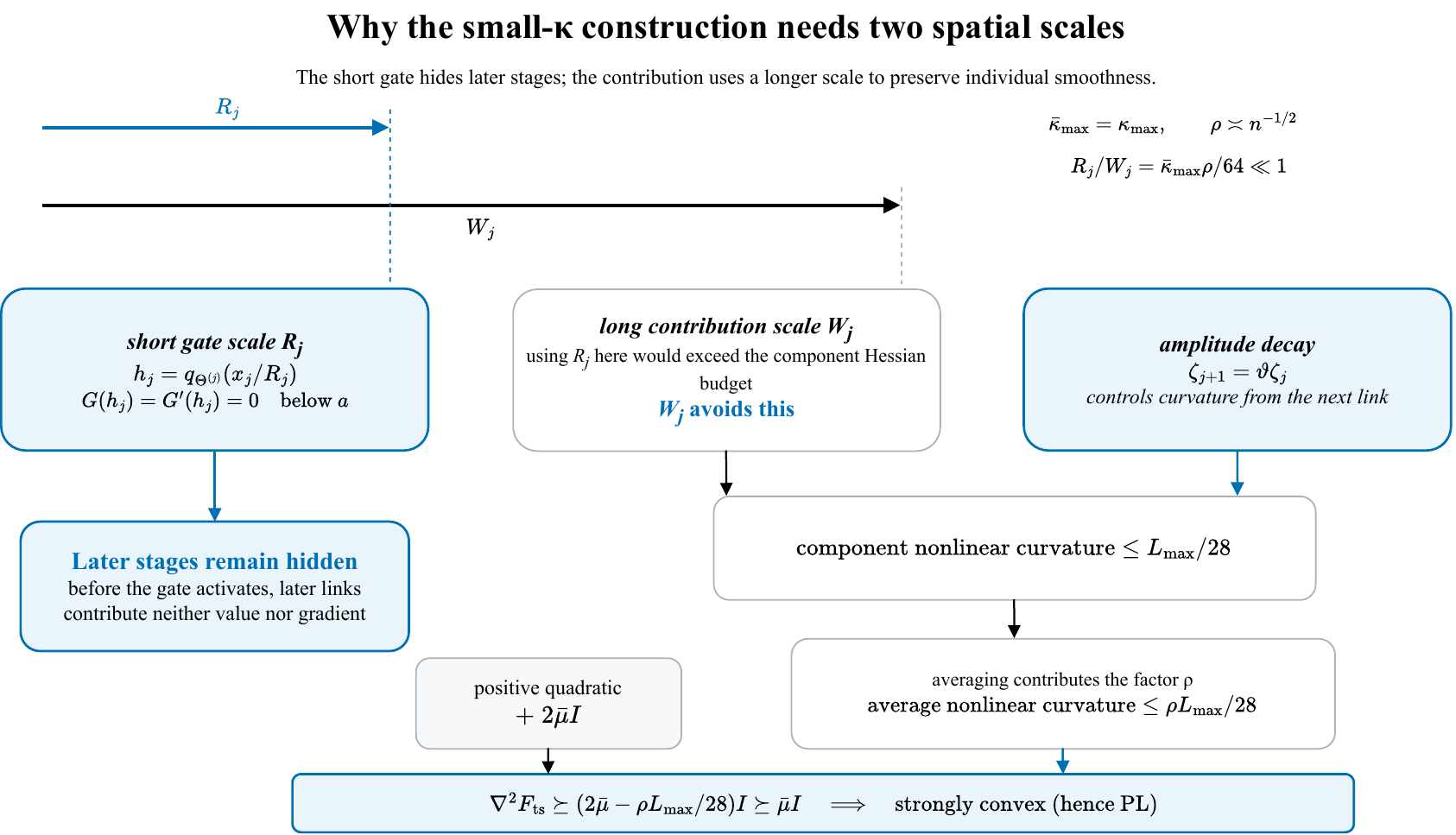}
\caption{Why two spatial scales are needed in the small-$\kappa$ construction.
The short scale $R_j$ makes the gate locally flat before it opens and thereby
hides later stages.  Using that same scale for the stage contribution would
exceed the $L_{\max}$ Hessian budget of a component.  The longer scale $W_j$,
together with the amplitude decay, reduces the nonlinear component curvature
to at most $L_{\max}/28$.  Averaging contributes the factor $\rho$, and the
positive quadratic then gives $\nabla^2F_{\rm ts}\succeq\bar\mu I$.}
\label{fig:two-scale-radii}
\end{figure}

\subsection{Two-scale chain: analytic properties}

\begin{lemma}[Nonlinear Hessian bound]
\label{lem:two-scale-curvature-budget}
Define
\begin{equation}
 B_\star:=\frac1{28}.
 \label{eq:two-scale-B-star}
\end{equation}
Then the nonlinear Hessian of every component in~\eqref{eq:two-scale-component-chain}
has operator norm at most $B_\star L_{\max}$.  After averaging the
components, the corresponding bound is $\rho B_\star L_{\max}$.
\end{lemma}

\begin{proof}
\emph{Scaled derivatives.}
Lemma~\ref{lem:radial-first-second} and the chain rule give
\begin{align}
 \|\nabla h_j\|&\le\frac1{R_j},
 &\|\nabla^2h_j\|_{\op}&\le\frac2{R_j^2},\nonumber\\
 |P_{ij}|&\le\frac{\zeta_jW_j}{\rho},
 &\|\nabla P_{ij}\|&\le\frac{\zeta_j}{\rho},
 &\|\nabla^2P_{ij}\|_{\op}&\le\frac{2\zeta_j}{\rho W_j}.
 \label{eq:two-scale-scaled-derivatives}
\end{align}
The corresponding bounds for $P_j$ omit $1/\rho$.

\emph{Local link blocks.}
For the link
$\mathcal L_{ij}=-G(h_{j-1})P_{ij}$, substitute these bounds into
Lemma~\ref{lem:gated-link-hessian}.  The actual scale choices make the current
diagonal, both cross blocks together, and the previous diagonal at most
\begin{align}
 B_2
 &=\frac{2\zeta_j}{\rho W_j}
 =\frac{L_{\max}}{32},\nonumber\\
 2M_1g_1B_1
 &=\frac{2M_1\zeta_j}{\rho R_{j-1}}
 =\frac{2M_1L_{\max}}{8192}
 \le\frac{L_{\max}}{240},\nonumber\\
 B_0(M_2g_1^2+M_1g_2)
 &=\frac{\zeta_jW_j(M_2+2M_1)}{\rho R_{j-1}^2}
 =\frac{64(M_2+2M_1)L_{\max}}{8192^2}
 \le\frac{L_{\max}}{3500}.
 \label{eq:two-scale-local-curvatures}
\end{align}
Here~\eqref{eq:scaled-gate-derivatives} bounds $M_1$ and $M_2$.

With $\delta_{\rm rad}=R_j/W_j=128\vartheta$, all three contributions
in~\eqref{eq:two-scale-local-curvatures} are bounded by universal fractions of
the $L_{\max}$ smoothness budget.  Their sum is below $B_\star L_{\max}$ because
\begin{equation}
 \frac1{32}+\frac1{240}+\frac1{3500}<\frac1{28}=B_\star.
 \label{eq:two-scale-curvature-fractions}
\end{equation}

\emph{Global block row and averaging.}
Each block row contains at most one current diagonal, one previous diagonal,
and two cross blocks.  The ungated first stage contribution $-P_{i1}$ obeys the same
bound for the current diagonal block.  Hence the block row sum is at most
$B_\star L_{\max}$, and the block row clause of
Lemma~\ref{lem:gated-link-hessian} gives the component bound.  For every
$2\le j\le M_{\rm ts}$, the exact average over the fixed table gives
\begin{equation}
 \frac1n\sum_{i=1}^n
 \nabla^2\!\left[-G(h_{j-1})P_{ij}\right]
 =\nabla^2\!\left[-G(h_{j-1})P_j\right].
 \label{eq:two-scale-average-link-hessian}
\end{equation}
The ungated term at the first stage obeys the analogous identity.  Moreover,
\begin{equation}
 |P_j|\le\zeta_jW_j,
 \qquad \|\nabla P_j\|\le\zeta_j,
 \qquad \|\nabla^2P_j\|_{\op}\le\frac{2\zeta_j}{W_j},
 \label{eq:two-scale-average-derivatives}
\end{equation}
so the value, gradient, and Hessian bounds for $P_j$ are each $\rho$ times the
corresponding component bounds in~\eqref{eq:two-scale-scaled-derivatives}.
Thus every nonlinear block bound is multiplied by $\rho$.
\end{proof}

\begin{proof}[Proof of Proposition~\ref{prop:two-scale-chain-properties},
items~\ref{item:two-scale-property-smoothness}--\ref{item:two-scale-property-convexity}]
Since $\bar\kappa_{\max}\ge3$,
\begin{equation}
 \frac{2\bar\mu+B_\star L_{\max}}{L_{\max}}
 =\frac2{\bar\kappa_{\max}}+B_\star
 \le\frac23+\frac1{28}<\frac34<1,
\end{equation}
which proves individual smoothness.  This is the reason for the convenient
threshold $\kappa_{\max}\ge3$: the quadratic and nonlinear terms fit within
the component smoothness bound with a uniform margin.  This threshold is not
optimized.  For the average,
$\bar\kappa_{\max}\rho<17/16$ gives
\begin{equation}
 \nabla^2F_{\rm ts}(x)
 \succeq
 \left(2-\frac{17}{16}\frac1{28}\right)\bar\mu I
 \succeq\bar\mu I,
\end{equation}
where $2-17/(16\cdot28)>1$.
Finally, $\bar\mu\ge\mu$ implies the declared $\mu$-PL property.
\end{proof}

\begin{proof}[Proof of Proposition~\ref{prop:two-scale-chain-properties},
item~\ref{item:two-scale-property-minimizer}]
\emph{Scalar stationary point.}
In the scalar equation in~\eqref{eq:two-scale-property-minimizer}, the left
side increases strictly
from zero and the right side decreases strictly from one.  At $z=1/3$,
\begin{equation}
 (1+\delta_{\rm rad}^2/9)^{-3/2}>(10/9)^{-3/2}>2/3,
\end{equation}
where the final comparison follows by cubing and squaring positive
quantities: $1000/729<9/4$.  Thus the unique crossing has $z_*>1/3$.

\emph{Average minimizer and gate values.}
At $x_j=R_jzu_{\Theta^{(j)}}$, direct differentiation gives
\begin{equation}
 \nabla P_j(x_j)
 =\zeta_j(1+\delta_{\rm rad}^2z^2)^{-3/2}u_{\Theta^{(j)}}.
\end{equation}
Because $2\bar\mu R_jz=2\zeta_jz$,
\eqref{eq:two-scale-property-minimizer} is exactly the stationarity
condition for $\bar\mu\|x_j\|^2-P_j(x_j)$.  Moreover,
\begin{equation}
 h_j(x_j^*)=\frac{z_*}{\sqrt{1+z_*^2}}
 >\frac1{\sqrt{10}}>\frac9{32}=b.
\end{equation}
Every predecessor coefficient therefore equals one, while every derivative
of a successor gate vanishes when the gate is fully open.  Hence
$\nabla F_{\rm ts}(x^*)=0$, and uniqueness follows from
Proposition~\ref{prop:two-scale-chain-properties},
item~\ref{item:two-scale-property-convexity}.
\end{proof}

\begin{proof}[Proof of Proposition~\ref{prop:two-scale-chain-properties},
item~\ref{item:two-scale-property-gaps}]
\emph{Initial gap.}
At the origin, $h_j(0)=0\le a$, so all successor gates satisfy $G=G'=0$ and
$\nabla F_{\rm ts}(0)=-\zeta_1u_{\Theta^{(1)}}$.  Strong convexity implies
\begin{equation}
 F_{\rm ts}(0)-F_{\rm ts}^*
 \le\frac{\|\nabla F_{\rm ts}(0)\|^2}{2\bar\mu}
 =\frac{\zeta_1^2}{2\bar\mu}=\Delta.
\end{equation}

\emph{Gap at the final stage.}
By item~\ref{item:two-scale-property-minimizer},
$h_{M_{\rm ts}}(x_{M_{\rm ts}}^*)>b$.  The gate coordinate is
$1/R_{M_{\rm ts}}$-Lipschitz, so
\begin{equation}
 \|x_{M_{\rm ts}}-x_{M_{\rm ts}}^*\|
 \ge(b-a)R_{M_{\rm ts}}=\frac{R_{M_{\rm ts}}}{8}.
\end{equation}
Using $\bar\mu$-strong convexity and
$R_{M_{\rm ts}}=\zeta_{M_{\rm ts}}/\bar\mu$ proves
\eqref{eq:two-scale-property-terminal-gap}.  In particular,
$\zeta_{M_{\rm ts}}^2\ge1024\bar\mu\varepsilon$ implies a gap of at least
$8\varepsilon$.
\end{proof}

\subsection{Two-scale chain: sequential hiding and stage count}

The two-scale chain gives
$\boldsymbol{\Omega\!\left(n\log(\Delta/\varepsilon)/
[1+\log(\sqrt n/\bar\kappa_{\max})]\right)}$ when there are multiple stages; the
auxiliary quadratic gives $\boldsymbol{\Omega(n)}$ otherwise.

\paragraph{Truncated construction at each progress level.}
At progress $g\in\{0,\ldots,M_{\rm ts}\}$, retain links only through the
current stage contribution:
\begin{equation}
\begin{aligned}
 (f_i^{\rm ts})^{[g]}(x)
 ={}&\bar\mu\sum_{j=1}^{M_{\rm ts}}\|x_j\|^2-P_{i1}(x_1)\\
 &-\sum_{j=2}^{\min\{g+1,M_{\rm ts}\}}
 G(h_{j-1}(x_{j-1}))P_{ij}(x_j).
\end{aligned}
\label{eq:two-scale-truncation}
\end{equation}
With $V_j(x_j)=\Psi(x_j/R_j)$, table independence, the dependence of each
component on its assigned table row, the condition $G=G'=0$ for omitted gates,
and the call accounting rule give \emph{(R1)--(R4)} of
Corollary~\ref{cor:optimized-radial-chain}.

Let
\begin{equation}
 s_{\rm ts}:=\log\frac1\vartheta
 =\log\frac{8192}{\bar\kappa_{\max}\rho},
 \qquad
 \ell_\kappa:=1+\log\frac{\sqrt n}{\bar\kappa_{\max}}.
 \label{eq:two-scale-log-scales}
\end{equation}
Write $\ell_\varepsilon:=\log(\Delta/\varepsilon)$.  For the case with multiple stages,
assume
\begin{equation}
 \ell_\varepsilon\ge2\log512+4s_{\rm ts}
 \label{eq:two-scale-high-precision}
\end{equation}
and set
\begin{equation}
 J:=\left\lfloor
 \frac{\ell_\varepsilon-\log512}{2s_{\rm ts}}
 \right\rfloor,
 \qquad M_{\rm ts}:=J+1.
 \label{eq:two-scale-stage-choice}
\end{equation}

\begin{lemma}[Logarithmic scale, stage count, and scale at the final stage]
\label{lem:two-scale-stage-count}
The scales in~\eqref{eq:two-scale-log-scales} satisfy
\begin{equation}
 \ell_\kappa<s_{\rm ts}<10\ell_\kappa.
 \label{eq:two-scale-log-comparison}
\end{equation}
Under~\eqref{eq:two-scale-high-precision}, moreover,
\begin{equation}
 J\ge\frac{\ell_\varepsilon}{4s_{\rm ts}}
 \ge\frac{\ell_\varepsilon}{40\ell_\kappa},
 \qquad
 \zeta_{M_{\rm ts}}^2\ge1024\bar\mu\varepsilon.
 \label{eq:two-scale-stage-residual}
\end{equation}
\end{lemma}

\begin{proof}
\emph{Compare the logarithmic scales.}
The upper bound on $\rho$ in~\eqref{eq:two-scale-local-optimized-bias} gives
\begin{equation}
 s_{\rm ts}>
 \log\frac{8192\cdot16}{17}
 +\log\frac{\sqrt n}{\bar\kappa_{\max}}>\ell_\kappa.
\end{equation}
The lower bound $\rho\ge n^{-1/2}$ gives
\begin{equation}
 s_{\rm ts}\le\log8192+\log\frac{\sqrt n}{\bar\kappa_{\max}}
 <10+(\ell_\kappa-1)\le10\ell_\kappa.
\end{equation}

\emph{Stage count and scale at the final stage.}
The inequality $\lfloor z\rfloor\ge z-1$ and
\eqref{eq:two-scale-high-precision} give
\begin{equation}
 J\ge
 \frac{\ell_\varepsilon-\log512-2s_{\rm ts}}{2s_{\rm ts}}
 \ge\frac{\ell_\varepsilon}{4s_{\rm ts}}.
\end{equation}
The upper comparison $s_{\rm ts}<10\ell_\kappa$ gives the second bound on the number of stages.
The upper bound on the floor in~\eqref{eq:two-scale-stage-choice} yields
\begin{equation}
 \zeta_{M_{\rm ts}}^2
 =2\bar\mu\Delta e^{-2s_{\rm ts}J}
 \ge2\bar\mu\Delta e^{-\ell_\varepsilon+\log512}
 =1024\bar\mu\varepsilon.
\end{equation}
\end{proof}

\begin{proposition}[Bound for multiple stages of the two-scale construction]
\label{prop:two-scale-high-precision}
Let
$\ell_\varepsilon=\log(\Delta/\varepsilon)$,
$s_{\rm ts}=\log(8192/(\bar\kappa_{\max}\rho))$, and
$\ell_\kappa=1+\log(\sqrt n/\bar\kappa_{\max})$.
Assume
$\ell_\varepsilon\ge2\log512+4s_{\rm ts}$, and set
$c_{\rm ts}:=1/163840$.
Then every deterministic algorithm whose IFO call budget $N$ satisfies
\begin{equation}
 N\le c_{\rm ts}\frac{n\ell_\varepsilon}{\ell_\kappa}
 \label{eq:two-scale-query-budget}
\end{equation}
fails with probability at least $95/128$ under the two-scale hard
distribution.  Specifically,
\begin{equation}
 F_{\rm ts}(\widehat X)-F_{\rm ts}^*\ge8\varepsilon.
\end{equation}
\end{proposition}

\begin{proof}
\emph{Dimension and sequential event.}
The conclusion is immediate for $N=0$, so assume $N\ge1$.
Take $r=\lfloor n/512\rfloor$, so $r\ge n/1024$ for $n\ge1024$.
Choose
\begin{equation}
 D\ge\max\left\{
 2048[1+\log(N+3)],\quad
  \frac{2048}{25}\log\bigl(128(N+1)M_{\rm ts}\bigr)
 \right\}.
 \label{eq:two-scale-sequential-dimension}
\end{equation}
Corollary~\ref{cor:optimized-radial-chain}, with
$\operatorname{align}_j=h_j$, $M_{\rm tot}=M_{\rm ts}$, and
$J=M_{\rm ts}-1$, shows that whenever $N\le rJ/2$, with probability
at least $95/128$ the output of the original IFO interaction satisfies
$h_{M_{\rm ts}}(\widehat X_{M_{\rm ts}})\le a$.

\emph{Gap at the final stage.}
On this event, Lemma~\ref{lem:two-scale-stage-count} and
Proposition~\ref{prop:two-scale-chain-properties},
item~\ref{item:two-scale-property-gaps}, give
\begin{equation}
 F_{\rm ts}(\widehat X)-F_{\rm ts}^*
 \ge\frac{\zeta_{M_{\rm ts}}^2}{128\bar\mu}
 \ge8\varepsilon.
\end{equation}

\emph{Query cost.}
Finally, the same lemma on the stage count gives
$J\ge\ell_\varepsilon/(40\ell_\kappa)$, and hence
\begin{equation}
 \frac{rJ}{2}
 \ge\frac{n}{2048}\frac{\ell_\varepsilon}{40\ell_\kappa}
 =\frac{n\ell_\varepsilon}{81920\ell_\kappa},
\end{equation}
Since $c_{\rm ts}=1/163840$, the budget in
\eqref{eq:two-scale-query-budget} therefore ensures $N\le rJ/2$.
\end{proof}

\begin{corollary}[Two-scale PL lower bound]
\label{cor:two-scale-pl-lower}
For every $n\ge n_0=1024$, $\kappa_{\max}\ge3$, and
$0<\varepsilon\le c\Delta$, every randomized algorithm with IFO access
using fewer than
\begin{equation}
 c'\left[
 n+\frac{n}{1+\log(\sqrt n/\bar\kappa_{\max})}
 \log\frac{\Delta}{\varepsilon}
 \right],
 \qquad \bar\kappa_{\max}=\min\{\kappa_{\max},\sqrt n\},
 \label{eq:two-scale-full-lower}
\end{equation}
calls fails with probability at least $11/16$ on some individually
$L_{\max}$-smooth finite sum with initial gap at most $\Delta$ and globally
$\mu$-PL average objective.
\end{corollary}

\begin{proof}
Choose $c\le1/128$ and $c'\le(4/5)c_{\rm ts}=1/204800$.
When the condition for multiple stages holds,
we have $\ell_\varepsilon\ge4s_{\rm ts}>4\ell_\kappa$, so
\[
 c'\left(n+\frac{n\ell_\varepsilon}{\ell_\kappa}\right)
 <\frac54c'\frac{n\ell_\varepsilon}{\ell_\kappa}
 \le c_{\rm ts}\frac{n\ell_\varepsilon}{\ell_\kappa}.
\]
Thus Proposition~\ref{prop:two-scale-high-precision} and the conversion in
Section~\ref{sec:deterministic-to-randomized} give a fixed instance with
failure probability at least $95/128$.

When this condition fails, suppose
$\ell_\varepsilon<2\log512+4s_{\rm ts}$.  Since
$s_{\rm ts}<10\ell_\kappa$,
\begin{equation}
 \ell_\varepsilon<2\log512+4s_{\rm ts}<54\ell_\kappa,
 \qquad
 n+\frac{n\ell_\varepsilon}{\ell_\kappa}<55n.
 \label{eq:two-scale-coarse-log}
\end{equation}
Our choices also give $c'\le1/220$ and $\varepsilon\le\Delta/128$.
Every integer budget satisfying
\begin{equation}
 N<c'\left(n+\frac{n\ell_\varepsilon}{\ell_\kappa}\right)
 <55c'n\le\frac n4
\end{equation}
falls under Lemma~\ref{lem:pl-sample-size}, which gives failure probability at
least $11/16$.  The two cases prove~\eqref{eq:two-scale-full-lower}.
\end{proof}

\paragraph{Conclusion for the two-scale chain.}
In the case with multiple stages, Proposition~\ref{prop:two-scale-chain-properties} and
Proposition~\ref{prop:two-scale-high-precision} establish
\begin{equation*}
\begin{aligned}
 &\max_i\operatorname{Lip}(\nabla f_i^{\rm ts})\le L_{\max},
 \qquad \nabla^2F_{\rm ts}\succeq\bar\mu I,
 \qquad F_{\rm ts}(0)-F_{\rm ts}^*\le\Delta,\\
 &N\le\frac{c_{\rm ts}n\ell_\varepsilon}{\ell_\kappa}
 \ \Longrightarrow\ 
 \Prob\bigl(F_{\rm ts}(\widehat X)-F_{\rm ts}^*\ge8\varepsilon\bigr)
 \ge\frac{95}{128}.
\end{aligned}
\end{equation*}
The auxiliary quadratic covers the remaining $O(n)$ case, yielding
Corollary~\ref{cor:two-scale-pl-lower}.

\phantomsection
\subsection*{Assembly across all condition-number regimes}
\addcontentsline{toc}{subsection}{Assembly across all condition-number regimes}

We use the two-scale branch below $\sqrt n$, the weighted branch above
$2\Gamma\sqrt n$, and the two-scale construction again on the bounded
transition interval between them.

Each branch with multiple stages costs $\Theta(n)$ calls per stage.  Multiplying by its
stage count and adding the auxiliary $\Omega(n)$ term gives the bounds below.

\emph{Small $\kappa_{\max}$: $3\le\kappa_{\max}\le\sqrt n$.}
Here $\bar\kappa_{\max}=\kappa_{\max}$, and
Corollary~\ref{cor:two-scale-pl-lower} gives
\begin{equation}
 \Omega\!\left(
 n+\frac{n\log(\Delta/\varepsilon)}
 {1+\log(\sqrt n/\kappa_{\max})}
 \right),
 \label{eq:restated-two-scale-small-regime}
\end{equation}
which is the small-$\kappa_{\max}$ regime of Theorem~\ref{thm:pl-main}.

\emph{Range covered by the weighted chain: $\kappa_{\max}\ge2\Gamma\sqrt n$.}
The weighted chain proves
\begin{equation}
 \Omega\!\left(
 n+\kappa_{\max}\sqrt n\log\frac{\Delta}{\varepsilon}
 \right)
 \label{eq:weighted-large-branch}
\end{equation}
whenever $\kappa_{\max}\ge\kappa_{\rm w}(n)=2\Gamma\sqrt n$.

\emph{Bounded transition: $\sqrt n\le\kappa_{\max}<2\Gamma\sqrt n$.}
Here $\bar\kappa_{\max}=\sqrt n$ and $\bar\mu\ge\mu$, so the two-scale
instance still belongs to the required $\mu$-PL class.  The corollary gives
\begin{equation}
 \Omega\!\left(n+n\log\frac{\Delta}{\varepsilon}\right).
 \label{eq:two-scale-large-fallback}
\end{equation}
Within the displayed transition interval,
$\kappa_{\max}\sqrt n<2\Gamma n$, so
\eqref{eq:two-scale-large-fallback} is at least a universal
$1/(2\Gamma)$ multiple of~\eqref{eq:weighted-large-branch}.  Thus, after
shrinking the universal theorem constant, the large-$\kappa$ bound
holds for every $\kappa_{\max}\ge\sqrt n$.

\emph{Failure probability and dimension.}
The arguments for multiple stages and for the complementary case have failure probabilities at least
$95/128$ and $11/16$, respectively.  The
dimension choice above permits
\begin{equation}
 d=M_{\rm ts}D
 =O\!\left[
 \left(1+
 \frac{\ell_\varepsilon}
 {1+\log(\sqrt n/\bar\kappa_{\max})}\right)
 \log\!\left(
 2+n+
 \frac{n\ell_\varepsilon}
 {1+\log(\sqrt n/\bar\kappa_{\max})}
 \right)
 \right]
 \label{eq:two-scale-dimension}
\end{equation}
for the two-scale construction.  For the weighted chain,
\begin{equation}
 d=O\!\left[
 \frac{\kappa_{\max}}{\sqrt n}
 \log\frac{\Delta}{\varepsilon}
 \log\!\left(
 2+n+\kappa_{\max}\sqrt n\log\frac{\Delta}{\varepsilon}
 \right)
 \right],
 \label{eq:restated-weighted-dimension-summary}
\end{equation}
as established in~\eqref{eq:pl-dimension}.  The auxiliary quadratic
uses only $D_0=4$.

\section{Restarted PAGE Upper Bound}
\label{app:page-upper}

\paragraph{Restatement of the upper bound in
Theorem~\ref{thm:pl-main}.}
Let $F=n^{-1}\sum_{i=1}^nf_i$ satisfy mean-squared smoothness with constant
$L_{\rm ms}$, the global $\mu$-PL inequality, and
$F(x^{(0)})-F^*\le\Delta$.  Write
$\kappa_{\rm ms}=L_{\rm ms}/\mu\ge1$.  For
$0<\varepsilon\le c\Delta$, Algorithm~\ref{alg:restarted-page} returns
$x^{(S_{\rm ep})}$ satisfying
\begin{equation}
 \E\bigl[F(x^{(S_{\rm ep})})-F^*\bigr]\le\varepsilon
 \label{eq:restated-page-target}
\end{equation}
with expected IFO cost bounded, for a universal constant $C>0$, by
\begin{equation}
 \E[\#\text{ IFO calls}]
 \le C\begin{cases}
 \displaystyle
 n+\dfrac{n\log(\Delta/\varepsilon)}
 {1+\log(\sqrt n/\kappa_{\rm ms})},
 &1\le\kappa_{\rm ms}\le\sqrt n,\\[2mm]
 \displaystyle
 n+\kappa_{\rm ms}\sqrt n
 \log\frac{\Delta}{\varepsilon},
 &\kappa_{\rm ms}\ge\sqrt n.
 \end{cases}
 \label{eq:restated-page-cost}
\end{equation}
For $\kappa_{\rm ms}\ge3$, this is the upper bound in
Theorem~\ref{thm:pl-main}.  The same algorithm and proof also cover
$1\le\kappa_{\rm ms}<3$.

\paragraph{Epoch length and the rate for small condition numbers.}
At the optimized bias, Appendix~\ref{app:information} shows that each additional
hidden stage costs $\Omega(n)$ IFO calls.  Appendix~\ref{app:pl} shows that the
gap in the hard instance decreases by
$\vartheta^2=\Theta((\kappa_{\max}/\sqrt n)^2)$ per stage.  Consequently, reducing
the gap from $\Delta$ to $\varepsilon$ requires
$\Theta(\log(\Delta/\varepsilon)/[1+\log(\sqrt n/\kappa_{\max})])$ stages.

Lemma~\ref{lem:page-epoch-contraction} gives
$q_T=\Theta(\kappa_{\rm ms}/\sqrt n)$ when $T=\Theta(n)$, at $\Theta(n)$
expected cost.  The estimator is PAGE~\citep{li2021page}, and
Proposition~\ref{prop:restarted-page} chooses the epoch length in both
$\kappa_{\rm ms}$ regimes.
Appendix~\ref{app:stopping} gives the conversion to a fixed call budget.

\subsection{Contraction and cost over one epoch}

One epoch has expected IFO cost at most $n+3T$ and bounds the expected PL gap
by $q_T=\kappa_{\rm ms}(1+\sqrt n)/T$ times the starting gap.  The restart
choice below makes $q_T\le1/4$.

Assume
\begin{equation}
 \frac1n\sum_{i=1}^n
 \|\nabla f_i(x)-\nabla f_i(y)\|^2
 \le L_{\rm ms}^2\|x-y\|^2
 \qquad\forall x,y,
 \label{eq:page-ms-premise}
\end{equation}
and write $\kappa_{\rm ms}=L_{\rm ms}/\mu$ when $F$ is $\mu$-PL.
Algorithm~\ref{alg:restarted-page} uses
\begin{equation}
 p=\frac1{n+1},
 \qquad
 \eta=\frac1{L_{\rm ms}(1+\sqrt n)}.
 \label{eq:page-epoch-parameters}
\end{equation}
Each refresh uses all $n$ components, while each difference update samples one
component.
The probability $p$ balances an $n$-component refresh against a one-component
difference update, which keeps the expected update cost constant.  The step
size $\eta$ satisfies the PAGE stability condition used below.
At the epoch start $x$, compute the exact full gradient, run $T$ PAGE
updates, and choose the epoch output uniformly from the iterates used in the
standard nonconvex guarantee.

\begin{lemma}[PAGE epoch contraction]
\label{lem:page-epoch-contraction}
Condition on an arbitrary epoch starting point $x$, and let $\widehat x$ be
the random epoch output.  Then
\begin{equation}
 \E\!\left[\|\nabla F(\widehat x)\|^2\mid x\right]
 \le\frac{2(F(x)-F^*)}{\eta T}.
 \label{eq:page-epoch-gradient}
\end{equation}
If $F$ is globally $\mu$-PL, then
\begin{equation}
 \E\!\left[F(\widehat x)-F^*\mid x\right]
 \le q_T(F(x)-F^*),
 \qquad
 q_T:=\frac{\kappa_{\rm ms}(1+\sqrt n)}{T}.
 \label{eq:page-epoch-gap}
\end{equation}
The expected number of IFO calls in the epoch is at most $n+3T$.
\end{lemma}

\begin{proof}
\emph{Lyapunov descent.}
For the analysis, extend the estimator recurrence by one auxiliary update to
define $g_T$.  Algorithm~\ref{alg:restarted-page} need not compute it: the epoch
output uses only $x_0,\ldots,x_{T-1}$, and the next epoch starts with a full
gradient.
Let $e_t:=g_t-\nabla F(x_t)$ and let $\E_t$ condition on the epoch history
through the construction of $g_t$.  Objective smoothness and the PAGE estimator
update give the two standard recursions
\begin{align}
 \E_t[F(x_{t+1})]
 &\le F(x_t)-\frac\eta2\|\nabla F(x_t)\|^2
 +\frac\eta2\|e_t\|^2
 -\left(\frac1{2\eta}-\frac{L_{\rm ms}}2\right)
  \|x_{t+1}-x_t\|^2,
 \label{eq:page-descent-recursion}\\
 \E_t\|e_{t+1}\|^2
 &\le(1-p)\|e_t\|^2
 +(1-p)L_{\rm ms}^2\|x_{t+1}-x_t\|^2.
 \label{eq:page-variance-recursion}
\end{align}
The second inequality follows by conditioning on the refresh and difference
updates and applying mean-squared smoothness; these are the PAGE recursions of
\citet[proof of Theorem~1]{li2021page}.

With an exact full-gradient initialization, define the Lyapunov function
\begin{equation}
 \mathcal L_t=F(x_t)-F^*
 +\frac{\eta}{2p}\|g_t-\nabla F(x_t)\|^2.
\end{equation}
Multiplying~\eqref{eq:page-variance-recursion} by $\eta/(2p)$, adding
\eqref{eq:page-descent-recursion}, and taking expectations gives
\begin{equation*}
 \E[\mathcal L_{t+1}]
 \le\E[\mathcal L_t]-\frac\eta2\E\|\nabla F(x_t)\|^2
 -c_\eta\E\|x_{t+1}-x_t\|^2,
 \qquad
 c_\eta:=\frac1{2\eta}-\frac{L_{\rm ms}}2
 -\frac{(1-p)\eta L_{\rm ms}^2}{2p}.
\end{equation*}
The condition
\begin{equation}
 \eta\le
 \frac1{L_{\rm ms}\bigl(1+\sqrt{(1-p)/p}\bigr)}
 \label{eq:page-stability-condition}
\end{equation}
ensures $c_\eta\ge0$.  Indeed, with
$s=\sqrt{(1-p)/p}$ and $t=\eta L_{\rm ms}\le(1+s)^{-1}$,
\begin{equation*}
 2\eta c_\eta=1-t-s^2t^2
 \ge\frac{s}{(1+s)^2}\ge0.
\end{equation*}
Dropping its nonpositive contribution therefore gives
\begin{equation}
 \E[\mathcal L_{t+1}]
 \le\E[\mathcal L_t]-\frac\eta2\E\|\nabla F(x_t)\|^2.
\label{eq:page-lyapunov-step}
\end{equation}

\emph{Epoch contraction.}
The parameters in~\eqref{eq:page-epoch-parameters} satisfy
$(1-p)/p=n$, so the selected stepsize is admissible.  Summing
\eqref{eq:page-lyapunov-step} for $t=0,\ldots,T-1$ gives
\begin{equation*}
 \frac\eta2\sum_{t=0}^{T-1}\E\|\nabla F(x_t)\|^2
 \le \mathcal L_0-\E[\mathcal L_T]
 \le F(x)-F^*,
\end{equation*}
where the last inequality uses $g_0=\nabla F(x)$ and $\mathcal L_T\ge0$.  Since
$\widehat x$ is uniform on $x_0,\ldots,x_{T-1}$, dividing by $\eta T/2$
proves~\eqref{eq:page-epoch-gradient}.  The pointwise PL inequality gives
\begin{equation}
 F(\widehat x)-F^*
 \le\frac1{2\mu}\|\nabla F(\widehat x)\|^2,
\end{equation}
which proves~\eqref{eq:page-epoch-gap} after substituting
$\eta^{-1}=L_{\rm ms}(1+\sqrt n)$.

\emph{IFO cost.}
The initial full gradient costs $n$ IFO calls.  The algorithm computes
$T-1$ estimator updates, omitting the auxiliary $g_T$.  A refresh update costs
$n$ calls, while a difference update for one component can be implemented by two
calls, one at each point.  The expected cost per update is therefore
\begin{equation}
 pn+2(1-p)\le3,
\end{equation}
Hence the expected epoch cost is at most $n+3(T-1)\le n+3T$, proving the
stated bound.
\end{proof}

\subsection{Restarts tuned to the condition number}

The epoch rule below specializes to $T=\Theta(n)$ for small
$\kappa_{\rm ms}$ and to the standard PAGE schedule with a fixed contraction factor for large
$\kappa_{\rm ms}$
\citep[Corollary~6]{li2021page}.

\begin{proposition}[Restarted PAGE]
\label{prop:restarted-page}
Let $F$ be globally $\mu$-PL and have mean-squared smoothness constant
$L_{\rm ms}$ as in~\eqref{eq:page-ms-premise}, let
$\kappa_{\rm ms}=L_{\rm ms}/\mu\ge1$, and suppose
$F(x^{(0)})-F^*\le\Delta$.  Algorithm~\ref{alg:restarted-page} returns a
point with expected gap at most $\varepsilon$ using
\begin{equation}
 \begin{cases}
 \displaystyle
 O\!\left(n+\frac{n\log(\Delta/\varepsilon)}
 {1+\log(\sqrt n/\kappa_{\rm ms})}\right),
 &1\le\kappa_{\rm ms}\le\sqrt n,\\[2mm]
 \displaystyle
 O\!\left(n+\kappa_{\rm ms}\sqrt n
 \log\frac{\Delta}{\varepsilon}\right),
 &\kappa_{\rm ms}\ge\sqrt n
 \end{cases}
 \label{eq:restarted-page-rate}
\end{equation}
expected IFO calls, provided
$0<\varepsilon\le c\Delta$ for a universal $c<1$.
\end{proposition}

\begin{proof}
Use the same epoch choice for every condition number.  Set
\begin{equation}
 A_{\rm ep}:=\kappa_{\rm ms}(1+\sqrt n),
 \qquad
 T:=\left\lceil4(n+A_{\rm ep})\right\rceil,
 \qquad
 q_{\rm ep}:=\frac{A_{\rm ep}}{T}\le\frac14.
 \label{eq:page-unified-epoch}
\end{equation}

Repeated conditioning and the tower property give
$\E[F(x^{(S_{\rm ep})})-F^*]\le
q_{\rm ep}^{S_{\rm ep}}\Delta$.  Hence Algorithm~\ref{alg:restarted-page}
uses
$S_{\rm ep}=\lceil\log(\Delta/\varepsilon)/\log(1/q_{\rm ep})\rceil$
epochs, and
Lemma~\ref{lem:page-epoch-contraction} gives the explicit bound on expected cost
\begin{equation}
 \E[\#\text{ IFO calls}]
 \le(n+3T)
 \left\lceil
 \frac{\log(\Delta/\varepsilon)}{\log(1/q_{\rm ep})}
 \right\rceil.
 \label{eq:page-unified-exact-cost}
\end{equation}

The ceiling accounts for the case in which one epoch already achieves the
target accuracy.

\emph{Small condition numbers.}
Suppose $1\le\kappa_{\rm ms}\le\sqrt n$.  Then
$A_{\rm ep}=\Theta(\kappa_{\rm ms}\sqrt n)$, and the chosen epoch length satisfies
\begin{equation*}
 T=\Theta(n),\qquad
 q_{\rm ep}=\Theta\!\left(\frac{\kappa_{\rm ms}}{\sqrt n}\right),
 \qquad n+3T=\Theta(n).
\end{equation*}
Thus
$\log(1/q_{\rm ep})=
\Theta(1+\log(\sqrt n/\kappa_{\rm ms}))$, and
\eqref{eq:page-unified-exact-cost} gives
\begin{equation}
 O\!\left(
 n+\frac{n\log(\Delta/\varepsilon)}
 {1+\log(\sqrt n/\kappa_{\rm ms})}
 \right).
 \label{eq:page-small-regime-cost}
\end{equation}
The schedule with a fixed contraction factor underlying
\citet[Corollary~6]{li2021page} sets
$T_{\rm ff}=\Theta(A_{\rm ep})$ and keeps only $q_T=\Theta(1)$.  In the
small-$\kappa_{\rm ms}$ regime, both epoch choices cost $\Theta(n)$ because
each begins with a full gradient.  Extending the epoch to $T=\Theta(n)$ leaves
its order cost unchanged but strengthens the contraction to
$\Theta(\kappa_{\rm ms}/\sqrt n)$.  Thus this schedule gives
$O(n\log(\Delta/\varepsilon))$.  Using the sharper $T=\Theta(n)$ contraction produces the
denominator in~\eqref{eq:page-small-regime-cost}; the estimator itself is
unchanged.

\emph{Large $\kappa_{\rm ms}$.}
If $\kappa_{\rm ms}\ge\sqrt n$, then $A_{\rm ep}=\Omega(n)$,
$T=\Theta(A_{\rm ep})$, and $q_{\rm ep}=\Theta(1)$.  Equation
\eqref{eq:page-unified-exact-cost} gives
$O(n+\kappa_{\rm ms}\sqrt n\log(\Delta/\varepsilon))$.  Hence the schedule
reduces to the usual choice with a fixed contraction factor in this regime.  These two cases prove
\eqref{eq:restarted-page-rate}; the restriction $\varepsilon\le c\Delta$
absorbs the additive cost of one epoch in the large-$\kappa_{\rm ms}$ case.
\end{proof}

\section{Randomization and Expected Complexity}
\label{app:stopping}

This appendix turns bounds for deterministic algorithms averaged over a hard
distribution into bounds for randomized algorithms.  It also relates fixed
budgets to expected stopping times.
The two transfers are, respectively,
\begin{equation*}
 \left[
 \forall A_{\rm det}:\ Pr_{F\sim\Pi}(\mathsf{Fail}(A_{\rm det},F))\ge p
 \right]
 \Longrightarrow
 \left[
 \forall A_\xi\ \exists F\in\operatorname{supp}\Pi:\
 \Pr_\xi(\mathsf{Fail}(A_\xi,F))\ge p
 \right],
\end{equation*}
and
\begin{equation*}
 \left[
 \E\tau\le Q,\quad
 \E[\operatorname{err}_F(\widehat X)]\le\varepsilon
 \right]
 \Longrightarrow
 \left[
 \Pr\!\left(\operatorname{err}_F(\widehat X^{(N)})\le t\right)
 \ge1-\frac{Q}{N}-\frac{\varepsilon}{t}
 \right].
\end{equation*}

\subsection{From deterministic to randomized algorithms}
\label{sec:deterministic-to-randomized}

For each parameter choice, let $\Pi$ be the distribution over fixed instances
sampled before interaction.  Consider any deterministic IFO algorithm that
makes $N$ calls, chooses its queries from the preceding history, and returns an
output measurable with respect to its full IFO history.  For every such
algorithm, the bounds established in
Appendices~\ref{app:information}--\ref{app:pl} have the form
\[
 \Pr_{F\sim\Pi}(\mathsf{Fail}(A,F))\ge p.
\]
Output coverage follows from Lemma~\ref{lem:future-alignment} and
Lemma~\ref{lem:pathwise-coupling}
for staged chains and Lemma~\ref{lem:dense-quadratic-row-bound} for the
auxiliary quadratic.

For a randomized algorithm, condition on its private seed $\xi\perp F$ and
average the deterministic bound:
\begin{equation}
 \E_{F\sim\Pi}\Pr_\xi(\mathsf{Fail}(A_\xi,F)\mid F)\ge p.
\end{equation}
Thus some fixed $F\in\operatorname{supp}\Pi$ has randomized failure probability
at least $p$.  Every supported instance satisfies
$F(x^{(0)})-F^*\le\Delta$ and has individually $L_{\max}$-smooth components;
each PL instance additionally satisfies
$\|\nabla F(x)\|^2\ge2\mu(F(x)-F^*)$ for every $x$.

\subsection{A generic capping lemma}

\begin{lemma}[Capping a stopped algorithm]
\label{lem:page-fixed-budget}
Fix an instance $F$ and a nonnegative measurable error $\operatorname{err}_F$.
Let $\tau$ be a stopping time for the algorithm's IFO filtration
$(\mathcal F_t)_{t\ge0}$ and let $\widehat X$ be $\mathcal F_\tau$-measurable.
Let $\widehat X^{(N)}$ be its output when capped after $N\ge1$ IFO calls, using
any fixed fallback if $\tau>N$.  For every $t>0$,
\begin{equation}
 \Pr\!\left(\operatorname{err}_F(\widehat X^{(N)})\le t\right)
 \ge
 1-\frac{\E\tau}{N}
   -\frac{\E[\operatorname{err}_F(\widehat X)]}{t}.
 \label{eq:generic-capping-bound}
\end{equation}
\end{lemma}

\begin{proof}
On $\{\tau\le N\}\cap\{\operatorname{err}_F(\widehat X)\le t\}$, the capped
output equals $\widehat X$.  Apply Markov's inequality to the two complement
events and take their union; no independence is used.
\end{proof}

\subsection{Applications to upper and lower bounds}

\paragraph{Upper bound.}
Fix $\delta\in(0,1)$ and apply Proposition~\ref{prop:restarted-page} at target
$\delta\varepsilon/2$.  If its expected call count is $Q$, use
\[
 N=\max\left\{1,\left\lceil\frac{2Q}{\delta}\right\rceil\right\},
 \qquad
 \operatorname{err}_F(x)=F(x)-F^*,
 \qquad t=\varepsilon.
\]
Lemma~\ref{lem:page-fixed-budget} gives success probability at least
$1-\delta$.  For fixed
$\delta$, this changes $\log(\Delta/\varepsilon)$ only by an additive constant
and proves the upper bound with a fixed budget in
Theorem~\ref{thm:pl-main}, together with its extension to
$1\le\kappa_{\rm ms}<3$.

\paragraph{Lower bounds.}
Suppose that on every fixed instance an algorithm satisfies $\E\tau\le Q$ and
$\E[\operatorname{err}_F(\widehat X)]\le\varepsilon$, where
\[
 \operatorname{err}_F(x)=
 \begin{cases}
  \|\nabla F(x)\|,&\text{nonconvex},\\
  F(x)-F^*,&\text{PL}.
 \end{cases}
\]
With $N=\max\{1,\lceil4Q\rceil\}$ and $t=4\varepsilon$,
Lemma~\ref{lem:page-fixed-budget} gives, for every fixed $F$,
\begin{equation}
 \Pr\!\left(\operatorname{err}_F(\widehat X^{(N)})\le4\varepsilon\right)
 \ge\frac12.
 \label{eq:capped-expected-success}
\end{equation}

Apply the randomized lower bound with a fixed budget above to this capped algorithm.
In each case with multiple stages,
$\{\operatorname{err}_F\le4\varepsilon\}\subset
\{\operatorname{err}_F<8\varepsilon\}$ and the latter event has probability at
most $33/128$; each auxiliary quadratic bounds the former by $5/16$.  Both
contradict~\eqref{eq:capped-expected-success}.

Absorbing the constant factors from truncation and rescaling the target
gives the expected nonconvex lower bound
\begin{equation}
 \Omega\!\left(n+\sqrt n\,
 \frac{\Delta L_{\max}}{\varepsilon^2}\right),
\end{equation}
and the expected PL lower bound
\begin{equation}
 \begin{cases}
 \displaystyle
 \Omega\!\left(n+\frac{n\log(\Delta/\varepsilon)}
 {1+\log(\sqrt n/\kappa_{\max})}\right),
 &3\le\kappa_{\max}\le\sqrt n,\\[2mm]
 \displaystyle
 \Omega\!\left(n+\kappa_{\max}\sqrt n
 \log\frac{\Delta}{\varepsilon}\right),
 &\kappa_{\max}\ge\sqrt n.
 \end{cases}
\end{equation}
These are the conclusions of Theorems~\ref{thm:nonconvex-main}
and~\ref{thm:pl-main}; the constant target rescaling is absorbed within their
stated accuracy regimes.

\clearpage
\section{Oracle Model and Smoothness Comparisons}
\label{app:smoothness-calibration}

This appendix states the formal IFO model, gives the parameter ranges for the
cited PL results, and compares individual, mean-squared, and objective
smoothness at the same numerical scale.  It also verifies that the nonconvex
dense chain has comparable individual and mean-squared smoothness constants.
Figure~\ref{fig:class-containment} displays the two containment relations used
throughout the comparison.

\subsection{Formal IFO model and theorem conventions}
\label{app:formal-ifo-model}

This subsection gives the formal details omitted from
Section~\ref{sec:setting}.  The function classes used there over arbitrary
finite dimensions are
\begin{equation}
 \mathcal F_{\rm ind}^n(\Delta,L_{\max})
 :=\bigcup_{d\ge1}\mathcal F_{\rm ind}^{n,d}(\Delta,L_{\max}),
 \qquad
 \mathcal F_{\rm ms}^n(\Delta,L_{\rm ms})
 :=\bigcup_{d\ge1}\mathcal F_{\rm ms}^{n,d}(\Delta,L_{\rm ms}),
 \label{eq:dimension-free-class}
\end{equation}
and, for $\mathsf s\in\{\mathrm{ind},\mathrm{ms}\}$,
\begin{equation}
 \mathcal F_{\mathsf s,\rm PL}^{n}(\Delta,L,\mu)
 :=\bigcup_{d\ge1}\mathcal F_{\mathsf s,\rm PL}^{n,d}
 (\Delta,L,\mu).
 \label{eq:pl-function-class}
\end{equation}

Fix a dimension $d$.  Let $\xi$ be the algorithm's private random seed and let
$Y_t=\mathcal O_{I_t}(X_t)$.  Define the filtration
\begin{equation}
 \mathcal F_0:=\sigma(\xi),
 \qquad
 \mathcal F_t:=\sigma\!\left(\xi,(I_s,X_s,Y_s)_{1\le s\le t}\right),
 \quad t\ge1.
 \label{eq:ifo-filtration}
\end{equation}
A randomized IFO algorithm chooses $(I_t,X_t)$ as an
$\mathcal F_{t-1}$-measurable random variable, stops at a stopping time $\tau$
with respect to $(\mathcal F_t)_{t\ge0}$, and returns an
$\mathcal F_\tau$-measurable output $\widehat X$.  All spaces are standard Borel.
The resulting class is denoted by
\begin{equation}
 \mathcal A_{\rm IFO}^{\rm rand}(d).
 \label{eq:rand-ifo-class}
\end{equation}
It permits repeated indices, arbitrary query points, randomized stopping, and
unqueried outputs.

For a dimension-free function class, an algorithm is a family indexed by dimension:
\begin{equation}
 \mathcal A_{\rm IFO}^{\rm rand}
 :=\left\{A=(A_d)_{d\ge1}:A_d\in
 \mathcal A_{\rm IFO}^{\rm rand}(d)\ \text{for every }d\right\}.
 \label{eq:dimension-indexed-algorithms}
\end{equation}
Let $\mathcal A_{\rm IFO}^{\rm span}(d)$ and
$\mathcal A_{\rm IFO}^{\rm zero}(d)$ denote the corresponding linear-span and
zero-respecting subclasses, following the usual conventions for lower bounds
\citep{zhou2019lower,carmon2021lower2}.  They restrict queries to the previous
linear span and revealed coordinate support, respectively.  With $X_0=0$ and
$d\ge2$,
\begin{equation}
 \mathcal A_{\rm IFO}^{\rm span}(d)
 \subseteq\mathcal A_{\rm IFO}^{\rm zero}(d)
 \subseteq\mathcal A_{\rm IFO}^{\rm rand}(d).
 \label{eq:algorithm-containment}
\end{equation}

These inclusions compare algorithms using the same exact value-and-gradient
IFO.  \citet[Definition~11]{han2024lower} study a span-restricted PIFO class
with independent component samples from a fixed distribution and infrequent
full-gradient calls; their PIFO also provides component proximal information.
Our unrestricted-query guarantee concerns the exact value-and-gradient IFO,
so the comparison separates query selection from oracle information.

For $\mathcal F=\bigcup_{d\ge1}\mathcal F^d$ and an error functional
$\mathcal E_F$, the complexity for a fixed budget used in the main text is
\begin{equation}
 \operatorname{Comp}_{\varepsilon,\delta}^{\rm IFO}
 (\mathcal F;\mathcal E)
 :=\inf\left\{N\in\mathbb N_0:
 \begin{array}{l}
 \exists A=(A_d)_{d\ge1}\in\mathcal A_{\rm IFO}^{\rm rand},\\[-1pt]
 \text{$A_d$ makes at most $N$ calls for every $d$, and}\\[-1pt]
 \displaystyle\sup_{d\ge1}\sup_{F\in\mathcal F^d}
 \Pr_{\xi}\!\left(\mathcal E_F(\widehat X_{A_d})>\varepsilon\right)
 \le\delta
 \end{array}\right\}.
 \label{eq:ifo-complexity}
\end{equation}
Once $d$ and $F$ are fixed, $\Pr_\xi$ refers only to the private seed.  In the
proofs of the lower bounds, $\Prob$ includes the auxiliary random hard instance as
well.  Calligraphic $\mathcal E$ denotes an error criterion, whereas $\E$
denotes expectation.

\paragraph{Nonconvex theorem parameters.}
The universal threshold and failure probability in
Theorem~\ref{thm:nonconvex-main} may be taken as
$n_0=2^{10}=1024$ and $p_0=11/16$, with universal $c_0>0$ and a universal
implicit constant in~\eqref{eq:nonconvex-main-rate}.
Appendix~\ref{app:nonconvex} proves the corresponding
claim for a fixed budget and may use
$D=\lceil c_D[1+\log((N+2)(T_{\rm chain}+1))]\rceil$ coordinates per stage,
where $c_D$ is universal.  For $n<n_0$, replicating one hard function across all
components gives the same asymptotic order after changing
universal constants
\citep{carmon2021lower2,emmenegger2022oracle}.

The PL lower bound uses the same failure probability $11/16$.  Its full
quantifiers for a fixed budget are restated at the beginning of
Appendix~\ref{app:pl}.  Appendix~\ref{app:stopping} converts both lower bounds
to expected call budgets in the worst case and converts the Restarted PAGE guarantee
to complexity with a fixed budget and constant success probability.

\paragraph{Parameter ranges for the PL bounds in
Table~\ref{tab:complexity}.}
\citet{yue2023lower} count full-gradient calls; their Theorem~6 inherits the
requirement $\kappa_F>2632476000/361$ on the condition number from Theorem~4 and
assumes $\varepsilon\le\Delta/16$.  The logarithmic Theorem~3.4 of
\citet{bai2024complexity} requires
$\kappa_{\rm ms}\ge729196\sqrt n$ and $\varepsilon<\Delta/200$, whereas their
Theorem~3.5, valid for all $\kappa_{\rm ms}\ge1$, gives $\Omega(n)$.
Their Corollary~3.6 is stated for all $L\ge\mu$, but its proof in Appendix~A.5
combines Theorems~3.4 and~3.5.  When the condition number is below the threshold
of Theorem~3.4, this combination does not establish the logarithmic term.
Our comparison therefore uses the explicit ranges of these two theorems.

The construction in \citet[Lemma~3.2 and Appendix~A.1]{bai2024complexity}
also has an individual-smoothness interpretation.  If the base function is
$\widehat L$-smooth and $\widehat\mu$-PL, its disjoint-block embedding can be
calibrated with the bounds
\begin{equation}
 L_{\max}=\widehat L,\qquad
 L_{\rm ms}=\frac{\widehat L}{\sqrt n},\qquad
 \mu=\frac{\widehat\mu}{n},\qquad
 \kappa_{\max}=\sqrt n\,\kappa_{\rm ms}.
\end{equation}
With these admissible smoothness bounds, Theorem~3.4 gives
\begin{equation}
 \Omega\!\left(\kappa_{\rm ms}\sqrt n\log\frac{\Delta}{\varepsilon}\right)
 =\Omega\!\left(\kappa_{\max}\log\frac{\Delta}{\varepsilon}\right)
 \quad\text{for }\kappa_{\max}\ge729196n,\quad
 \varepsilon<\Delta/200.
\end{equation}
Thus this earlier PL construction already gives dependence on the condition
number under individual smoothness, but the corresponding term lacks the
$\sqrt n$ factor obtained here.

The new upper entries in Table~\ref{tab:complexity} use
$\kappa_{\rm ms}$, and the new lower entries use $\kappa_{\max}$.  The latter
hard instances are individually smooth and hence also lie in the mean-squared
class.

\subsection{Assumptions on the function classes}
\label{app:function-classes}

Table~\ref{tab:function-classes} collects the three regularity assumptions used
in the main text and representative analyses for each.  Proposition~\ref{prop:ms-separators}
then gives counterexamples to three inequalities that require separate
control of each component
when the same numerical smoothness constant is substituted under mean-squared
smoothness.

\begin{figure}[H]
\centering
\includegraphics[width=\linewidth]{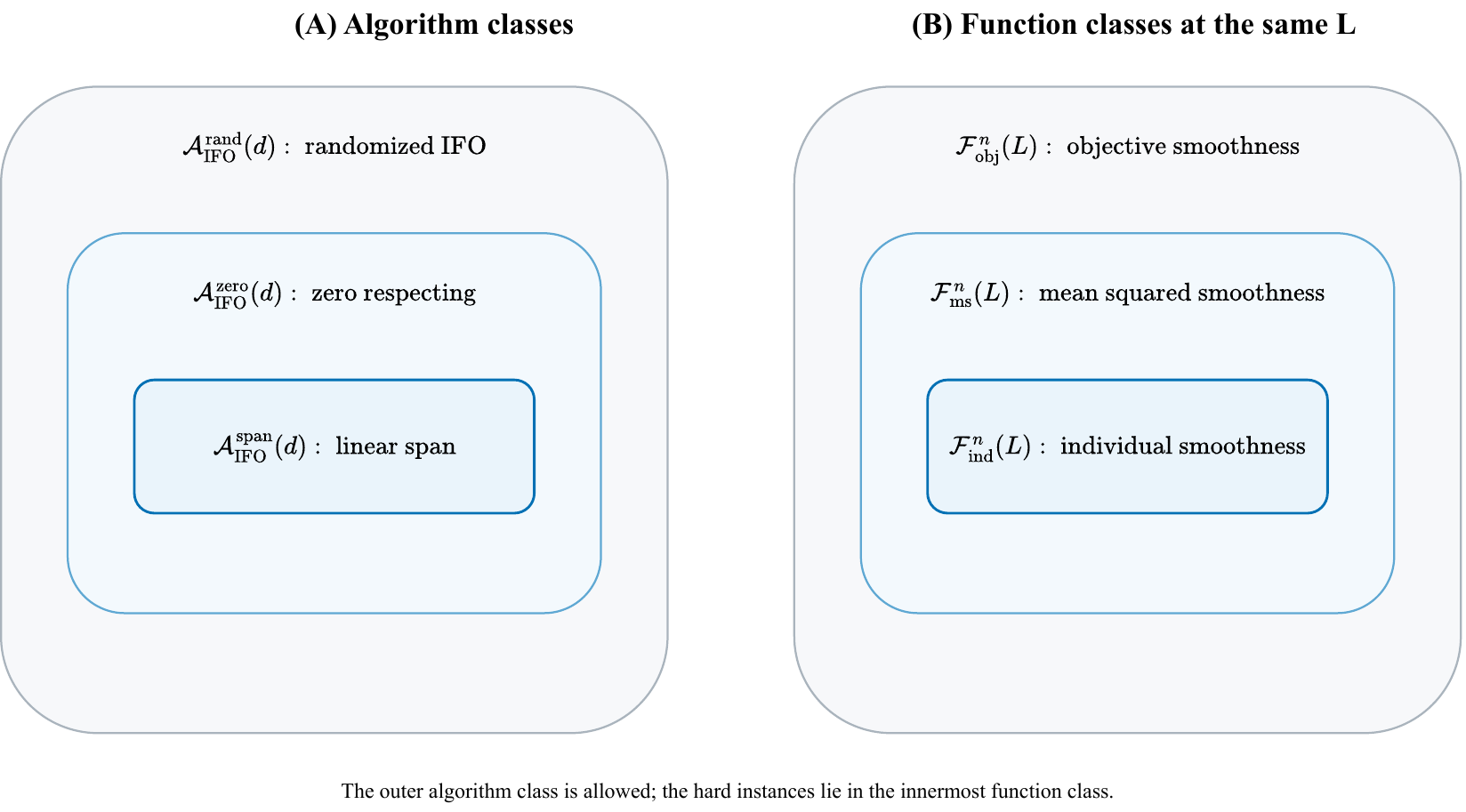}
\caption{Containment relations for the algorithm and function classes.  The
lower bounds allow algorithms from the outermost class in panel (A), while the
hard instances belong to the innermost class in panel (B).  For the function
classes, the inclusions use the same numerical smoothness bound $L$.}
\label{fig:class-containment}
\end{figure}

\begin{table}[H]
\caption{Definitions and representative analyses for the function classes in
Figure~\ref{fig:class-containment}.}
\label{tab:function-classes}
\centering
\footnotesize
\setlength{\tabcolsep}{3.5pt}
\renewcommand{\arraystretch}{1.1}
\begin{tabularx}{\textwidth}{@{}>{\raggedright\arraybackslash}p{0.18\textwidth}>{\raggedright\arraybackslash}p{0.43\textwidth}>{\raggedright\arraybackslash}X@{}}
\toprule
\textbf{Class} & \textbf{Defining regularity} & \textbf{Representative sources} \\
\midrule
$\mathcal F_{\rm ind}^n$
& $\max_i\|\nabla f_i(x)-\nabla f_i(y)\|\le L_{\max}\|x-y\|$
& Finite-sum lower bounds and analyses of component gradient updates or stored gradients
  \citep{zhou2019lower,han2024lower,emmenegger2022oracle,defazio2014saga,kavis2022adaspider}; this work. \\
$\mathcal F_{\rm ms}^n$
& $\bigl(n^{-1}\sum_i\|\nabla f_i(x)-\nabla f_i(y)\|^2\bigr)^{1/2}\le L_{\rm ms}\|x-y\|$
& PAGE and SPIDER \citep{fang2018spider,li2021page}; \citet[Definition~3.2]{zhou2019lower}; PL lower bounds in \citet{bai2024complexity}. \\
$\mathcal F_{\rm obj}^n$
& $\|\nabla F(x)-\nabla F(y)\|\le L_F\|x-y\|$
& Lower bounds for stationary points of the full objective \citep{carmon2020lower1,carmon2021lower2}. \\
\bottomrule
\end{tabularx}
\parbox{\textwidth}{\vspace{2pt}\footnotesize
The first inclusion follows by averaging squared bounds; the second follows from Jensen's inequality.  A lower bound on a larger aggregate class does not automatically hold on the smaller individual class.}
\end{table}

\begin{proposition}[Mean-squared smoothness does not imply componentwise
bounds at the same constant]
\label{prop:ms-separators}
Fix $n\ge2$ and $\ell>0$.  For each of the following inequalities, there is a
convex quadratic finite sum with mean-squared smoothness constant exactly
$\ell$ for which the inequality fails:
\begin{enumerate}[label=(\roman*),leftmargin=*,itemsep=2pt,topsep=3pt]
    \item the same smoothness bound for every component
    \begin{equation}
        \|\nabla f_i(x)-\nabla f_i(y)\|
        \le \ell\|x-y\|
        \qquad\text{for every }i;
        \label{eq:false-samplewise-ms}
    \end{equation}
    \item the componentwise co-coercivity bound
    \begin{equation}
        \frac1n\sum_{i=1}^n
        \|\nabla f_i(u)-\nabla f_i(x)\|^2
        \le
        2\ell\bigl(F(u)-F(x)-\langle\nabla F(x),u-x\rangle\bigr);
        \label{eq:false-cocoercive-ms}
    \end{equation}
    \item the bound with a separate reference point for each component
    \begin{equation}
        \frac1n\sum_{i=1}^n
        \|\nabla f_i(x)-\nabla f_i(\phi_i)\|^2
        \le
        \ell^2\frac1n\sum_{i=1}^n\|x-\phi_i\|^2.
        \label{eq:false-multipair-ms}
    \end{equation}
\end{enumerate}
\end{proposition}

\begin{proof}
\emph{Counterexample to a uniform component bound.}
For~\eqref{eq:false-samplewise-ms}, consider the following one-dimensional
quadratic family with one nonzero component:
\begin{equation}
    f_1(x)=\frac{\sqrt n\,\ell}{2}x^2,
    \qquad
    f_2(x)=\cdots=f_n(x)=0.
    \label{eq:one-spike-quadratic}
\end{equation}
For every $x,y\in\mathbb R$,
\begin{equation}
    \frac1n\sum_{i=1}^n
    |\nabla f_i(x)-\nabla f_i(y)|^2
    =\ell^2|x-y|^2,
    \label{eq:one-spike-ms-exact}
\end{equation}
so the mean-squared smoothness constant is exactly $\ell$, whereas the first component has smoothness constant $\sqrt n\,\ell$ and violates~\eqref{eq:false-samplewise-ms} whenever $x\ne y$.

Thus the gradient of one component can have a Lipschitz constant $\sqrt n$
times larger than the mean-squared constant.  A proof step can therefore fail if it
substitutes the mean-squared constant into a descent or self-bounding inequality
for the sampled component.

\emph{Counterexample to the componentwise co-coercivity bound.}
The same family also disproves~\eqref{eq:false-cocoercive-ms}.  At $x=0$ and $u=1$,
\begin{equation}
    \frac1n\sum_{i=1}^n
    |\nabla f_i(1)-\nabla f_i(0)|^2=\ell^2,
    \qquad
    F(1)-F(0)-\langle\nabla F(0),1\rangle
    =\frac{\ell}{2\sqrt n}.
    \label{eq:one-spike-cocoercive-gap}
\end{equation}
Thus the right-hand side of~\eqref{eq:false-cocoercive-ms} is only $\ell^2/\sqrt n$.  More generally, the coefficient multiplying the term of Bregman type on the right-hand side must be at least $2\sqrt n\,\ell$, or equivalently the smoothness parameter in the displayed $2C$ form must satisfy $C\ge\sqrt n\,\ell$.

\emph{Counterexample with separate reference points for the components.}
To disprove~\eqref{eq:false-multipair-ms}, set the ambient dimension to $d=n$
and define
\begin{equation}
    f_i(z)=\frac{\sqrt n\,\ell}{2}z_i^2,
    \qquad i\in[n].
    \label{eq:coordinate-spike-quadratic}
\end{equation}
For all $z,w\in\mathbb R^n$,
\begin{equation}
    \frac1n\sum_{i=1}^n
    \|\nabla f_i(z)-\nabla f_i(w)\|^2
    =\ell^2\|z-w\|^2,
    \label{eq:coordinate-spike-ms-exact}
\end{equation}
again giving exact mean-squared smoothness constant $\ell$.  Take $x=0$ and $\phi_i=e_i$.  Then the left- and right-hand sides of~\eqref{eq:false-multipair-ms} are, respectively,
\begin{equation}
    \frac1n\sum_{i=1}^n
    \|\nabla f_i(0)-\nabla f_i(e_i)\|^2
    =n\ell^2,
    \qquad
    \ell^2\frac1n\sum_{i=1}^n\|e_i\|^2=\ell^2.
\end{equation}
Mean-squared smoothness controls all components at one common pair $(z,w)$.  It
does not imply the displayed estimate when component $i$ is paired with its own
reference point $\phi_i$.
\end{proof}

The three counterexamples concern, respectively, descent and self-bounding
arguments applied to a sampled component
\citep{vaswani2019painless,loizou2021stochastic,mishchenko2020random,kavis2022adaspider},
variance-to-Bregman bounds \citep{liu2022adaptive}, and bounds using
a separate reference point for each component \citep[Lemma~2]{defazio2014saga}.  These inequalities therefore do not
follow from mean-squared smoothness alone at the same numerical constant.

\subsection{Mean-squared smoothness of the nonconvex dense chain}

The following lemma shows that the nonconvex dense chain has mean-squared
smoothness comparable to its individual smoothness.  The PL lower bounds do
not require this comparison of least constants: individual smoothness with
bound $L$ already implies mean-squared smoothness with the same bound $L$.

\begin{lemma}[The nonconvex dense chain has comparable mean-squared smoothness]
\label{lem:dense-chain-ms}
Suppose $T\ge2$ and scale the components of the dense chain by
$f_i(y)=\alpha R_i(\beta y)$, where
\begin{equation}
 \alpha\beta^2=\frac{L_{\max}\rho}{H_{\rm main}}
 \label{eq:dense-chain-local-scaling-identity}
\end{equation}
as in~\eqref{eq:nonconvex-scales}.  Here $a=5/32$ and $b=9/32$ are the lower
and upper gate thresholds, and
$H_{\rm main}=341$, defined in~\eqref{eq:explicit-component-H}, is the
unscaled component Hessian bound.
If
$L_{\rm ms}$ denotes their least mean-squared smoothness constant, then
\begin{equation}
 \frac{8(1-b^2)^{3/2}}{H_{\rm main}}L_{\max}
 \le L_{\rm ms}\le L_{\max}.
 \label{eq:dense-chain-ms-comparison}
\end{equation}
In particular, $L_{\rm ms}=\Theta(L_{\max})$ with universal constants,
independently of $n,D,$ and $T$.
\end{lemma}

\begin{proof}
\emph{Witness point.}
Since $G(b)-G(a)=1$ and $b-a=1/8$, the mean value theorem gives an
$s_0\in(a,b)$ with $G'(s_0)=8$.  Write $u=u_{\Theta^{(1)}}$ and choose
$z_1=ru$, where $r=s_0/\sqrt{1-s_0^2}$, so that $q_1(z_1)=s_0$ and
\begin{equation}
 \nabla q_1(z_1)=(1-s_0^2)^{3/2}u.
\end{equation}

\emph{Unscaled Hessian block.}
Set $z_2=\cdots=z_T=0$ and let $v$ be the unit vector supported on the
first block with $v_1=u$.  In the second output block, the only Hessian term
applied to $v$ is the cross derivative of the second chain link.  Hence, for
every component $i$,
\begin{equation}
 \bigl\|(\nabla^2R_i(z)v)_2\bigr\|
 =G'(s_0)(1-s_0^2)^{3/2}
   \|\nabla\widetilde q_{i2}(0)\|
 \ge\frac{8(1-b^2)^{3/2}}\rho,
\end{equation}
because $\nabla\widetilde q_{i2}(0)=s_{i2}/\rho$ and $\|s_{i2}\|=1$.

\emph{Difference quotient argument.}
After scaling, $\nabla^2f_i(z/\beta)=\alpha\beta^2\nabla^2R_i(z)$ and
$\alpha\beta^2=L_{\max}\rho/H_{\rm main}$.  Mean-squared smoothness implies, for every
unit vector $v$ and every $h\ne0$,
\begin{equation}
 \left(\frac1n\sum_{i=1}^n
 \left\|\frac{\nabla f_i(x+hv)-\nabla f_i(x)}{h}\right\|^2\right)^{1/2}
 \le L_{\rm ms}.
\end{equation}
Letting $h\to0$ at any fixed $x$ gives
\begin{equation}
 L_{\rm ms}\ge
 \left(\frac1n\sum_{i=1}^n\|\nabla^2f_i(x)v\|^2\right)^{1/2}.
 \label{eq:ms-hessian-directional-lower}
\end{equation}

\emph{Final scaling and substitution.}
At $x=z/\beta$ and for the $v$ chosen above, the second block gives
$\|\nabla^2f_i(x)v\|\ge
8\alpha\beta^2(1-b^2)^{3/2}/\rho$ for every $i$.  Substituting
$\alpha\beta^2=L_{\max}\rho/H_{\rm main}$
in~\eqref{eq:ms-hessian-directional-lower} proves the lower bound
in~\eqref{eq:dense-chain-ms-comparison}; the upper bound follows from
individual $L_{\max}$-smoothness.
\end{proof}

\end{document}